\documentclass[twocolumn]{article}
\usepackage{mglgrid}

\usepackage{scalerel}[2016/12/29]
\usepackage{stackengine}
\usepackage[normalem]{ulem}
\usepackage{siunitx}
\usepackage{mathtools}
\usepackage{nicefrac}
\usepackage{centernot}
\usepackage{mleftright}

\usepackage{amsmath}
\usepackage{amsfonts}
\usepackage{amsthm}
\usepackage{thmtools}
\usepackage{thm-restate}

\usepackage{stmaryrd}
\SetSymbolFont{stmry}{bold}{U}{stmry}{m}{n}

\usepackage{bbding}

\usepackage{xspace}

\usepackage{bm}

\usepackage{color}
\usepackage[table]{xcolor}
\definecolor{mglred}{RGB}{165,30,30}
\definecolor{mglblue}{RGB}{59,59,200}
\definecolor{mglcyan}{RGB}{165,30,165}
\newcommand{\mglcolorbox}[2]{\setlength{\fboxsep}{0pt}\colorbox{#1}{\strut #2}}

\newlist{myenum}{enumerate}{2}
\setlist[myenum,1]{label=\arabic*., ref=\arabic*}
\setlist[myenum,2]{label=\alph*., ref=\themyenumi.\alph*}

\usepackage{booktabs}
\usepackage{multirow}
\usepackage{algorithm}

\usepackage{array}
\newcolumntype{L}[1]{>{\raggedright\let\newline\\\arraybackslash
    \hspace{0pt}}m{#1}}
\newcolumntype{C}[1]{>{\centering\let\newline\\\arraybackslash
    \hspace{0pt}}m{#1}}
\newcolumntype{R}[1]{>{\raggedleft\let\newline\\\arraybackslash
    \hspace{0pt}}m{#1}}

\usepackage{caption}
\usepackage{tkz-graph}
\usepackage{tikz}
\usepackage{pgfplots,tikz,pgfplotstable}
\pgfplotsset{compat=1.18}
\usetikzlibrary{positioning, fit, arrows.meta, shapes}

\pgfplotsset{
    legend image with text/.style={
        legend image code/.code={%
            \node[anchor=center] at (0.3cm,0cm) {#1};
        }
    },
}
\pgfplotsset{uniformhash/.style = {black!50!white, loosely dotted, very thick}}
\pgfplotsset{bdp10/.style = {teal!70!black, semithick}}
\pgfplotsset{bdp5/.style = {teal!70!black, densely dashed, thick}}
\pgfplotsset{bdp1/.style = {teal!70!black, loosely dotted, very thick}}
\pgfplotsset{asymmetric/.style = {blue!70!black, densely dashed, thick}}
\pgfplotsset{martingale/.style = {orange!70!white, densely dashdotted, thick}}
\pgfplotsset{adp10/.style = {red!70!black, semithick}}
\pgfplotsset{adp5/.style = {red!70!black, densely dashed, thick}}
\pgfplotsset{adp1/.style = {red!70!black, loosely dotted, very thick}}

\usetikzlibrary{bayesnet}

\renewcommand{\edge}[3][]{ %
  \foreach \x in {#2} { %
    \foreach \y in {#3} { %
      \path (\x) edge [->, >={triangle 45}, #1] (\y) ;%
    }
  }
}

\allowdisplaybreaks
\let\leq\leqslant
\let\geq\geqslant
\newcommand\scalemathbox[2]{\text{\scalebox{#1}{$#2$}}}
\newcommand\raisemathbox[2]{\text{\raisebox{#1}{$#2$}}}
\newcommand{\prerendermath}[3]{%
  \newsavebox{#1}%
  \savebox{#1}{$#2$}%
  \newcommand{#3}{\usebox{#1}}%
}

\makeatletter
\let\old@font@info\@font@info
\def\@font@info#1{%
\expandafter\ifx\csname\detokenize{#1}\endcsname\relax
  \old@font@info{#1}%
\fi
\expandafter\xdef\csname\detokenize{#1}\endcsname{}%
}
\makeatother

\newcommand\lsh[1]{\!#1\,\xspace}
\newcommand\lshd[1]{\!\!#1\,\xspace}
\newcommand\lshh[1]{\halfnegkern#1\kern0.08333em\xspace}
\newcommand\halfnegkern{\kern-0.08333em\xspace}
\newcommand\quartnegkern{\kern-0.04666em\xspace}

\newcommand\fixmathkern[2]{%
  \sbox0{$#1#2$}\sbox2{\emph{#1#2}}%
  #1\kern\dimexpr\wd2-\wd0\relax
  #2}
\newcommand\mathendkern[2]{%
  \sbox0{$#1$#2}\sbox2{\emph{#1#2}}%
  \kern\dimexpr\wd2-\wd0\relax
  #2}
\newcommand\cC{c\quartnegkern,}

\newcommand\fC{f\kern-0.21em,}
\newcommand\fP{f\kern-0.21em.}
\newcommand\gC{g\quartnegkern,}
\newcommand\UC{U\kern-0.15em,}
\newcommand\sFC{\sF\!,}
\newcommand\sFP{\sF\!.}

\newcommand\supkern{\!\xspace}
\newcommand\noscriptspace{\kern-\scriptspace}

\mathtoolsset{%
  above-shortintertext-sep = 0.1\groundskip,
  below-shortintertext-sep = 0.1\groundskip,
}
\newcommand\explainstep[1]{%
  \shortintertext{
    \hspace{\groundskip}%
    \small\smash{$\downarrow$} #1}}

\DeclareMathOperator*{\image}{im}

\DeclareMathOperator*{\argmin}{arg\,min}

\newcommand\newtxmathabsfix{\kern0.025em}
\DeclarePairedDelimiter{\absX}{\lvert}{\rvert}
\newcommand\abs[2][]{\absX[#1]{\if\relax\detokenize{#1}\relax#2\else\newtxmathabsfix#2\newtxmathabsfix\fi}}
\newcommand\newtxmathnormfix{\kern0.035em}
\DeclarePairedDelimiter{\normX}{\lVert}{\rVert}
\newcommand\norm[2][]{\normX[#1]{\if\relax\detokenize{#1}\relax#2\else\newtxmathnormfix#2\newtxmathnormfix\fi}}

\DeclareMathOperator{\E}{E}
\DeclareMathOperator{\Var}{Var}
\DeclareMathOperator{\Std}{Std}

\DeclareMathOperator{\Cov}{Cov}

\DeclarePairedDelimiterX{\infdivx}[2]{(}{)}{%
  #1\;\delimsize\|\penalty 0 \;#2%
}

\makeatletter
\newcommand{\oset}[3][0ex]{%
  \mathrel{\mathop{#3}\limits^{
    \vbox to#1{\kern-2\ex@
    \hbox{$\scriptstyle#2$}\vss}}}}
\makeatother

\newcommand\len[1]{\abs{#1}}
\newcommand\sign{\operatorname{sign}}

\newcommand\ct[1]{\mathrm{#1}}
\newcommand{\bigO}{O}

\renewcommand{\neg}{\protect\hstretch{0.8}{-}}
\newcommand\negone{{\neg\! 1}}

\newcommand\vect[1]{\bm{#1}}
\newcommand\vectrv[1]{\bm{#1}}
\newcommand{\indep}{\mathrel{\text{\scalebox{1.10}{$\perp\mkern-10mu\perp$}}}}
\newcommand{\indic}{\mathbb{1}}

\newtheorem{theorem}{Theorem}
\newtheorem{definition}[theorem]{Definition}
\newtheorem{lemma}[theorem]{Lemma}
\newtheorem{proposition}[theorem]{Proposition}
\newtheorem{corollary}[theorem]{Corollary}

\newtheorem{assumption}{Assumption}
\newtheorem{remark}[theorem]{Remark}
\newtheorem{example}[theorem]{Example}
\newcommand\inlineassumption[1]{(#1)}
\newcommand\inlineassumptionref[1]{#1}
\newcommand\numberthis{\addtocounter{equation}{1}\tag{\theequation}}
\newcommand\nobreakpar{\par\nobreak\@afterheading}
\makeatother

\newcommand\pig[1]{\scalerel*[5.5pt]{\Big#1}{%
  \ensurestackMath{\addstackgap[1.55pt]{\big#1}}}}
\newcommand\pigl[1]{\mathopen{\pig{#1}}}
\newcommand\pigr[1]{\mathclose{\pig{#1}}}

\newcommand\Pig[1]{\scalerel*[5.5pt]{\bigg#1}{%
  \ensurestackMath{\addstackgap[1.55pt]{\Big#1}}}}
\newcommand\Pigl[1]{\mathopen{\Pig{#1}}}
\newcommand\Pigr[1]{\mathclose{\Pig{#1}}}

\DeclarePairedDelimiterX{\PP}[1]{(}{)}{#1}
\DeclarePairedDelimiterX{\BB}[1]{[}{]}{#1}
\DeclarePairedDelimiterX{\set}[1]{\{}{\}}{#1}

\newcommand\given{%
  \ifdefined\delimsize%
    \nonscript\;\delimsize\vert\allowbreak\nonscript\;%
  \else%
    \mid%
  \fi%
}

\newcommand{\EB}[2][]{\E\BB[#1]{#2}}
\newcommand{\VarB}[2][]{\Var\BB[#1]{#2}}
\newcommand{\StdB}[2][]{\Std\BB[#1]{#2}}

\newcommand{\sR}{\mathbb{R}}
\newcommand{\sN}{\mathbb{N}}
\newcommand{\gammadec}{{\gamma_{\mathrm{dec}}}}
\newcommand{\gammadecprime}{{{\gamma'}_{\!\!\!\mathrm{dec}}}}
\newcommand\Alpha{\mathcal{A}}

\newcommand\Epsilon{\mathcal{E}}
\newcommand\EpsilonvecX{\bm{\mathcal{E}}}
\newcommand\Epsilonsub{\mathcal{E}}
\newcommand\Epsilonvecsub{\Epsilonvec\hspace{-0.06em}}
\newcommand\Tau{\raisemathbox{0.02em}{\mathcal{T}}}
\newcommand\Tausub{\Tau\!}
\newcommand\TauvecX{\raisemathbox{0.01em}{\bm{\mathcal{T}}}}

\newcommand{\sS}{\mathcal{S}}
\newcommand{\sC}{\mathcal{C}}
\newcommand{\sD}{\mathcal{D}}
\newcommand{\sU}{\mathcal{U}}
\newcommand{\sF}{\mathcal{F}}
\newcommand{\LE}{\hat{\kern0pt L}}      
\newcommand{\TA}{\bar{T}}               
\newcommand{\EAsub}{\bar{\Epsilon}\hspace{-0.06em}} 
\newcommand{\easub}{\bar{\epsilon}\hspace{-0.06em}} 
\newcommand{\TauAsub}{\bar{\Tau}\!}     
\newcommand{\DT}{\delta}                
\newcommand{\DE}{\mathit{\hat{\Delta}}} 
\newcommand{\de}{\mathit{\hat{\delta}}} 
\newcommand\fgc{\fC \gC c}
\newcommand\fg{\fC g}
\newcommand\fc{\fC c}
\newcommand\gc{\gC c}
\newcommand{\DI}{\mathit{\Lambda}}           
\newcommand{\DInorm}{\mathit{\hat\Lambda}}   
\newcommand{\di}{\bm\lambda}                 

\newcommand{\DIvecX}{%
  \raisemathbox{0.02em}{%
    \scalemathbox{0.9}{\vect{\mathit{\Lambda\hspace{-0.1em}}}}}}
\newcommand{\DInormvecX}{%
  \raisemathbox{0.02em}{%
    \scalemathbox{0.9}{\vect{\mathit{\hat\Lambda}}}}\hspace{-0.1em}}
\newcommand{\DImeanvecX}{%
  \raisemathbox{0.02em}{%
    \scalemathbox{0.9}{\vect{\mathit{\breve\Lambda}}}}\hspace{-0.1em}}
\newcommand{\divec}{\bm\lambda}

\newcommand\mc{\ct{M}_c}
\prerendermath{\DIvecBox}{\DIvecX}{\DIvec}
\prerendermath{\DInormvecBox}{\DInormvecX}{\DInormvec}
\prerendermath{\DImeanvecBox}{\DImeanvecX}{\DImeanvec}
\prerendermath{\EpsilonvecBox}{\EpsilonvecX}{\Epsilonvec}
\prerendermath{\TauvecBox}{\TauvecX}{\Tauvec}

\usepackage[hidelinks,linktoc=all]{hyperref}
\usepackage{hypcap}
\usepackage{bookmark}

\usepackage[nameinlink]{cleveref}
\crefname{myenumi}{item}{items}
\crefname{myenumii}{item}{items}
\crefname{assumption}{assumption}{assumptions}
\crefname{approximation}{approximation}{approximations}

\usepackage[noend]{algpseudocode}
\usepackage{algorithmicx}

\algdef{SE}[DOWHILE]{Do}{doWhile}{\algorithmicdo}[1]{\algorithmicwhile\ #1}%

\title{The Right Call for Software Benchmarking}
\subtitle{Consistent Decisions in Stateful Environments}

\author{\authorname{G\'abor Melis}\\
  \authoremail{melisgl@google}.\authoremail{com}\\
  \authororg{Google DeepMind}}

\begin{document}

\maketitle

\abstract{In the perpetual pursuit of performance, modern computing systems rely ever more on stateful mechanisms to accommodate the dynamics of workloads and physical environments, bolstering efficiency but confounding benchmarking and thereby the optimization of software.
Indeed, by their nature, adaptive mechanisms introduce temporal dependencies between measurements and render naive estimators of individual program performance biased.
Observing that rectifying such biases necessitates speculative assumptions about system dynamics, we call for prioritizing performance differentials over absolute measures and formalize software benchmarking as the decision problem of identifying the fastest program, for which relative knowledge suffices.
To this end, we propose simple experiment designs admitting consistent estimators of contrasts, whereby program-specific biases cancel under tenable assumptions.
These designs asymptotically yield the correct decision and afford a robust methodology for finite-budget benchmarking in stateful environments, bearing broad implications for the development of performance-sensitive software.}


\section{Introduction}

Software benchmarking \citep{kuhn1997structure, tichy1998should, sim2003using, alcocer2015tracking, hasselbring2021benchmarking} is at the heart of performance optimization: how could we guide this process if not by measurement?
Therein lies the fundamental importance of benchmarking, not unlike the experiments that ground natural sciences in physical reality.
Like them, benchmarking must distinguish cause from effect, but on modern computers we face an ever-changing set of strong confounders in the form of stateful processes in hardware controllers, the kernel, the myriad of daemons and other concurrently running programs, defying reliable control.
Beyond single computers, this issue is compounded: changing network conditions, co-hosting and the opaqueness of cloud services all conspire against reliable results.
Thus, our ability to detect small performance differences is constrained, yet precisely that sensitivity is required for the effective development and tuning of mature compilers, machine learning frameworks and database systems.

The present paper studies the challenges inherent to benchmarking programs in stateful environments, assuming the provision of a benchmark suite, whose design, while undeniably important, lies beyond the intended scope.
We shall examine the inadequate theoretical underpinnings of prevailing methods when applied in such environments and propose new approaches founded upon more tenable postulates.
Our contributions are threefold:
\begin{enumerate}
\item Formalize benchmarking as a decision problem concerning the fastest program in a given set under a general, albeit intractable, measurement model (\Cref{sec:problem-formalization});
\item Call for software benchmarking to focus on estimating performance deltas from the same experiment instead of comparing absolute performance numbers across experiments (\Cref{sec:absolute-performance,sec:deltas-to-decision});
\item Propose experiment designs that guide the selection of programs to call to ensure both asymptotically correct decisions (in the limit of infinitely many measurements) and finite-sample guarantees (\Crefrange{sec:delta-experiments}{sec:blocked-experiments}).
\end{enumerate}
The goal of this work is to accelerate the development of performance-sensitive software.

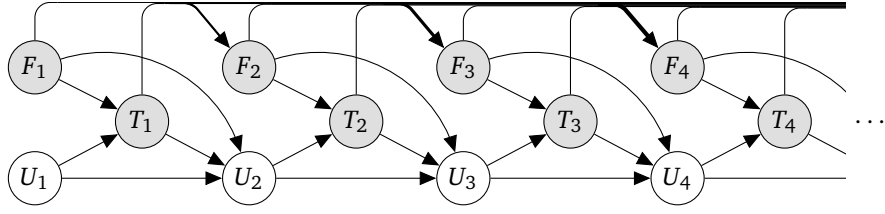
\begin{figure*}[t]
\centering
\begin{tikzpicture}
  \node[latent] (u1) {$U_1$};
  \node[obs, above=0.75cm of u1] (f1) {$F_1$};
  \node[obs, fill=gray!25, right=0.7 of u1, yshift=+0.75cm] (t1) {$T_1$};
  \node[latent, right=0.7 of t1, yshift=-0.75cm] (u2) {$U_2$};
  \edge {u1, f1} {t1}
  \draw[->] (f1) to [bend left=45] (u2);
  \edge {t1, u1} {u2}

  \node[obs, above=0.75cm of u2] (f2) {$F_2$};
  \node[obs, right=0.7 of u2, yshift=+0.75cm] (t2) {$T_2$};
  \node[latent, right=0.7 of t2, yshift=-0.75cm] (u3) {$U_3$};
  \edge {u2, f2} {t2}
  \draw[->] (f2) to [bend left=45] (u3);
  \edge {t2, u2} {u3}

  \node[obs, above=0.75cm of u3] (f3) {$F_3$};
  \node[obs, right=0.7 of u3, yshift=+0.75cm] (t3) {$T_3$};

  \node[latent, right=0.7 of t3, yshift=-0.75cm] (u4) {$U_4$};
  \edge {u3, f3} {t3}
  \draw[->] (f3) to [bend left=45] (u4);
  \edge {t3, u3} {u4}

  \node[obs, above=0.75cm of u4] (f4) {$F_4$};
  \node[obs, right=0.7 of u4, yshift=+0.75cm] (t4) {$T_4$};

  \path (f1.north) to ++(0,0.5cm) coordinate (hwtop);
  \node[right=0.45 of t4] (dots) {$\dots$};
  \clip (u1.south -| u1.west) rectangle (hwtop -| dots.west);

  \node (hwright) at ([xshift=1.5cm]t4.east) {};
  \draw[rounded corners=0.2cm] (f1.north) to ++(0,0.5cm) to ++(0.3cm,0) coordinate (hwtop); \draw (hwtop) to (hwtop -| hwright);
  \draw[rounded corners=0.2cm] (t1.north) to ([yshift=-0.4pt]t1.north |- hwtop) to ++(0.3cm,0) coordinate (hwtop); \draw (hwtop) to (hwtop -| hwright);
  \draw[line width=0.8pt, rounded corners=0.1cm, shorten <=0.1cm] (f2.north west) to ([yshift=0.2pt,xshift=-0.5cm]f2.north west |- hwtop) to ++(-0.3cm,0);
  \draw[<-,rounded corners=0.1cm] (f2.north west) to ([yshift=0.0pt,xshift=-0.5cm]f2.north west |- hwtop) to ++(-0.3cm,0);
  \draw[rounded corners=0.2cm] (f2.north) to ([yshift=-0.01cm]f2.north |- hwtop) to ++(0.3cm,0) coordinate (hwtop); \draw (hwtop) to (hwtop -| hwright);
  \draw[rounded corners=0.2cm] (t2.north) to ([yshift=-0.01cm]t2.north |- hwtop) to ++(0.3cm,0) coordinate (hwtop); \draw (hwtop) to (hwtop -| hwright);
  \draw[line width=1.2pt, rounded corners=0.1cm, shorten <=0.1cm] (f3.north west) to ([yshift=0.6pt,xshift=-0.5cm]f3.north west |- hwtop) to ++(-0.3cm,0);
  \draw[<-, rounded corners=0.1cm] (f3.north west) to ([yshift=0.3pt,xshift=-0.5cm]f3.north west |- hwtop) to ++(-0.3cm,0);
  \draw[rounded corners=0.2cm] (f3.north) to ([yshift=-0.01cm]f3.north |- hwtop) to ++(0.3cm,0) coordinate (hwtop); \draw (hwtop) to (hwtop -| hwright);
  \draw[rounded corners=0.2cm] (t3.north) to ([yshift=-0.01cm]t3.north |- hwtop) to ++(0.3cm,0) coordinate (hwtop); \draw (hwtop) to (hwtop -| hwright);
  \draw[line width=1.6pt, rounded corners=0.1cm, shorten <=0.1cm] (f4.north west) to ([yshift=0.9pt,xshift=-0.5cm]f4.north west |- hwtop) to ++(-0.3cm,0);
  \draw[<-, rounded corners=0.1cm] (f4.north west) to ([yshift=0.5pt,xshift=-0.5cm]f4.north west |- hwtop) to ++(-0.3cm,0);
  \draw[rounded corners=0.2cm] (f4.north) to ([yshift=-0.01cm]f4.north |- hwtop) to ++(0.3cm,0) coordinate (hwtop); \draw (hwtop) to (hwtop -| hwright);
  \draw[rounded corners=0.2cm] (t4.north) to ([yshift=-0.01cm]t4.north |- hwtop) to ++(0.3cm,0) coordinate (hwtop); \draw (hwtop) to (hwtop -| hwright);

  \node[latent, right=0.7 of t4, yshift=-0.75cm] (u5) {$U_5$};
  \edge {f4, u4} {t4}
  \draw[->] (f4) to [bend left=45] (u5);
  \edge {t4, u4} {u5}

\end{tikzpicture}
\caption{The general generative model (\Cref{def:configuration-model}) of sequential measurements in the presence of uncontrolled state $\vectrv{U}$ for a given configuration $c \in \sC$, depicted as a Bayesian network with observed nodes shaded and deterministic nodes as diamonds.
The program $F_i$ chosen to run is conditioned directly on all already observed variables (the previously run programs $\vectrv{F}_{<i}$ and their performance measurements $\vectrv{T}_{<i}$, typically run times), which in combination with $\vectrv{U}$ makes the estimation of the run times hard.
}
\label{fig:configuration-model}
\end{figure*}

\section{Conventions}

Standard sets are denoted with blackboard letters (e.g. $\sR$), other sets with calligraphic ($\sFC$ $\sU$), constants with uppercase Roman ($\ct{M}$), scalars with lowercase italic ($\fC$ $s$), and vectors with boldface italic ($\vect{u}$).
Random variables (henceforth rvs, or rv in singular) are denoted with uppercase italic ($F$\!, $S$) if scalar-valued, or in boldface ($\vectrv{U}$) if vector-valued.
Where realizations of a rv are denoted by a Greek lowercase letter ($\tau\!, \epsilon$) whose uppercase is identical to the italic, we use calligraphic uppercase ($\Tau\!, \Epsilon$) for the rv.

We abbreviate the condition or event $X = x$ to simply $x$ provided no ambiguity arises.
We employ the phrase \emph{under the event $X=x$} or \emph{under $x$} to denote operations within a conditional probability space.
In this context, $\E[Y]$ is understood as $\E[Y \given x]$, $\Var[Z \given y]$ as $\Var[Z \given y\lshh, x]$, and $A \indep B$ as $A \indep B \given x$, while statements such as $Y = 1$ are implicitly conditioned on $X=x$.
Conversely, we sometimes implicitly extend the probability space to include parameters, writing $P(X \given \theta)$ instead of $P(X; \theta)$.

When the domain is clear from context, we write $\forall c$ rather than $\forall c \in \sC$ or $\forall c \in \image(C)$.
It is assumed throughout that expectations exist.
The standard deviation of a rv $X$ is denoted by $\Std[X]$.
We may write $\E_{\mid X}[Y]$ instead of $\E[Y \given X]$ to improve readability, especially when nesting conditionals and absolute values.
Note that this differs from $\E_X[X + Y]$, which takes the expected value with respect to $X$ only and results in a rv that is a function of $Y$.
This latter form may be used solely for emphasis.

The length of sequence $\vect{a}$ is denoted by $\len{\vect{a}}$.
We abbreviate confidence interval as CI, data processing inequality as DPI, and mutual information as MI.
Our mathematical results are presented as theorems or, where of lesser consequence, propositions.
Technical tools used in the proofs are lemmas.

\section{Problem Formalization}
\label{sec:problem-formalization}

Let $\sF$ be the set of programs to compare with size $2 \leq \ct{F} < \infty$ and $C$ a rv over the set of configurations $\sC$.
The relative importances of configurations in the benchmark are given by the known prior distribution of $C$.
Configurations specify the controllable aspects of the benchmarking environment (such as hardware and software settings, input data).
In common terminology, a configuration is an individual benchmark, while the set $\sC$ corresponds to a benchmark suite.

The loss function $l \colon \sF \times \sC \to \sR$ quantifies the cost (e.g. run time) of running a program in a given configuration.
Our goal is to find $\argmin_{f \in \sF} l(f)$, the program with the best expected loss $l(f) = \E_C l(\fc)$.\footnotemark{}
This can be viewed as a no-data decision problem in $\fC$ where no-data refers to the absence of prior performance measurements.
The expected loss represents the Bayes risk in this context.
\footnotetext{Due to the symmetry between $\sF$ and $\sC$, this formalization can describe the task of benchmarking hardware on a set of programs, but our proposed solutions require the uncontrolled state to be shared (i.e. affect the performance of all $f \in \sF$ similarly, see \Cref{sec:delta-experiments}), which is much more reasonable for programs running on the same computer than for different computers running the same program.}

\subsection{The Measurement Model}
\label{sec:the-measurement-model}

To ground the design of experiments and decision rules for identifying the fastest program in data, we posit a generative model of measurements.
As per our premise, we must consider stateful environments where the system's behaviour is influenced by past events or hidden states; together with the specified configuration, this allows our model to account for both the uncontrolled and controlled factors that influence program performance in a sequential setting.
The model presented is fully general inasmuch as experiment designs may choose the next program to run based on all previous observables, and the state of the environment may depend on all earlier variables, whether observed or unobserved.

We introduce the uncontrolled state $U$\!, a rv over the set $\sU$\halfnegkern, to represent environmental dynamics without assuming any ability to intervene, observe, or model its underlying distribution and temporal evolution.
Factors comprising this state, such as concurrent workloads, OS scheduling, cache and TLB contents, memory controller transitions, and CPU thermal throttling, can vary even within the span of a single program execution.

The same variable may be controlled and uncontrolled at different times: by flushing the file system cache before running a program, the initial cache contents are controlled, but their evolution depends on the concurrent workload and the kernel version among many other influencing factors.

\begin{definition}[Configuration model]
\label{def:configuration-model}
Given a configuration $c \in \sC$ and an experiment design $d$ that prescribes the sequence of programs to run $\vectrv{F}$ and the stopping time $N$, the model generates a sequence of performance measurements $\vectrv{T}$ while the uncontrolled state $\vectrv{U}$ evolves.
\begin{myenum}
\item \textbf{Initialization}
\begin{myenum}
\item The initial uncontrolled state $u_1 \in \sU$ is supplied by the execution environment.
\item Initialize the index of the next run: $i = 1$.
\end{myenum}
\item \textbf{Iteration} While the experiment is ongoing ($i \leq n$, where $n \sim N \given \vect{t}_{<i}, \vect{f}_{<i},\cC d$):
\begin{myenum}
\item Choose a program $f_i \sim F_i \given \vect{t}_{<i}, \vect{f}_{<i},\cC d$, and
\item measure its performance: $t_i \sim T \given U = u_i, F = f_i, c$.\label{itm:t-given-u-f-c}
\item Let the uncontrolled state evolve: $u_{i+1} \sim U^{\mathrm{next}} \given T = t_i, U = u_i, F = f_i, c$.
\item Set the index of the next run: $i = i + 1$.
\end{myenum}
\end{myenum}
\end{definition}

In the above, $t_i$ and $u_{i+1}$ are drawn from distributions that do not depend on $i$ or $d$, modelling the state as Markovian.
Typically, these distributions, $T \given u, \fc$ and $U^{\mathrm{next}} \given t, u, \fc$, are unknown.
The rvs $C$, $N$, $\vectrv{F}$ and $\vectrv{T}$ (figuratively referred to as the run time) are observed; $\vectrv{U}$ is unobserved.

To tie the measurement model to the problem statement, we assume that the true loss $l(\fc)$ is a function of the distribution $T \given \fc$.
This distribution represents the performance of program $f$ under configuration $c$ marginalized over the possible uncontrolled states:
\begin{align*}
P(T \given \fc) = \E_U\BB{P\PP{T \given \UC \fc} P\PP{U \given c}}.
\end{align*}
This factorization separates the influence of $T \given \UC \fc$ (\Cref{itm:t-given-u-f-c}), from the unknown distribution of uncontrolled states, $U \given c$.
This latter distribution captures the overall prevalence and importance of uncontrolled states in real workloads, of which any particular execution environment -- along with the uncontrolled states $U_1$ and $U^{\mathrm{next}}$ it generates -- may be unrepresentative.
The allusion to the ``execution environment'' is to mark it as an extension point where additional interactions beyond the scope of the configuration model may arise.

We now define the model for benchmark suites that test all configurations.
The suite model can be understood as the product extension of the probability spaces of individual, per-configuration measurement models except we allow for additional dependencies between the uncontrolled states within certain groups of configurations, for example, those run on the same computer or in the same data center.

\begin{definition}[Benchmark suite model]
\label{def:benchmark-suite-model}
Given a partition $\sS$ of the set $\sC$, the benchmark suite model carries out the measurements for each configuration according to \Cref{def:configuration-model}, allowing for arbitrary dependencies between the initial uncontrolled states $\{U_1^c\}_{c \in \sS'}$ of configurations in the same part $\sS' \in \sS$, but requiring that they are otherwise independent.
\end{definition}

\section{The Mirage of Absolute Performance}
\label{sec:absolute-performance}

We allowed the loss $l(\fc)$ to be a function of the entire, unknown \emph{distribution} of the run times, $T \given \fc$,\footnote{Allowing this dependence on the distribution is still in keeping with von Neumann--Morgenstern rationality \citep{VonNeumann1944-VONTOG-4} if we take per-configuration aggregates as outcomes, instead of individual measurements.} although in the following sections we address only two common, pointwise cases: arithmetic ($l(\fc) = \E[T \given \fc]$) and geometric averaging ($l(\fc) = \E[\ln T \given \fc]$) \citep{mashey2004war, speccpu2006}.
Consider, for example, arithmetic averaging.
From the law of total expectation, we have that
\begin{align*}
l(\fc) = \E[T \given \fc] = \E_{U \given \fc} \E[T \given \UC \fc],\numberthis\label{eq:avg-over-u}
\end{align*}
which suggests estimating $l(f) = \E_C l(\fc)$ by sampling or enumerating $c \sim C$ and averaging estimates of $\E[T \given \UC \fc]$ over $U \given \fc$.
The most straightforward estimator is running $f$ repeatedly in configuration $c$ and averaging the observed run times.

\begin{definition}[Average over runs]\label{def:average-of-runs}
Given a sequence of programs $\vect{f}$ and corresponding values $\vect{x}$, both of length $n$, we define the average over runs of a program $f \in \sF$ as
\begin{align*}
\bar{x}_{f} = \frac{1}{n_f} \sum_{i=1}^n \indic_{f_i = f} \, x_i,
\end{align*}
where $n_f = \sum_{i=1}^r \indic_{f_i = f}$ is the number of times $f$ is run, and $\indic$ is the indicator function.
\end{definition}


There is a family of benchmarking methods that differ only in experiment design (i.e. in deciding what program to run, when to stop) and otherwise use this averaging estimator $\TA_{f}$ of the expected run time.
However, regardless of the experiment design, unbiased or consistent estimation of \eqref{eq:avg-over-u} is not possible without further assumptions as it involves the uncontrolled state, which may affect the loss.
The inadequacy of the simple average has been noted by various authors \citep{gil2011microbenchmark, chen2014statistical, barrett2017virtual}.

Despite its limitations, there are cases where $\TA_{f}$ can be unbiased and consistent:
\begin{itemize}
\item \inlineassumption{A.1} If $\sU$ is a singleton, then this is trivial.
\item \inlineassumption{A.2} If we can sample from $U \given c$, then we can appeal to the Central Limit Theorem (CLT).
\item \inlineassumption{A.3} If $\vectrv{U} \given\cC d$ is Markovian with stationary distribution $U \given c$, then we can use the Markov-chain CLT.
\item \inlineassumption{A.4} If the effect of the previous uncontrolled state on the run time follows a zero-mean random walk, then we can use the first-difference estimator.
\item \inlineassumption{A.5} If the environmental process mixes strongly enough, $\TA_{f}$ can be asymptotically unbiased and consistent.
\end{itemize}
However, all these assumptions are unrealistic, hard to verify or offer only asymptotic guarantees.
\inlineassumptionref{A.1} states that all features of the environment that can affect performance are known and controlled.
With the complexity of current hardware and software, such total control is practically unattainable.
Even on uncomplicated embedded systems, this assumption can be hard to verify and is not the focus of this work.
Next, \inlineassumptionref{A.2} is violated in the presence of stateful processes such as those mentioned earlier.

Concerning \inlineassumptionref{A.3}, these stateful environmental processes are Markovian, but they do not necessarily have a stationary distribution (e.g. due to reducibility or periodicity), so the Markov-chain CLT is not applicable.
In practice, we argue that the sequence of states produced by, for example, repeated runs of a program on any given machine may not visit the whole support of $U \given c$, let alone have a matching stationary distribution.
A particularly relevant example is that of room temperature, which -- if not controlled for -- may stay constant, drift slowly or exhibit a strong periodicity (e.g. due to a thermostat), affecting CPU temperature, thermal controller state \citep{enwiki:1215993085} and ultimately performance.
\inlineassumptionref{A.4} is easily violated by a single factor, such as the temperature, whose effect does not tend to follow a zero-mean random walk.
Finally, even if \inlineassumptionref{A.5} is known to be true, for it to be useful with a finite number of measurements, we need a non-trivial lower bound on the rate of mixing, which requires some knowledge of the environmental dynamics.

In summary, we have argued through examples that consistent estimation of $l(\fc)$, the cost of running a program in a given configuration, needs strong and often unrealistic assumptions about the uncontrolled state.
Therefore, relying on absolute performance measurements for benchmarking in stateful environments can be misleading and unreliable.

While unexplored possibilities undoubtedly remain, the more general argument is that without detailed knowledge of the environmental dynamics all assumptions are suspect.
Can we reverse-engineer the logic of the thermal controller and the kernel?
Is that worthwhile?
On every hardware and software configuration?
Experimental science has long grappled with similar challenges, employing tools like fixed-effect models, randomization and blocking.
In the following, we adapt these techniques to benchmarking.

\section{Deltas to Decision}
\label{sec:deltas-to-decision}

As discussed in the previous section, direct estimation of absolute performance numbers comparable across experiments is infeasible in the presence of uncontrolled state.
However, our formalization of benchmarking as a decision problem in \Cref{sec:problem-formalization} only requires estimating the performance of one program relative to another in the same experiment, which is an easier task.
Intuitively, if we have a biased estimator $\LE(\fc)$ of $l(\fc)$, then $\smash[b]{\argmin_{f \in \sF}} l(\fc) = \smash[b]{\argmin_{f \in \sF}} \E[\LE(\fc)]$ as long as its bias is the same for all $f \in \sFP$\footnotemark{}
We expand on this idea in the following.

\begin{definition}[Deltas]
\label{def:deltas}
In the context of some $\fg \in \sFC$ we define the true per-configuration and the suitewise average performance deltas as
\begin{gather*}
\DT^c(\fg) = l(\fc) - l(\gc),\\
\DT(\fg) = l(f) - l(g).
\end{gather*}
Assuming that $\LE(\fc)$ is a (potentially biased) estimator of $l(\fc)$, we define its suitewise average as $\LE(f) = \sum_{c \in \sC} P(c) \LE(\fc)$.
With these, the estimated per-configuration delta and the estimated average delta are defined as
\begin{gather*}
\DE^c(\fg) = \LE(\fc) - \LE(\gc)\\
\DE(\fg) = \LE(f) - \LE(g).
\end{gather*}
\end{definition}

\footnotetext{More generally, utility in decision theory is invariant to positive linear transformations.}

We use the shorthands $\DT^c$\supkern, $\DT$, $\DE^c$\supkern, $\DE$ where $f$ and $g$ are unambiguous in the context.

\begin{remark}
Although for each configuration, $\LE(\fc)$ lives in the probability space of its configuration model, the benchmark suite model ties them together in an extended probability space, so $\LE(f) = \smash[b]{\sum_{c \in \sC}} P(c) \LE(\fc)$ is well-defined given a partition $\sS$ and the intra-partition dependencies as required by \Cref{def:benchmark-suite-model}.
\end{remark}

\subsection{Finite-Sample Behaviour}
\label{sec:finite-sample-behaviour}

In the finite-sample regime, absolute certainty is unattainable.
We must therefore quantify the probability of making a correct decision.
The \emph{timing} of this quantification leads to a critical distinction:
\begin{enumerate}
\item \textbf{pre-data}: before starting the experiment, ignorant of the realization of the data;
\item \textbf{post-data}: after completing the experiment, conditioning on the observed data and the resulting decision.
\end{enumerate}
The pre-data question is: \emph{What is the probability that we will make the right decision if we follow the decision-making process given by the experiment design and the decision rule?}
In contrast, the post-data question is: \emph{What is the probability that we made the right decision given the decision-making process that we followed and the observed data?}

The pre-data perspective aligns naturally with frequentist inference, while the post-data perspective is the hallmark of Bayesian inference.
However, the dichotomy is not strict: a Bayesian can perform pre-data inference via the prior predictive distribution, and a frequentist can answer post-data questions by constructing confidence intervals conditioned on ancillary statistics \citep{casella1992conditional,goutis1995frequentist}.
Thus, the pre- and post-data distinction cuts across the choice of probabilistic paradigm.

Beyond their philosophical underpinnings, an operational distinction is that frequentists typically bound the worst-case error (ensuring validity for \emph{any} true parameter $\theta$), whereas Bayesians optimize the average-case error (integrating over a prior $\pi(\theta)$).
When the answer to the pre-data question is that \emph{the decision-making process gives the correct decision with probability 0.95}, the frequentist qualifies that with \emph{\dots well, at least 0.95, no matter what the real performance delta is}, while the Bayesian goes on to say \emph{\dots on average over possible deltas, given our prior beliefs}.

In this work, we evaluate methods based on their frequentist, pre-data properties, as we view benchmarking as a repeated process in software development.
While post-data analysis could theoretically improve efficiency of this process by quantifying the strength of evidence, our inability to model the complex environmental dynamics precludes a full Bayesian treatment.

\subsubsection{Decisions under Uncertainty}
\label{sec:decisions-under-uncertainty}

A natural approach to decision making is to construct symmetric CIs for $\DT(\fg)$ around $\DE(\fg)$ and deem $f$ faster than $g$ if the entire realized CI is negative.
However, as correct decision making requires knowing only the sign of $\DT(\fg)$, one of the two bounds embodied by the CI is unnecessary.
Thus, for increased statistical power (that is, reaching the same level of confidence in our decision making process with fewer runs), we require only one-sided bounds.

A decision rule maps observations to actions.
Here, the desired action is to identify a program from the set of fastest programs in $\sF$.
However, the finite number of measurements and the resulting uncertainty necessitates an augmented action space, $\sF \cup \{0\}$, where $0$ represents abstaining from a decision, and it is assumed that $0 \not\in \sF$.

\begin{definition}[The decision rule]
\label{def:decision-rule}
As input, the rule takes the estimates $\de(\fg)$ realized from the estimators $\DE(\fg)$ in the context of a measurement model with $n$ runs.
The required decision confidence $\gammadec \in [0, 1]$ is a parameter.
Let $\gammadecprime = 1 - (1 - \gammadec) / (\ct{F} - 1)$ be its Bonferroni correction.
Also, let $Z(\fg) = \DE(\fg) - \DT(\fg)$ be the rv representing the noise.
The decision rule picks
\begin{itemize}
\item randomly any $f^*\!\noscriptspace \in \sF$ such that $P(\de(f^*\lshd, g) \leq Z(f^*\lshd, g)) \geq \gammadecprime$ for all $g \in \sF \setminus \{f^*\}$,
\item $0$ if there is no such $f^*\!\noscriptspace$.
\end{itemize}
\end{definition}

Algebraically, the inequality $\DE(\fg) \leq Z(\fg)$ expands to $\DE(\fg) \leq \DE(\fg) - \DT(\fg)$, which simplifies to $\DT(\fg) \leq 0$.
Thus, the decision rule ensures that the data supports the conclusion that $f^*$ is no slower than any other program with confidence $\gammadecprime$ under the noise model $Z$.

While the exact noise distribution $Z$ is typically unknown in practice, we establish the theoretical properties of the decision rule here assuming it is known.In \Cref{sec:looking-ahead}, we discuss how to safely operationalize this rule using adversarial noise envelopes and finite sample budgets.

The rest of this subsection provides an analysis of the probabilistic guarantees of the decision rule.

\begin{definition}[Noise quantiles]
\label{def:noise-p-quantiles}
We define the $p$-quantile bound of the noise distribution $Z(\fg)$ as
\begin{align*}
Q_p(\fg) = \inf \{ x \in \sR \mid P(Z(\fg) \leq x) \geq p \}.
\end{align*}
\end{definition}

\begin{remark}[Equivalence to delta thresholding]
\label{rem:equivalence-thresholding}
The antisymmetry of $\DE$ and $\DT$ implies the antisymmetry of $Z$.
Thus, we can rewrite the probability condition in the decision rule as
\begin{align*}
P\bigl(\de(\fg) \leq Z(\fg)\bigr) &\geq \gammadecprime \\
P\bigl(\de(\fg) \leq -Z(g, f)\bigr) &\geq \gammadecprime \\
P\bigl(Z(g, f) \leq -\de(\fg)\bigr) &\geq \gammadecprime.
\end{align*}
By the definition of the quantile $\smash{Q_{\gammadecprime}}(g, f)$, this is the same as $-\de(\fg) \geq Q_{\gammadecprime}(g, f)$, which is equivalent to
\begin{gather*}
\de(\fg) \leq -Q_{\gammadecprime}(g, f).
\end{gather*}
\end{remark}

\vspace{-\baselineskip}

\begin{restatable}[Unambiguous region]{proposition}{thmunambiguousregion}
\label{thm:unambiguous-region}
Let $\gamma \geq \gammadec$ (the ``detection confidence''), and let $\gamma' = 1 - (1 - \gamma) / (\ct{F} - 1)$.
We define the \textit{separation threshold} $S_{\gamma}(\fg)$ as
\begin{align*}
S_{\gamma}(\fg) = - Q_{\gamma'}(\fg) - Q_{\gammadecprime}(g, f).
\end{align*}
If the true deltas for some $f \in \sF$ are in the corresponding unambiguous region, meaning
\begin{align*}
\DT(\fg) \leq S_{\gamma}(\fg) \quad \text{and} \quad \DT(\fg) < 0
\end{align*}
for all $g \neq f$, then this $f$ is the unique fastest program, and the decision rule picks $f$ with probability at least $\gamma$.
\end{restatable}
\noindent For the proofs of propositions in this section, see \Cref{sec:proofs-for-deltas-to-decisions}.

\Cref{tab:confusion-matrix-probabilistic} shows the probabilistic confusion matrix of the decision rule by regions defined by the true deltas $\DT(f, g) \leq 0$, where $f$ is one of the fastest programs and $g \neq f$.\footnotemark
\begin{itemize}
\item \textbf{Tie} ($\exists g \colon \DT(\fg)=0$):
For a specific fastest program $f$ to be picked, it must statistically ``beat'' every rival.
The bottleneck is the tied rival $g$: since $\DT(\fg)=0$, $f$ can only win if the noise erroneously exceeds the decision threshold.
This is a false positive event with probability at most $1 - \gammadecprime$.
Since the decision rule enforces a unique winner, the events of picking any specific program are disjoint.
Thus, the total probability $p_0$ of picking \textit{any} of the $\ct{T}$ fastest programs is at most the sum of these tail probabilities: $\ct{T}(1 - \gammadecprime)$.
Substituting the definition of $\smash{\gammadecprime}$, we verify $p_0 = \frac{\ct{T}}{\ct{F}-1} (1 - \gammadec)$.
\item \textbf{Ambiguous region} ($\forall g \colon \DT(\fg) < 0$ and $\exists g \colon \DT(\fg) > S_\gamma(\fg)$):
The action is Fastest or Abstain with probability $\gammadec$, but we cannot bound their probabilities separately.
It may be that the decision rule always abstains.
\item \textbf{Unambiguous region} ($\forall g \colon \DT(\fg) < 0$ and $\DT(\fg) \leq S_\gamma(\fg)$):
The fastest program is picked with confidence $\gamma$.
Abstention or picking a slower program have at most probability $1 - \gamma$ together.
\end{itemize}

\footnotetext{
Although we emphasize the decision-theoretic point of view throughout, it is worth pointing out that our decision rule corresponds to frequentist multiple hypothesis testing.
For a specific candidate $f^*$\lshd, deciding that it is the winner requires rejecting the union of $\ct{F} - 1$ pairwise null hypotheses, where each $H_{0, g}$ states that $f^*$ is not faster than $g$.
This is an \textit{intersection--union} test: we declare $f^*$ the winner only if \textit{every} pairwise null is rejected.
In this framework, we refer to the following global error types.
\begin{itemize}
\item \textbf{Type I Error (False Positive):}
This occurs when the global null hypothesis holds (i.e. there is no unique winner) yet the decision rule rejects it by picking one of the programs.
This is a mild error in benchmarking as there is no performance loss.
\item \textbf{Type II Error (False Negative):}
This occurs when the alternative hypothesis holds (i.e. there is a unique winner) but the decision rule fails to reject the null by abstaining.
This represents a lack of statistical power: the signal exists, but the test was too conservative or the noise too high to detect it with the required confidence.
\item \textbf{Type III Error (Sign Error):}
Unlike standard binary hypothesis testing, our setup allows for a third outcome: rejecting the null in the \textit{wrong direction} (declaring a slower program the winner).
\end{itemize}}

\begin{table}
\caption{The probabilistic confusion matrix for the decision rule (\Cref{def:decision-rule}) with a given decision confidence $\gammadec$.
The detection confidence $\gamma$ is purely for analysis, where it is assumed that $\gammadec \leq \gamma \leq 1$.
The number of programs is $\ct{F}$, the number of programs tied for fastest is $\ct{T}$.
We let $p_0 = \frac{\ct{T}}{\ct{F}-1} (1 - \gammadec)$ and $p_1 = \frac{\ct{F} - \ct{T}}{\ct{F}-1} (1 - \gammadec)$.
The deltas between the fastest and other programs are denoted by $\DT(\fg) \leq 0$ for all $g \neq f$.
The region \textbf{Tie} corresponds to the condition $\exists g \colon \DT(\fg)=0$; \textbf{Ambiguous} to $\forall g \colon \DT(\fg) < 0 \land \exists g \colon \DT(\fg) > S_\gamma(\fg)$; and \textbf{Unambiguous} to $\forall g \colon \DT(\fg) < 0 \land \DT(\fg) \leq S_\gamma(\fg)$, where $S_\gamma(\fg)$ denotes the separation bound for the pair $\fg$ defined in \Cref{thm:unambiguous-region}.
Cells corresponding to Type I, Type II and Type III errors are coloured \mglcolorbox{yellow!35}{yellow}, \mglcolorbox{orange!50}{orange} and \mglcolorbox{red!75}{red}, respectively, while those representing the correct action are white.}
\label{tab:confusion-matrix-probabilistic}
\centering
\newcommand{\CellHit}{\cellcolor{white}}
\newcommand{\CellTypeI}{\cellcolor{yellow!35}}
\newcommand{\CellTypeII}{\cellcolor{orange!50}}
\newcommand{\CellTypeIII}{\cellcolor{red!75}}

\begin{tabular}{l @{\quad} c c c}
\toprule
& \multicolumn{3}{c}{Region of $\DT(\fC *)$} \\
\cmidrule(l){2-4}
Action & Tie                                  & Ambiguous                                        & Unambig. \\
\midrule
Fastest & \CellTypeI   $\leq p_0$             & \CellHit                                         & \CellHit     $\gamma \leq$ \\
Abstain & \CellHit     $1 - (p_0 + p_1) \leq$ & \CellTypeII \multirow{-2}{*}{$\gammadec \leq$}   & \CellTypeII \\
Slower  & \CellTypeIII $\leq p_1$             & \CellTypeIII                $\leq 1 - \gammadec$ & \CellTypeIII \multirow{-2}{*}{$\leq 1 - \gamma$} \\
\bottomrule
\end{tabular}
\end{table}

\subsection{Asymptotics}

In this section, we prove that the decision rule makes asymptotically correct decisions if the per-configuration estimators exhibit a certain form of consistency.
For sample efficiency in the finite regime, we based the decision rule (\Cref{def:decision-rule}) on the notion of stochastic upper bounds defined by the noise quantiles (\Cref{rem:equivalence-thresholding}).
Asymptotically, we expect these noise quantiles to tighten around zero.
We formalize this as a one-sided variant of weak consistency.

\begin{definition}[Upper consistency]
We say that a sequence of estimators $\hat{\theta}_n$ of some $\theta \in \sR$ is upper consistent if the probability that the estimator exceeds the true value by any fixed margin vanishes as $n \to \infty$.
Formally,
\begin{align*}
\forall \omega > 0 \colon \lim_{n \to \infty} P(\hat{\theta}_n - \theta \leq \omega) = 1.
\end{align*}
\end{definition}

To tie the discussion to the decision rule, we point out that the upper consistency of $\DE(\fg)$ is equivalent to requiring that the noise quantiles (\Cref{def:noise-p-quantiles}) become non-positive in the limit.
Specifically, for any fixed confidence $p < 1$, the estimator $\DE(\fg)$ is upper consistent if and only if
\begin{align*}
\limsup_{n \to \infty} Q_p(\fg) \leq 0.
\end{align*}

In finite-sample analysis, we must also characterize the rate at which these bounds shrink.
Next, we show that if the per-configuration estimators $\DE^c$ admit upper bounds that scale as $\bigO(n^{\negone/2})$, then $\DE$ preserves this rate.
Crucially, we construct the suitewise confidence level by combining the per-configuration probabilities according to the dependency structure of the benchmark suite model.

\begin{restatable}[Suitewise upper consistency and rates]{proposition}{thmdectodeconsistencyrate}
\label{thm:dec-to-de-consistency-rate}
Let $\fg \in \sF$.
Assume that for each configuration $c \in \sC$, the estimator $\DE^c(\fg)$ of $\DT^c(\fg)$ admits an upper bound $U_n^c$ such that $P(\DE^c(\fg) - \DT^c(\fg) \leq U_n^c) \geq 1 - \alpha_c$ for some failure probability $\alpha_c \in (0, 1)$.
Assume further that $\max(0, U_n^c) = \bigO(\phi(n))$ uniformly for all $c \in \sC$.
Let the suitewise bound be $U_n = \sum_{c \in \sC} P(c) U_n^c$.
Then, provided the sum and product exist, $U_n$ is a probabilistic upper bound satisfying
\begin{align*}
P(\DE(\fg) - \DT(\fg) \leq U_n) \geq \prod_{\sS' \in \sS} \Pigl( 1 - \min\Bigl(1, \sum_{c \in \sS'} \alpha_c \Bigr) \Pigr),
\end{align*}
and $\max(0, U_n) = \bigO(\phi(n))$.
\end{restatable}

\begin{corollary}
If $\lim_n \phi(n) = 0$, then $\DE^c$ is upper consistent for all $c$ and so is $\DE$.
\end{corollary}

Finally, we prove that the decision rule is consistent: given a sufficient number of runs, the unambiguous region expands to cover any program with a strictly negative true delta.

\begin{restatable}[Consistency of the decision rule]{proposition}{thmuncertainruleconsistency}
\label{thm:uncertain-rule-consistency}
Assume there is a unique fastest program $f \in \sF$.
If the estimator $\DE$ is upper consistent, then
\begin{align*}
\lim_{n \to \infty} P(\text{decision rule picks } f) = 1.
\end{align*}
\end{restatable}

\begin{remark}[Consistency and ties]
In the case of tied programs (where $\DT(\fg) = 0$), the probability of correctly abstaining does not converge to $1$ as $n \to \infty$, but rather remains lower-bounded by $1 - (p_0 + p_1) = 1 - \frac{\ct{F}}{\ct{F}-1}(1 - \gammadec)$ (see Tie / Abstain cell in \Cref{tab:confusion-matrix-probabilistic}).
This is a standard feature of fixed-confidence statistical tests: to force the false positive rate to zero asymptotically, one would need to let $\gammadec \to 1$.
\end{remark}

\subsection{Looking Ahead}
\label{sec:looking-ahead}

In our setting, the noise distribution $Z$ is typically unknown.
Consequently, we cannot directly evaluate the quantile thresholds $Q_p(\fg)$ in the decision rule.
To ensure robustness, we define the noise model used in the decision rule not as a single physical realization but as the \emph{adversarial envelope} over possible noise models $\mathcal{Z}$.
In practice, we want to detect when the actual noise is not in $\mathcal{Z}$.
Our proposed methods are based on the simplest such setup: requiring a hard bound on the magnitude of the noise.

\begin{definition}[Adversarial noise quantiles]
We define the adversarial $p$-quantile $Q_p^{\mathrm{adv}}$ as
\begin{align*}
Q_p^{\mathrm{adv}} = \smash[b]{\sup_{Z \in \mathcal{Z}, \fg \in \sF}} \inf \{ x \in \sR \mid P(Z(\fg) \leq x) \geq p \}.
\end{align*}
\end{definition}

Replacing the unknown noise quantiles with their adversarial counterparts makes the decision rule \emph{more conservative}.
Since the adversarial bound is a supremum, satisfying the decision condition $\de(\fg) \leq -\smash{Q^{\mathrm{adv}}_p}$ implies satisfying $\de(\fg) \leq -Q_p(g, f)$ for the true, unknown noise distribution.
With the choice of $U_n^c = \smash{Q^{\mathrm{adv}}_p}$, our results for both the finite and the asymptotic regimes remain valid.

In the context of our experimental designs, we will bound $\smash[b]{Q^{\mathrm{adv}}_p}$ to establish the upper consistency of the per-configuration estimators $\DE^c(\fg)$, which guarantees asymptotically correct decisions via \Cref{thm:dec-to-de-consistency-rate,thm:uncertain-rule-consistency}.

However, for practical benchmarking, we must satisfy user requirements on the confidence level and the maximum width of these stochastic upper bounds, corresponding to the minimum effect size that the decision rule can detect.
Furnished with known, vanishing probabilistic upper bounds, we can readily answer a user's question: \emph{Is this benchmarking method able to identify the fastest program?}
We might reply thus:
\emph{Yes, because the underlying estimator of the deltas is upper consistent, but you'll need to wait forever to be sure.}
A practical mind would then enquire along these lines:
\emph{I don't have the time.
How many runs are needed to detect a one percent difference with 0.95 probability?}
To which, the answer: \emph{Let me compute that for you from the bounds on the estimator.}
We will prove in \Crefrange{sec:delta-experiments}{sec:alpha-experiments} and validate experimentally in \Cref{sec:experiments} that the number of runs required grows at the theoretically optimal rate with our methods.

\begin{figure*}[t]
\centering
\begin{tikzpicture}
  \node[latent] (u1) {$U_1$};
  \node[obs, above=1.25cm of u1] (f1) {$F_1$};
  \node[det, fill=gray!25, right=0.05 of f1, yshift=-1cm] (d1) {$\DI_1$};
  \node[latent, right=1.0 of u1, yshift=+0.5cm] (e1) {$\Epsilonsub_1$};
  \node[latent, right=1.0 of f1, yshift=-0.5cm] (tau1) {$\Tausub_1$};
  \node[det, fill=gray!25, right=0.95 of e1, yshift=+0.5cm] (t1) {$T_1$};
  \node[latent, right=0.7 of t1, yshift=-1cm] (u2) {$U_2$};
  \edge {f1} {d1}
  \edge {u1} {e1}
  \edge {f1} {tau1}
  \edge {e1, tau1} {t1}
  \draw[->] (f1) to [bend left=45] (u2);
  \edge {t1, u1} {u2}

  \node[obs, above=1.25cm of u2] (f2) {$F_2$};
  \node[det, fill=gray!25, right=0.05 of f2, yshift=-1cm] (d2) {$\DI_2$};
  \node[latent, right=1.0 of u2, yshift=+0.5cm] (e2) {$\Epsilonsub_2$};
  \node[latent, right=1.0 of f2, yshift=-0.5cm] (tau2) {$\Tausub_2$};
  \node[det, fill=gray!25, right=0.95 of e2, yshift=+0.5cm] (t2) {$T_2$};
  \node[latent, right=0.7 of t2, yshift=-1cm] (u3) {$U_3$};
  \edge {f2} {d2}
  \edge {u2} {e2}
  \edge {f2} {tau2}
  \edge {e2, tau2} {t2}
  \draw[->] (f2) to [bend left=45] (u3);
  \edge {t2, u2} {u3}

  \node[obs, above=1.25cm of u3] (f3) {$F_3$};
  \node[det, fill=gray!25, right=0.05 of f3, yshift=-1cm] (d3) {$\DI_3$};
  \node[latent, right=1.0 of u3, yshift=+0.5cm] (e3) {$\Epsilonsub_3$};
  \node[latent, right=1.0 of f3, yshift=-0.5cm] (tau3) {$\Tausub_3$};
  \node[det, fill=gray!25, right=0.95 of e3, yshift=+0.5cm] (t3) {$T_3$};
  \edge {f3} {d3}
  \edge {u3} {e3}
  \edge {f3} {tau3}
  \edge {e3, tau3} {t3}

  \path (f1.north) to ++(0,0.5cm) coordinate (hwtop);
  \node[latent, right=0.7 of t3, yshift=-1cm] (u4) {$U_4$};
  \node[obs, above=1.25cm of u4] (f4) {$F_4$};
  \draw[->] (f3) to [bend left=45] (u4);
  \edge {t3, u3} {u4}

  \draw[rounded corners=0.2cm] (f1.north) to ++(0,0.5cm) to ++(1.5cm,0) coordinate (hwtop);
  \node (hwright) at ([xshift=0.5cm]f4.north west |- hwtop) {};
  \draw (hwtop) to (hwtop -| hwright);

  \draw[rounded corners=0.1cm, ->] (hwtop) to ([xshift=-0.53cm]f2.north west |- hwtop) to (f2.north west);
  \draw[rounded corners=0.2cm] (f2.north) to ([yshift=-0.4pt]f2.north |- hwtop) to ++(1.5cm,0) coordinate (hwtop);
  \draw[->] (hwtop) to ([xshift=0.1cm]hwtop -| hwright) node[det, fill=gray!25, right]{$\vectrv{N}$};

  \draw[rounded corners=0.1cm, ->] (hwtop) to ([xshift=-0.53cm]f3.north west |- hwtop) to (f3.north west);
  \draw[line width=0.8pt, rounded corners=0.1cm, shorten >=0.2cm] ([yshift=+0.2pt]hwtop) to ([yshift=+0.2pt,xshift=-0.53cm]f3.north west |- hwtop) to (f3.north west);
  \draw[rounded corners=0.2cm] (f3.north) to ([yshift=-0.4pt]f3.north |- hwtop) to ++(1.5cm,0) coordinate (hwtop);
  \draw (hwtop) to (hwtop -| hwright);

  \draw[rounded corners=0.1cm, ->] ([yshift=0.4pt]hwtop) to ([yshift=0.4pt,xshift=-0.53cm]f4.north west |- hwtop) to (f4.north west);
  \draw[line width=1.2pt, rounded corners=0.1cm, shorten >=0.2cm] ([yshift=0.4pt]hwtop) to ([yshift=0.4pt,xshift=-0.53cm]f4.north west |- hwtop) to (f4.north west);

  \node[right=0.2 of t3, yshift=0.0cm] (dots) {$x$};
  \fill[white] ([xshift=-0.1cm]u1.south -| dots.west) rectangle ([xshift=1pt, yshift=-5pt]hwtop -| f4.east);
  \fill[white] ([xshift=-0.1cm,yshift=-5pt]hwtop -| dots.west) rectangle ([xshift=0.1cm,yshift=+1.2pt]hwtop -| dots.east);
  \draw[gray, decorate, decoration={snake, amplitude=0.7mm, segment length=1.075cm}] ([xshift=0.1cm]dots.west |- u1) -- ([xshift=0.1cm,yshift=1pt]dots.west |- hwtop);
  \draw[gray, decorate, decoration={snake, amplitude=0.7mm, segment length=1.075cm}] ([xshift=0.35cm]dots.west |- u1) -- ([xshift=0.35cm,yshift=1pt]dots.west |- hwtop);

\end{tikzpicture}
\caption{The delta design (\Cref{def:delta-design}).
This constrains the configuration model (\Cref{fig:configuration-model}) but leaves the sampling strategy for $\vectrv{F}$ underspecified.
Under its ceteris paribus assumption (\Cref{sec:delta-experiments}), the observed [log] run time $T_i$ is a the sum of the unobserved true [log] run time $\Tausub_i$ and the environmental noise $\Epsilonsub_i$.
The dependence of $\Epsilonsub_i$ on $F_i$ is mediated solely through $\vectrv{F}_{<i}$; so, by the DPI, we have $I(\Epsilonsub_i; F_i) \leq I(\vectrv{F}_{<i}; F_i)$.
Importantly, the mutual information bound $I(\vectrv{F}_{<i}; F_i)$ is determined by the sampling process independently from the stateful environment, giving us control on the effect of the environmental noise $\Epsilon$.
The shorthands $\DI_i = \indic_{F_i=f} - \indic_{F_i=g}$ and $N_f = \sum_{i=1}^R \indic_{F_i = f}$ (\Cref{def:delta-shorthands}) are used in the proofs.}
\label{fig:delta-design}
\end{figure*}
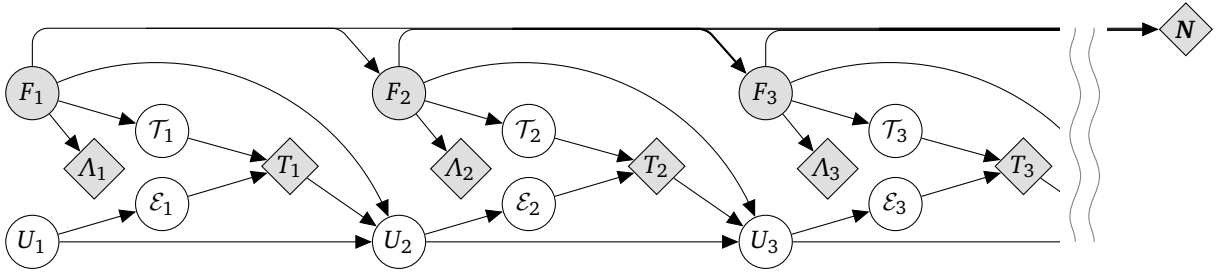

\section{Delta Experiments}
\label{sec:delta-experiments}

To characterize the behaviour of an experiment design, our strategy is to derive stochastic upper bounds for $\DE^c$ that vanish with the number of runs $n$.
In the finite case, this sequence of bounds allows us to find the minimum runs required to meet a given sensitivity and confidence level.
For asymptotics, appealing to \Cref{thm:dec-to-de-consistency-rate,thm:uncertain-rule-consistency} helps derive the consistency of the decision-making process.

The experiment designs, presented in later sections, exhibit significant similarities, and to avoid unnecessary repetition, we will discuss their properties in as much generality as possible.
To this end, we now describe the class of delta designs, their common substrate, which constrains the configuration model (\Cref{def:configuration-model}) but leaves the sampling strategy for $\vectrv{F}$ underspecified.

First, we introduce alternative versions of our \emph{ceteris paribus} (``all things being equal'') assumption positing that environmental noise affects the loss additively and uniformly across all programs.

\begin{assumption}[Additive global effect]\label{ass:additive-global-effect}
Let $l(\fc) = \EB{T \given \fc}$ and $T \given \UC \fc$ of the measurement model (\Cref{def:configuration-model}) be generated as follows:
\begin{enumerate}
\item Draw the true run time $\tau \given \fc$.
\item Draw the environmental noise $\varepsilon \given u, c$.
\item Let the observed run time $t = \tau + \varepsilon$.
\end{enumerate}
Here, $\tau\halfnegkern, \varepsilon \in \sR$ are unobserved, their distributions unknown.
\end{assumption}

\begin{assumption}[Multiplicative global effect]\label{ass:multiplicative-global-effect}
Let $l(\fc) = \EB{\ln T \given \fc}$ and $T \given \UC \fc$ of the measurement model (\Cref{def:configuration-model}) be generated as follows:
\begin{enumerate}
\item Draw the true run time $\tau \given \fc$.
\item Draw the environmental noise $\varepsilon \given u, c$.
\item Let the observed run time $t = \tau \varepsilon$.
\end{enumerate}
Here, $\tau\halfnegkern, \varepsilon \in \sR^+$ are unobserved, their distributions unknown.
\end{assumption}

These losses cover the two most common cases \citep{mashey2004war}.
Choosing between them implies taking a side in the long-running arithmetic vs geometric mean argument \citep{mashey2004war}.
Crucially, our approach answers this question not by assuming a distribution for the raw run times $\Tau$ or the noise $\Epsilon$, but by deriving the estimator directly from the functional form of the noise interaction.

Since, in the multiplicative version, the noise is combined with the logarithm of the loss and the distributions of $\Tau$\! and $\Epsilon$\halfnegkern are unknown, we obtain essentially the same model as with the additive assumption.
Hence, in the follows, we simply assume the ceteris paribus holds and do not distinguish between the two cases.
We write $T$\supkern, $\Tau$\supkern, and $\Epsilon$\lshh, even though under \Cref{ass:multiplicative-global-effect} they represent $\ln T$\lsh, $\ln \Tau$, and $\ln \Epsilon$\lshh, respectively, as if we observed $\ln T$ instead of $T$\mathendkern{T}{.}

In practice, we prefer the multiplicative model because it is a good match for common effects that accumulate over the course of a run (e.g. CPU throttling).
For some effects and for some sets of programs, the additive model may be more appropriate, but as we argue in \Cref{sec:the-single-machine-case}, unbiased estimation with both multiplicative and additive effects is possible only with a single hardware and software configuration.

Both assumptions relate the more tractable $T \given \UC \fc$ to $T \given \fc$.
Playing the role of ceteris paribus, they assert sufficiency of control by requiring the effect of the uncontrolled variables to be the same on all programs.
Conversely, all variables whose effects vary by program must be controlled or randomized \citep{curtsinger2013stabilizer}.

\begin{remark}[Inevitability of a ceteris paribus]
The assumption that uncontrolled noise affects all programs uniformly is a necessary consequence of not modelling state dynamics, a key requirement for our method's broad applicability.
To see why, consider the case where a single uncontrolled state influences programs differently.
Not even with infinite data could we determine the better program, as any observed performance difference could be attributed to this state.
This is particularly clear when the unobserved uncontrolled state remains constant throughout the benchmarking process.
\end{remark}

Also note that the global-effect assumption is made for a given set of programs, and the more alike the programs in $\sFC$ the more similar the impact of uncontrolled states on them.
For example, if two programs do proportionally the same amount of CPU and I\narrowslash{}O work, then CPU throttling or concurrent I\narrowslash{}O may influence their performance similarly.

\begin{definition}[Delta design, see \Cref{fig:delta-design}]
\label{def:delta-design}
We say that an experiment design $d$ based on the measurement model is a delta design and denote it with $d \in \sD^\Delta$ if under $d$
\begin{itemize}
\item \Cref{ass:additive-global-effect} or \labelcref{ass:multiplicative-global-effect} holds,
\item $\norm{T}_\infty < \infty$,
\end{itemize}
and, given any $r \in \sN$ and for all $i \in [1, n]$,
\begin{itemize}
\item $F_i \indep (\vectrv{T}_{<i}, C) \given F_{<i}, n$,
\item $\operatorname{supp}(F_i \given n) = \sFC$ and
\item $\LE(\fc) = \TA_{f}$.
\end{itemize}
\end{definition}

The assumption that $\LE(\fc) = \TA_{f}$ (see \Cref{def:average-of-runs}) for all $f \in \sF$ determines $\DE^c = \DE(\fgc)$ in \Cref{def:deltas}, so far only partially specified for being dependent on the choice of $\LE$.

By this definition, the $F_i$ depend only on $\vectrv{F}_{<i}$ and assign non-zero probability to all programs.
This constrains but does not fully specify how $\vectrv{F}$ is sampled.

\begin{figure*}[t]
\centering
\scalebox{0.68}{
\begin{tikzpicture}
  \node[latent, xshift=-0.5cm] (u1) {$U_1$};
  \node[obs, above=1.25cm of u1] (f1) {$F_1$};
  \node[latent, right=0.7 of u1, yshift=+0.5cm] (e1) {$\Epsilonsub_1$};
  \node[latent, right=0.7 of f1, yshift=-0.5cm] (tau1) {$\Tausub_1$};
  \node[det, fill=gray!25, right=0.7 of e1, yshift=0.5cm] (t1) {$T_1$};
  \node[latent, right=0.25 of t1, yshift=-1.0cm] (u2) {$U_2$};
  \edge {u1} {e1}
  \edge {f1} {tau1}
  \edge {e1, tau1} {t1}
  \edge {u1, t1} {u2}
  \draw[->] (f1) to [bend left=50] (u2);

  \node[obs, above=1.25cm of u2] (f2) {$F_2$};
  \node[latent, right=0.7 of u2, yshift=+0.5cm] (e2) {$\Epsilonsub_2$};
  \node[latent, right=0.7 of f2, yshift=-0.5cm] (tau2) {$\Tausub_2$};
  \node[det, fill=gray!25, right=0.7 of e2, yshift=+0.5cm] (t2) {$T_2$};
  \node[latent, right=1 of t2, yshift=-1.0cm] (u3) {$U_3$};
  \edge {u2} {e2}
  \edge {f2} {tau2}
  \edge {e2, tau2} {t2}
  \edge {u2, t2} {u3}
  \draw[->] (f2) to [bend left=40] (u3);

  \node[obs, above=1.25cm of u3] (f3) {$F_3$};
  \node[latent, right=0.7 of u3, yshift=+0.5cm] (e3) {$\Epsilonsub_3$};
  \node[latent, right=0.7 of f3, yshift=-0.5cm] (tau3) {$\Tausub_3$};
  \node[det, fill=gray!25, right=0.7 of e3, yshift=+0.5cm] (t3) {$T_3$};
  \node[latent, right=0.25 of t3, yshift=-1.0cm] (u4) {$U_4$};
  \edge {u3} {e3}
  \edge {f3} {tau3}
  \edge {e3, tau3} {t3}
  \edge {u3, t3} {u4}
  \draw[->] (f3) to [bend left=50] (u4);

  \node[obs, above=1.25cm of u4] (f4) {$F_4$};
  \node[latent, right=0.7 of u4, yshift=+0.5cm] (e4) {$\Epsilonsub_4$};
  \node[latent, right=0.7 of f4, yshift=-0.5cm] (tau4) {$\Tausub_4$};
  \node[det, fill=gray!25, right=0.7 of e4, yshift=+0.5cm] (t4) {$T_4$};
  \node[latent, right=0.25 of t4, yshift=-1.0cm] (u5) {$U_5$};
  \edge {u4} {e4}
  \edge {f4} {tau4}
  \edge {e4, tau4} {t4}
  \edge {u4, t4} {u5}
  \draw[->] (f4) to [bend left=50] (u5);

  \node[obs, above=1.25cm of u5] (f5) {$F_5$};
  \node[latent, right=0.7 of u5, yshift=+0.5cm] (e5) {$\Epsilonsub_5$};
  \node[latent, right=0.7 of f5, yshift=-0.5cm] (tau5) {$\Tausub_5$};
  \node[det, fill=gray!25, right=0.7 of e5, yshift=+0.5cm] (t5) {$T_5$};
  \node[latent, right=0.25 of t5, yshift=-1.0cm] (u6) {$U_6$};
  \edge {u5} {e5}
  \edge {f5} {tau5}
  \edge {e5, tau5} {t5}
  \edge {u5, t5} {u6}
  \draw[->] (f5) to [bend left=50] (u6);

  \node[obs, above=1.25cm of u6] (f6) {$F_6$};
  \node[latent, right=0.7 of u6, yshift=+0.5cm] (e6) {$\Epsilonsub_6$};
  \node[latent, right=0.7 of f6, yshift=-0.5cm] (tau6) {$\Tausub_6$};
  \node[det, fill=gray!25, right=0.7 of e6, yshift=+0.5cm] (t6) {$T_6$};
  \edge {u6} {e6}
  \edge {f6} {tau6}
  \edge {e6, tau6} {t6}

  \plate[inner xsep=0.25cm, inner ysep=0.4cm, yshift=0.4cm] {} {
    (u1) (f1) (tau1) (e1) (t1) (u2) (f2) (tau2) (e2) (t2)} {}

  \plate[inner xsep=0.25cm, inner ysep=0.4cm, yshift=0.4cm] {} {
    (u3) (f3) (tau3) (e3) (t3) (u4) (f4) (tau4) (e4) (t4)
    (u5) (f5) (tau5) (e5) (t5) (u6) (f6) (tau6) (e6) (t6)} {}

  \draw[rounded corners=0.2cm] (f1.north) to ++(0,0.5cm) to ++(1.5cm,0) coordinate (hwtop);
  \draw[->,rounded corners=0.1cm] (hwtop) to ([xshift=-0.53cm]f2.north west |- hwtop) to (f2.north west);

  \draw[rounded corners=0.2cm] (f3.north) to (f3.north |- hwtop) to ++(1.5cm,0) coordinate (hwtop);
  \node (hwright) at ([xshift=-1cm]f6.north west |- hwtop) {};
  \draw (hwtop) to (hwtop -| hwright);

  \draw[rounded corners=0.1cm, ->] (hwtop) to ([xshift=-0.53cm]f4.north west |- hwtop) to (f4.north west);
  \draw[rounded corners=0.2cm] (f4.north) to ([yshift=-0.4pt]f4.north |- hwtop) to ++(1.5cm,0) coordinate (hwtop);
  \draw (hwtop) to (hwtop -| hwright);

  \draw[rounded corners=0.1cm, ->] (hwtop) to ([xshift=-0.53cm]f5.north west |- hwtop) to (f5.north west);
  \draw[line width=0.8pt, rounded corners=0.1cm, shorten >=0.2cm] ([yshift=+0.2pt]hwtop) to ([yshift=+0.2pt,xshift=-0.53cm]f5.north west |- hwtop) to (f5.north west);
  \draw[rounded corners=0.2cm] (f5.north) to ([yshift=-0.4pt]f5.north |- hwtop) to ++(1.5cm,0) coordinate (hwtop);

  \draw[rounded corners=0.1cm, ->] ([yshift=0.4pt]hwtop) to ([yshift=0.4pt,xshift=-0.53cm]f6.north west |- hwtop) to (f6.north west);
  \draw[line width=1.2pt, rounded corners=0.1cm, shorten >=0.2cm] ([yshift=0.4pt]hwtop) to ([yshift=0.4pt,xshift=-0.53cm]f6.north west |- hwtop) to (f6.north west);

  \node[right=2.0 of e6, yshift=-0.3cm] {$\dots$};
  \clip (-0.8,-0.6) rectangle + (24.25,3.8);
  \node[latent, right=1 of t6, yshift=-1.0cm] (u7) {$U_7$};
  \edge {u6, t6} {u7}
  \draw[->] (f6) to [bend left=40] (u7);
\end{tikzpicture}
}
\caption{The blocked experiment design $\beta([2, 4, \dots])$.
The individual blocks are delta designs (\Cref{sec:delta-experiments}) linked via their last and first uncontrolled states.
Run counts of programs within blocks are balanced.
The balancing couples $F_i$ and $F_j$ in the same block, but those in different blocks are independent.}
\label{fig:block-design}
\end{figure*}
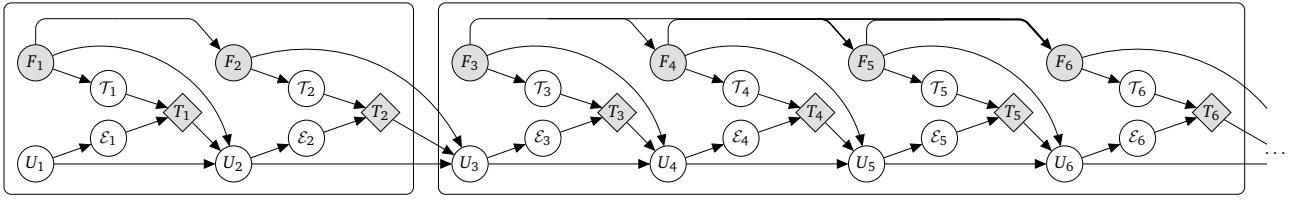

\begin{definition}[$\Delta$ shorthands]
\label{def:delta-shorthands}
Let $\fC g \in \sFC$ $f \neq g$ and $c \in \sC$.
We introduce the following shorthands:
\begin{gather*}
\DIvec = \indic_{\vectrv{F}=f} - \indic_{\vectrv{F}=g},\\
N_f = \sum_{i=1}^n \indic_{F_i=f}, \qquad N_g = \sum_{i=1}^n \indic_{F_i=g},\\
\mc = \max\set[\big]{\norm[\big]{T \mid c}_\infty, \norm[\big]{\Tau \mid c}_\infty, \norm[\big]{\Epsilon \mid c}_\infty}.
\end{gather*}
\end{definition}

\begin{remark}[Shift invariance]
Note that $\abs{T}$ is bounded in delta designs, which implies the same for $\abs{\Tau}$ and $\abs{\Epsilon}$.
This can be easily seen by considering $\Tausub_1$ and $\Epsilonsub_1$, which are independent.
Hence, $\cramped{\mc < \infty}$.
Furthermore, $l(\fc)$ is unidentifiable under the additive noise assumption since the shifted $\Tau + \ct{C}$ and $\Epsilon - \ct{C}$ lead to the same observable behaviour for any $\ct{C} \in \sR$.
With this in mind, we could even assume that $\ct{C}$ is chosen to minimize $\max\set{\norm{\Tau + \ct{C} \given c}_\infty, \norm{\Epsilon - \ct{C} \given c}_\infty}$, in which case it can be shown that $\mc = \norm{T \given c}_\infty$.
\end{remark}

\begin{remark}[Independent vs uncorrelated]
If the run times were i.i.d. (see e.g. \citet{georges2007statistically}), their mean would be approximately normal due to the CLT.
Subject to the validity of this approximation, a frequentist CI based on $\DE^c$ would then be the same as the Bayesian credible interval for an appropriately chosen prior.
However, the $\DInorm_i T_i$ terms in $\DE^c$, are at best pairwise uncorrelated, not mutually independent, so the CLT does not apply \citep{svante1988some}.
\end{remark}

With the framework of delta experiments in place, we now present a block-based experiment design that instantiates this model.

\section{Blocked Experiments}
\label{sec:blocked-experiments}

A natural approach is to collect an equal number of measurements for all programs.
Under independent sampling, this is impractical because the difference in run counts $N_f - N_g$ behaves as a simple random walk, for which the expected return time to zero is infinite [Feller 1968].
Instead, the standard solution relies on sampling without replacement.

\begin{definition}[Blocking]
\label{def:blocking}
A \emph{blocking} $\vect{b}$ is a sequence of natural numbers that are divisible by $\ct{F}$ such that $1 \leq \len{\vect{b}} < \infty$.
We indicate block-relative indexing with a superscript.
In the context of a blocking $\vect{b}$, we define $\vect{x}^k = \vect{x}_{p:q}$, where $\vect{x}$ is a sequence, $p = \sum_{i=1}^{k-1} b_i + 1$ and $q = \sum_{i=1}^k b_i$.
\end{definition}

Next, we define blocked experiments, which have blockwise independent $\vectrv{F}$ and equal run counts within each block for all programs.

\begin{definition}[Block design, see \Cref{fig:block-design}]
\label{def:block-design}
A blocked experiment design for a given blocking $\vect{b}$ is denoted by $\beta(\vect{b})$.
In the single-block case, $\vect{b} = [b]$ and $\beta([b])$ (abbreviated as $\beta(b)$) is a delta design where $F_i$ are sampled without replacement from a multiset of $\frac{b}{\ct{F}}$ elements of each $f \in \sFP$
In the multi-block case, $\beta(\vect{b})$ is the ``concatenation'' of the single-block experiment designs corresponding to the elements of $\vect{b}$, linked via their last and first uncontrolled states.
\end{definition}

In the rest of this section, we establish the upper consistency of blocked experiments by bounding the first two moments of $\DE^c(\fg)$ and appealing to Cantelli's inequality (the one-sided Chebyshev inequality).
This is sufficient to prove that the asymptotic rate is optimal, but the resulting bound is loose.
For the finite case, we thus approximate the worst-case behaviour of the noise with dynamic programming, which allows direct comparison with simple randomized experiments, presented later in \Cref{sec:alpha-experiments}.

\subsection{Single-Block Experiments}
\label{sec:single-block-experiments}

Here, we state the theorems about the contraction rates of the variance and the bias of the $\DE$ estimator in single-block designs.
For the proofs, see \Cref{sec:proofs-for-single-block-experiments}.

\begin{restatable}{theorem}{thmbetavariancerate}
\label{thm:beta-variance-rate}
For the $\beta(n)$ design, $\Var[\DE^c]$ is $\bigO(n^\negone)$.
\end{restatable}

\vspace{-\baselineskip}

\begin{restatable}{theorem}{thmbetabiasrate}
\label{thm:beta-bias-rate}
For the $\beta(n)$ design, $\E[\DE^c - \delta^c] = \bigO(n^{\negone/2})$.
\end{restatable}

\subsection{Multi-Block Experiments}
\label{sec:multi-block-experiments}

\begin{restatable}[Multi-block bounds]{proposition}{thmmultiblockbounds}
\label{thm:multi-block-bounds}
Let $\vect{b}$ be a blocking and $c \in \sC$.
Assume that the per-block biases and standard deviations are bounded as
\begin{gather*}
\abs[\big]{\E[\DE^c\ \given \beta(b_k)] - \DT^c} \leq B_k,
\qquad
\Std[\DE^c \given \beta(b_k)] \leq S_k
\end{gather*}
for all $k \in [1, \len{\vect{b}}]$.
Then, these bounds combine linearly as
\begin{gather*}
\abs[\big]{\E[\DE^c] - \DT^c}
\leq \smash[t]{\sum_{k=1}^{\len{\vect{b}}}} \frac{b_k}{\sum_i b_i} B_k,
\qquad
\Std[\DE^c]
\leq \smash{\sum_{k=1}^{\len{\vect{b}}} \frac{b_k}{\sum_i b_i} S_k}.
\end{gather*}
\end{restatable}

\vspace{-\baselineskip}

See \Cref{sec:proof-for-multi-block-experiments} for the proof.

\begin{remark}[Incomplete $\beta$-experiments]
A claim analogous to the above is true when the last block is incomplete.
Blocks that are yet to begin or do not have runs for both $f$ and $g$ are to be ignored.
\end{remark}

\subsection{Combining Per-Configuration Bounds}
\label{sec:combining-per-configuration-bounds}

With the bounds on the worst-case bias and variance of the per-configuration estimator $\DE^c$, we construct similar bounds for the suitewise average estimator $\DE$.
Then, we derive stochastic upper bounds for $\DE - \DT$.

\begin{restatable}[$\DE^c$ bounds to $\DE$]{proposition}{thmdectoboundedde}
\label{thm:dec-to-bounded-de}
Assume that $\abs{\E[\DE^c] - \DT^c} \leq b^c$ and $\Std[\DE^c] \leq s^c$ for all $c \in \sC$.
Then, under the benchmark suite model with partition $\sS$ of $\sC$,
\begin{gather*}
\abs[\big]{\EB{\DE} - \DT} \leq \sum_{c \in \sC} P(c) b^c\!,\\
\VarB{\DE} \leq \smash{\sum_{\sS' \in \sS} \PP[\Big]{\sum_{c \in \sS'} P(c) s^c}^2}\!.
\end{gather*}
\end{restatable}

See \Cref{sec:proofs-for-combining-per-configuration-bounds} for the proof of this and the following proposition.

\begin{restatable}[$\beta$ moment-based noise quantile]{proposition}{thmbetamomentbasedquantile}
\label{thm:beta-moment-based-quantile}
Let $b(n)$ and $s(n)$ denote the respective suitewise bounds on the absolute bias and standard deviation of $\DE$ for a blocked design $\beta(n)$, where $b(n) = \bigO(n^{\negone/2})$ and $s(n) = \bigO(n^{\negone/2})$.
For any confidence level $p \in (0, 1)$, the adversarial noise quantile $\smash[b]{Q_p^{\mathrm{adv}}}$ is bounded as
\begin{align*}
Q_p^{\mathrm{adv}} \leq b(n) + s(n) \smash[t]{\sqrt{\frac{p}{1 - p}}},
\end{align*}
which contracts at a rate of $\bigO(n^{\negone/2})$.
\end{restatable}

Given the decision rule parameters $\gamma$ and $\gammadecprime$ and a desired positive sensitivity $\omega$, we can determine the smallest number of runs $n$ by analytically or numerically solving the separation threshold equation
\begin{align*}
-S_{\gamma}(\fg) = Q_{\gamma'}^{\textrm{adv}} + Q_{\gammadecprime}^{\textrm{adv}} \leq \omega.
\end{align*}

\begin{remark}[Optimal $\beta(n)$ separation rate]
Even though the $\beta(n)$ bias rate $\bigO(n^{\negone/2})$ proven above is worse than the $\bigO(n^\negone)$ with i.i.d. noise, the standard deviation rate is $\bigO(n^{\negone/2})$ in both cases.
Consequently, the overall rate for the separation threshold, which is a positive linear combination of the standard deviation and the bias, is optimal.
\end{remark}

\begin{remark}[One-sided variance]
Although the variance of $\DE^c$ vanishes asymptotically, the variance of its constituent $\LE(\fc)$ alone may not.
This is evident when the noise $\Epsilonsub_i$ are perfectly correlated, preventing it from averaging out.
This underlines the necessity of focussing on deltas instead of absolute performance.
\end{remark}

We have bounded the bias and the standard deviation of a block design for a given blocking, which is sufficient for making decisions under uncertaintly (see \Cref{sec:decisions-under-uncertainty} and \Cref{sec:looking-ahead}) for a given number of runs $n$, but to prove consistency and to find the number of runs required for a desired level of sensitivity at a minimum effect size, we need computable bounds for all $n$.
However, there are many possible blockings for a given $n$, so we must first decide how the block sizes shall evolve as the overall size increases.
We first study the asymptotics of constant block size designs and the issues inherent therein, then we move on to explore the effect of growing block sizes.

\subsection{Asymptotics with Constant Block Size}
\label{sec:asymptotics-with-constant-block-size}

In this section, we assume that the sizes of all blocks are $b \in \sN$, and instead of $\beta([b, \dots, b])$ we write $\beta(n, b)$, where $b$ divides $n$ to distinguish this from the general case.

A quick inspection of the bounds in \Cref{thm:multi-block-bounds} reveals that they stay constant as the number of blocks increases.
This is because the sums are of constant terms, determined by the block sizes, and $\len{\vect{b}} / \sum_i b_i = b^\negone$ is constant, too.
We now demonstrate that the problem does not lie within the bounds; rather, constant block size designs are biased, even asymptotically.

\begin{example}[$\beta(n, b)$ bias via $\cramped{F_i \shortrightarrow U_{i+1}}$]
Let $\sF= \set{\fg}$, $b = 2$.
Under $\beta(n, b)$, in any block $k$, the first run $F^k_{\halfnegkern1}$ is equally likely to be preceeded by $f$ and $g$.
However, if the second run $F^k_{\halfnegkern2} = f$, then $F^k_{\halfnegkern1} = g$, so runs of $f$ are more likely overall to be after a run of $g$ than after an $\fP$
Thus, with the choices $\Epsilonsub_i = U_i$, $U_{i+1} = \indic_{F_i=f}$ and $U_1=0$, the bias of $\TA_{f}$ varies by program, which renders $\DE^c$ biased.
\end{example}

All blocks have the same bias for all $n$, so $\beta(n, b)$ is asymptotically biased.
This bias is a potentially blocking issue, and we aim to understand existing practice by reverse engineering the assumptions that may allow blocking to work.
We discuss two prominent examples of constant block size designs, paired benchmarking and random interleaving, among related works in \Cref{sec:paired-benchmarking} and \Cref{sec:google-benchmark}.

The above example relied on propagating information via $\cramped{F_i \shortrightarrow U_{i+1}}$, but as the next one shows, the bias does not go away if we remove this connection.

\begin{example}[$\beta(n, b)$ bias via $\cramped{F_i \shortrightarrow T_i \shortrightarrow U_{i+1}}$]
Consider a set $\sF$ consisting of three programs with deterministic run times: $f$ and $g$ with $1\si{s}$, and $h$ with $2\si{s}$.
Suppose that $\Epsilonsub_i = 1$ if $U_i$ is even, and $0$ otherwise.
Let $u_1=0$.
Notice that all blocks start with an even $U_i$ because they take exactly $1\si{s}+1\si{s}+2\si{s}=4\si{s}$.
Enumerating the $6$ possible orderings of programs within blocks, simple counting reveals that $f$ and $g$ are run in an even state with probability $1/2$ but $h$ with $2/3$.
Hence, $h$ experiences higher bias through $\Epsilonsub_i$ than the other two.
\end{example}

It is also easy to see that the variance may not vanish.
In our first example, consider extending the uncontrolled state to a pair, where the new component retains the original $u_1$ indefinitely, and using this $u_1$ to modulate the bias.
In the second, randomize the parity of $u_1$.

So, with our current assumptions, constant block size designs can exhibit program bias even asymptotically.
Assuming that all programs affect the uncontrolled state the same way would help eliminate the bias, but one only needs to consider how CPU temperature is affected by run time to see that this is hardly a tenable proposition.
Alternatively, we could require that the uncontrolled state disperses, and the current state $U_i$ becomes indistinguishable from the next, $U_{i+1} \given F_i = f$, which would prohibit any short-term effect on performance.

Similarly, to guarantee that the variance vanishes, the inter-block dependencies must be constrained.
But the dependencies can propagate through the uncontrolled state, so we are back to the original problem discussed in \Cref{sec:absolute-performance}: lacking detailed knowledge of the dynamics of the environment, all such assumptions are suspect.
Therefore, to achieve consistent estimation in stateful environments, alternatives to constant block size designs are necessary.

\subsubsection{Paired Benchmarking}
\label{sec:paired-benchmarking}

Acknowledging the problems posed by stateful environments, \citet{bazhenov2023} proposes paired benchmarking of programs to reduce variance.
Just as paired testing is a kind of blocking in statistics, their proposed method is a special case of $\beta$-experiments with $b=\ct{F}=2$.
They demonstrate variance reduction empirically but do not address the bias issue or the question of consistency.
In our analysis, we demonstrated the potential for bias and provided sufficient conditions for asymptotic unbiasedness and proved consistency under some conditions.
To the extent that these conditions are less realistic than those of $\delta$-experiments, we prefer the latter, although we cannot rule out that better assumptions could change this.

In the context of A\narrowslash{}B testing, paired benchmarking is similar to the time-grouped randomization of \citet{wu2022non}, who propose assigning every pair of consecutive customers to treatment and control in randomized order to reduce variance.
In their setting, i.i.d. assumptions guarantee unbiasedness.
In contrast, our starting point is that stateful environments rule out independence.

\subsubsection{Google Benchmark}
\label{sec:google-benchmark}

The \citet{googlebenchmark} has the ability to run and measure a number of functions with random interleaving\footnotemark, with claims of up to $40$\% reduction in run-to-run variance but no discussion of bias apart from random interleaving approximating a ``real workload''.

We may use this library to benchmark alternative implementations of functions against each other, which amounts to a single block experiment $\beta(n)$\footnote{Actually, in Google Benchmark, the order is not fully random as a number of runs of the same function are grouped into chunks to avoid measurement problems due to tiny durations, and these chunks are shuffled, but that still fits within $\beta$-experiments.}.
Our reservations about the validity of assumptions supporting asymptotic unbiasedness apply.

\footnotetext{\url{https://github.com/google/benchmark/issues/1051}}

On the other hand, if the functions being benchmarked serve different purposes, we are more likely to be interested in absolute rather than relative performance figures, perhaps to be compared to the previous version of the program that has these functions.
Needless to say, comparing results from different experiments reintroduces all the problems that uncontrolled state poses.

Still, a large number of unrelated functions may help with dispersion of uncontrolled state, and relative benchmarking of alternative implementations randomly interleaved with many unmeasured, unrelated functions may empirically work better than with a more homogenous approach, such as inserting random waits.

\subsection{Asymptotics with Growing Blocks}

As shown in \Cref{sec:asymptotics-with-constant-block-size}, constant block size designs can exhibit non-vanishing bias and variance.
In general, in a multi-block setting, the per-block bounds on the bias and the variance combine linearly (\Cref{thm:multi-block-bounds}), so the block size must grow exponentially in order to have a single block dominate, which guarantees consistency.

\begin{theorem}[Growing-blocks rate]
\label{thm:growing-blocks-rate}
Let $(\vect{b}_n)$ be an infinite sequence of blockings (\Cref{def:blocking}) whose average block size increases without bound, $\lim \overline{\vect{b}_n} = \infty$.
Then, for all $c \in \sC$, the bias and the standard deviation vanish:
\begin{gather*}
\lim_{n \to \infty} \E[\DE^c -\DT^c \given \beta(\vect{b}_n)] = 0,\\
\lim_{n \to \infty} \Std[\DE^c \given \beta(\vect{b}_n)] = 0.
\end{gather*}
Furthermore, the rate for both is $\bigO(n^{\negone/2})$ if the block size grows exponentially, that is, $b_{n,k(n)} \propto n$ for some $k(n)$.
\end{theorem}

The theorem is stated in terms of a sequence of blockings, for the following reasons.
A blocked design defines behaviour in a single experiment with a finite number of runs, but for the analysis of asymptotics we obviously need infinitely many runs.
With a given constant block size, the extension to infinity is uniquely determined, but in the general case it is not.
One possibility is to work with a single blocking with infinitely many blocks, corresponding to a single experiment.
Here, we instead consider an infinite sequence of finite blockings, which subsumes single experiments as the special case where subsequent blockings can only extend previous ones with new blocks.
In addition, this approach is applicable to repeated experiments where, for example, there is a single block of increasing size.

\subsection{Dynamic Programming for $\beta$-Designs}
\label{sec:beta-dynamic-programming}

For a $\beta$-design, we compute the worst-case \emph{success probability} $P(\EAsub_f - \EAsub_g \leq \omega)$ for given $n \in \sN$ and $\omega > 0$.
This cumulative distribution function is the inverse of the quantile function $Q^{\mathrm{adv}}_p$; thus, we can find, for example by binary search, the minimum number of runs $n$ required to meet any given decision threshold $Q_{\gammadecprime}(g, f)$ at decision confidence $\gammadecprime$ (\Cref{rem:equivalence-thresholding}).
We can similarly solve numerically for $n$ to meet a given unambiguous region threshold $S_\gamma = - Q_{\gamma'}(\fg) - Q_{\gammadecprime}(g, f)$ (\Cref{thm:unambiguous-region}), with decision confidence $\gammadec$ and detection confidence $\gamma$.
In the common case where $\gamma = \gammadec$, this simplifies to finding the smallest $n$ such that $P(\EAsub_f - \EAsub_g \leq -S_\gamma / 2) \geq \gamma$.

The calculation is performed with dynamic programming.
After $r$ runs, the state is the quadruple $(\sigma, r, n_f, n_g)$, where
\begin{align*}
\sigma = \sum_{i=1}^r \di_i \epsilon_i,
\qquad n_f = \sum_{i=1}^r \indic_{f_i=f}\lshh,
\qquad n_g = \sum_{i=1}^r \indic_{f_i=g}\lshh.
\end{align*}
To keep the action space tractable, the noise is quantized to the set $\mathcal{Q} = \{0, \pm 1 / \ct{Q}, \pm 2 / \ct{Q}, \ldots, \pm 1\}$, where $\ct{Q} \in \sN$ is the number of quanta.
By increasing $\ct{Q}$, we can approximate the continuous case, where the noise can take any value in $[\negone, 1]$.
In \Cref{sec:experiments}, we will demonstrate experimentally that this approximation is very tight even with a moderate number of quanta.

Conveniently, due to balancedness and our worst-case focus, runs of programs other than $f$ and $g$ can be ignored, and we can assume that $\ct{F} = 2$ if $n$ is scaled by $2 / \ct{F}$.
Without loss of generality, we can also assume that $\mc = 1$ as $\omega$ can be scaled by $1/\mc$ otherwise.

The value of a state $s$ is the success probability $V(s) = P(\EAsub_f - \EAsub_g \leq \omega \given s)$.
In terminal states ($r=n$), where $\easub_f - \easub_g = \sigma \, \ct{F} / n$, the success probabilities are either $0$ or $1$.
For $r \in [1,n)$, the value update is
\begin{align*}
&V(\sigma, r, n_f, n_g)\\
&          \qquad = \min\limits_{q \in \mathcal{Q}} \Pigl(  \frac{n / 2 - n_f}{n - r} \; V(\sigma + q, r + 1, n_f + 1, n_g)\\
&\hphantom{\qquad = \min\limits_{q \in \mathcal{Q}} \biggl(} + \frac{n / 2 - n_g}{n - r} \; V(\sigma - q, r + 1, n_f, n_g + 1)\Pigr),
\end{align*}
where $n - r$, $n/2 - n_f$ and $n/2 - n_g$ are the number of runs left in total, for $f$ and for $g$, respectively.

The size of the state space is $\bigO(n^3 \ct{Q})$ because the range of $\sigma$ grows linearly with $r$, and $n_g = r - n_f$ is fully determined.
While this means that scalability is a problem, we will only use the DP approximation to compare the blocked design's efficiency to that of the simple randomized design, which we explore next.

\begin{assumption}
\label{ass:no-program-noise-in-dp}
Because the distribution of the program-specific noise $\Tausub_i$ is unknown, it cannot be marginalized out in the value update.
Furthermore, it cannot be lumped together with the environmental noise because, unlike $\Epsilonsub_i$, we have $\Tausub_i \centernot{\indep} F_i \given \divec_{< i}, \sigma_{i-1}$.
Thus, the DP calculation assumes no program-specific noise.
\end{assumption}

\section{Simple Randomized Experiments}
\label{sec:alpha-experiments}

Our first experiment design chooses a program independently and uniformly at random for each run.
In classical experiment design, this is called simple randomized design, where treatment assignments are sampled with replacement from the set of treatments ($\sFC$ here).

\begin{definition}[Simple randomized design]
\label{def:alpha-design}
A simple randomized experiment design (or $\alpha$-design) is a delta design where, for every run $i \in [1, n]$, the program $F_i$ is chosen uniformly and independently from the set of programs $\mathcal{F}$.
\end{definition}

To avoid division by zero in $\LE(\fc) = \TA_{f}$, we must also require $N_f > 0$, which introduces slight dependencies in $\vectrv{F}$\mathendkern{F}{.}

\subsection{Asymptotics}

\begin{definition}[Shorthands]
Let $L_n = \sum_{i=1}^n \DI_i$, $E_n = \sum_{i=1}^n \Epsilonsub_i$ and $\Sigma_n = \sum_{i=1}^n \DI_i \Epsilonsub_i$.
\end{definition}

\vspace{-\baselineskip}

\begin{restatable}{proposition}{thmalphanoisedifference}
\label{thm:alpha-noise-difference}
Under the $\alpha$-design with $\ct{F} = 2$, we can express the noise difference as
\begin{align*}
\EAsub_f - \EAsub_g
= \frac{2 n \Sigma_n - 2 L_n E_n}{n^2 - L_n^2}.
\end{align*}
\end{restatable}
\noindent For the proof of this and other propositions in this section, see \Cref{sec:proofs-for-simple-randomized-experiments}.

\begin{restatable}[$\alpha$ noise $\sqrt{n}$-consistency]{theorem}{thmalphanoiserate}
\label{thm:alpha-noise-rate}
For any configuration $c \in \sC$ and programs $\fg \in \sF$, with the $\alpha$-design, $\EAsub_f - \EAsub_g = \bigO_p(n^{\negone/2})$ given $\vect{N} > 0$.
\end{restatable}

Note that the theorem's claim is about the worst-case: no matter how $\Epsilonsub_i$ are chosen, the effect of the possibly auto-correlated environmental noise on averages decays at the same asymptotic rate as in the independent case, which implies that this worst-case rate is also the best possible.

\begin{restatable}[$\alpha$ $\sqrt{n}$-consistency]{theorem}{thmalphasqrtnconsistency}
For any configuration $c \in \sC$ and programs $\fg \in \sF$, with the $\alpha$-design, $\DE^c - \delta^c = \bigO_p(n^{\negone/2})$ given $\vect{N} > 0$.
\end{restatable}

\subsection{Finite-Sample Behaviour}

Having established the asymptotic rate, we want to know how many runs to perform for a given level of confidence $\gamma$ and minimum effect size $\omega$.
We propose three approaches for the $\ct{F} = 2$ case, which we finally combine into a bound for the general case.

\subsubsection{Dynamic Programming for $\alpha$-Designs}
\label{sec:alpha-dynamic-programming}

For an $\alpha$-design, we compute the worst-case \emph{success probability} $P(\EAsub_f - \EAsub_g \leq \omega)$ for given $n \in \sN$ and $\omega > 0$.
This parallels the DP approximation for the $\beta$-design (\Cref{sec:beta-dynamic-programming}), with the following notable differences:
\begin{itemize}
\item The state space is larger as we need to track partial sums and run counts for $f$ and $g$ separately.
\item We cannot ignore programs other than $f$ and $g$ as their run counts are stochastic; to keep the state space size manageable, we restrict our analysis to $\ct{F} = 2$.
\end{itemize}
To handle more programs efficiently, we present a martingale-based approximation to this DP (\Cref{app:martingale-approximation-for-several-programs}) in the next section.

After $r$ runs, the state is the quintuple $(\sigma_f, \sigma_g, r, n_f, n_g)$, where
\begin{align*}
&\sigma_f = \smash[t]{\sum_{i=1}^r} \indic_{f_i=f} \epsilon_i,
&\sigma_g = \smash[t]{\sum_{i=1}^r} \indic_{f_i=g} \epsilon_i\\
&n_f = \sum_{i=1}^r \indic_{f_i=f},
&n_g = \sum_{i=1}^r \indic_{f_i=g}.
\end{align*}

In terminal states ($r=n$), if $n_f > 0$ and $n_g > 0$, the value is $\indic_{\sigma_f / n_f - \sigma_g / n_g \leq \omega}$; otherwise, it is $0$ (undefined means constitute failure).
For $r \in [1,n)$, the value update is
\begin{align*}
&V(\sigma_f, \sigma_g, r, n_f, n_g)\\
&          \qquad = \min\limits_{q \in \mathcal{Q}} \bigl(  0.5 \, V(\sigma_f + q, \sigma_g, r + 1, n_f + 1, n_g)\\
&\hphantom{\qquad = \min\limits_{q \in \mathcal{Q}} \bigl(} + 0.5 \, V(\sigma_f, \sigma_g + q, r + 1, n_f, n_g + 1)\bigr).
\end{align*}

The size of the state space is $\bigO(n^4 \ct{Q}^2)$ because we must track two independent accumulated noise sums ($\sigma_f, \sigma_g$) alongside a $n_f$.
Note that $n_g = r - n_f$ because $\ct{F} = 2$.
This scales significantly worse than the $\beta$-design DP.
In the following, we propose cheap approximations and verify them against these dynamic programming results in the computationally feasible range.

\subsubsection{Approximations}
\label{sec:approximations}

To motivate the approximations developed later in this section, we first present a theorem characterizing the asymptotic behaviour of the estimator.
Given the illustrative purpose of the theorem and the heuristic nature of the approximations it motivates, it adopts $\bar{\mathcal{E}} \indep \DIvec$ as a simplifying assumption that greatly reduces the complexity of the exposition.
While the average noise $\bar{\Epsilon}$ is clearly not independent from the coin flips $\DIvec$ under our model, it can be viewed as a weak forgetting assumption.

\begin{restatable}[$\alpha$-approximation]{proposition}{thmalphaapproximation}
\label{thm:alpha-approximation}
Under an $\alpha$-design with $\ct{F} = 2$ and $n$ runs, let $\tilde{M} = \frac{2}{n} \Sigma_n$ and $\tilde{C} = \frac{2}{n} L_n \bar{\mathcal{E}}$.
Assume worst-case noise and that the mean noise is independent of the program choices, $\bar{\mathcal{E}} \indep \DIvec$.
Then,
\begin{gather*}
\EAsub_f - \EAsub_g = \tilde{M} - \tilde{C} + \bigO_p(n^{\neg 3/2}),\\
\Cov[\tilde{M}, \tilde{C}] = \Var[\tilde{C}],
\end{gather*}
and consequently for the variance
\begin{gather*}
\Var[\EAsub_f - \EAsub_g] \leq \Var[\tilde{M}] + \bigO(n^{\neg 2}).
\end{gather*}
\end{restatable}

We now present two practical bounds for experimental planning with the $\alpha$-design in the two-program case.

\begin{restatable}[Martingale approximation]{approximation}{approxmartingale}
\label{thm:martingale-approximation}
Under an $\alpha$-design with $\ct{F} = 2$, $n$ runs, with any $\omega > 0$,
\begin{align*}
P(\EAsub_f - \EAsub_g \leq \omega) \gtrsim 1 - \exp\left(\frac{-n \omega^2}{8 \mc^2}\right).
\end{align*}
\end{restatable}

See \Cref{sec:proofs-for-simple-randomized-experiments} for the derivations of this and the following approximations.

\begin{restatable}[Asymmetric approximation]{approximation}{approxasymmetric}
\label{thm:asymmetric-approximation}
Under an $\alpha$-design with $\ct{F} = 2$, $n$ runs, with any $\omega > 0$,
\begin{align*}
P(\EAsub_f - \EAsub_g \leq \omega) \gtrsim \biggl[2 \Phi\biggl(\frac{\sqrt{n}}{2\mc} \omega\biggr) - 1\biggr]^2.
\end{align*}
\end{restatable}

Comparison with the dynamic programming solution suggests that the asymmetric approximation (\Cref{thm:asymmetric-approximation}) is better, so that is our recommendation for the two-program case.
However, the martingale approximation (\Cref{thm:martingale-approximation}) admits a convenient closed-form generalization to more programs, as we show next.

\begin{restatable}[Martingale approximation for $\ct{F} > 2$]{approximation}{approxforseveralprograms}
\label{app:martingale-approximation-for-several-programs}
Under an $\alpha$-design with $n$ runs and any $\omega > 0$, provided that the martingale approximation (\Cref{thm:martingale-approximation}) holds,
\begin{align*}
P(\EAsub_f - \EAsub_g \leq \omega) \gtrsim 1 - \bigl(p e^{\neg \ct{K}} + 1 - p\bigr)^n,
\end{align*}
where $\ct{K} = \omega^2/(8\mc^2)$ and $p = 2/\ct{F}$.
\end{restatable}

This provides a computationally trivial $\bigO(1)$ approximation.

\begin{figure}[p]
\centering
\begin{tikzpicture}
  \begin{axis} [ylabel=$\gamma$, xlabel=$n$,
      xlabel near ticks,
      ylabel near ticks,
      xlabel shift={-2pt},
      ylabel shift={-2pt},
      xmin=2,
      ymin=0, ymax=1,
      legend pos=south east,
      legend style={nodes={scale=0.7, transform shape}},
      legend cell align={left},
      height=0.6*\columnwidth,
      width=0.98\columnwidth,
    ]
    \pgfplotstableread{data/threshold-0.5.tbl}{\sorted}
    \addplot[adp1] table [x=n, y=adp1] {\sorted};
    \addlegendentry{$\alpha$ DP1};
    \addplot[adp5] table [x=n, y=adp5] {\sorted};
    \addlegendentry{$\alpha$ DP5};
    \addplot[adp10] table [x=n, y=adp10] {\sorted};
    \addlegendentry{$\alpha$ DP10};
    \addplot[bdp1] table [x=n, y=bdp1] {\sorted};
    \addlegendentry{$\beta$ DP1};
    \addplot[bdp5] table [x=n, y=bdp5] {\sorted};
    \addlegendentry{$\beta$ DP5};
    \addplot[bdp10] table [x=n, y=bdp10] {\sorted};
    \addlegendentry{$\beta$ DP10};
  \end{axis}
\end{tikzpicture}
\caption{DP estimates of the worst-case confidence $\gamma$ with the detection threshold $\omega$ set to $0.5\mc$ for the $\alpha$ (\Cref{sec:alpha-dynamic-programming}) and $\beta$-designs (\Cref{sec:beta-dynamic-programming}) with $1$, $5$ and $10$ quanta.
The sequence of quantized DP estimates decreases to the worst case confidence from above.
Since the change from $5$ to $10$ is already quite small, we substitute the DP10 versions for the worst-case confidences.}
\label{fig:quantization-at-threshold-0.5}
\end{figure}
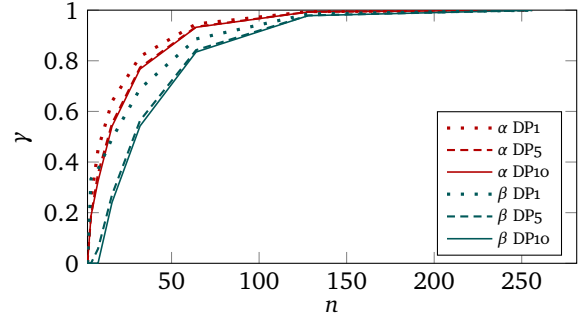

\begin{figure}[p]
\centering
\begin{tikzpicture}
  \begin{axis} [ylabel=$\gamma$, xlabel=$n$,
      xlabel near ticks,
      ylabel near ticks,
      xlabel shift={-2pt},
      ylabel shift={-2pt},
      xmin=2,
      ymin=0, ymax=1,
      legend pos=south east,
      legend style={nodes={scale=0.7, transform shape}},
      legend cell align={left},
      height=0.6*\columnwidth,
      width=0.98\columnwidth,
    ]
    \pgfplotstableread{data/threshold-0.25.tbl}{\sorted}
    \addplot[adp1] table [x=n, y=adp1] {\sorted};
    \addlegendentry{$\alpha$ DP1};
    \addplot[adp5] table [x=n, y=adp5] {\sorted};
    \addlegendentry{$\alpha$ DP5};
    \addplot[adp10] table [x=n, y=adp10] {\sorted};
    \addlegendentry{$\alpha$ DP10};
    \addplot[bdp1] table [x=n, y=bdp1] {\sorted};
    \addlegendentry{$\beta$ DP1};
    \addplot[bdp5] table [x=n, y=bdp5] {\sorted};
    \addlegendentry{$\beta$ DP5};
    \addplot[bdp10] table [x=n, y=bdp10] {\sorted};
    \addlegendentry{$\beta$ DP10};
  \end{axis}
\end{tikzpicture}
\caption{Dynamic programming estimates of the worst-case confidence $\gamma$ with the detection threshold $\omega$ set to $0.25\mc$.
As the threshold is halved compared to the previous figure, we see the graph being stretched $4$ times along the $n$ axis.
This stretching is approximate due to the discrete nature of the problem.}
\label{fig:quantization-at-threshold-0.25}
\end{figure}
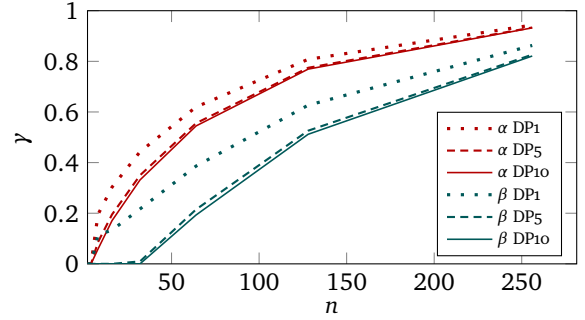

\begin{figure}[p]
\centering
\begin{tikzpicture}
  \begin{axis} [ylabel=$\gamma$, xlabel=$n$,
      xlabel near ticks,
      ylabel near ticks,
      xlabel shift={-2pt},
      ylabel shift={-2pt},
      xmin=2,
      ymin=0, ymax=0.6,
      legend pos=north west,
      legend style={nodes={scale=0.7, transform shape}},
      legend cell align={left},
      height=0.6*\columnwidth,
      width=0.98\columnwidth,
    ]
    \pgfplotstableread{data/threshold-0.125.tbl}{\sorted}
    \addplot[adp1] table [x=n, y=adp1] {\sorted};
    \addlegendentry{$\alpha$ DP1};
    \addplot[adp5] table [x=n, y=adp5] {\sorted};
    \addlegendentry{$\alpha$ DP5};
    \addplot[adp10] table [x=n, y=adp10] {\sorted};
    \addlegendentry{$\alpha$ DP10};
    \addplot[bdp1] table [x=n, y=bdp1] {\sorted};
    \addlegendentry{$\beta$ DP1};
    \addplot[bdp5] table [x=n, y=bdp5] {\sorted};
    \addlegendentry{$\beta$ DP5};
    \addplot[bdp10] table [x=n, y=bdp10] {\sorted};
    \addlegendentry{$\beta$ DP10};
  \end{axis}
\end{tikzpicture}
\caption{Dynamic programming estimates of the worst-case confidence $\gamma$ with the detection threshold $\omega$ set to $0.125\mc$.
There is similar stretching of the graph as before.}
\label{fig:quantization-at-threshold-0.125}
\end{figure}
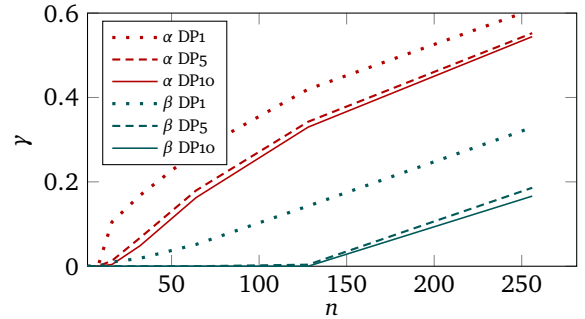

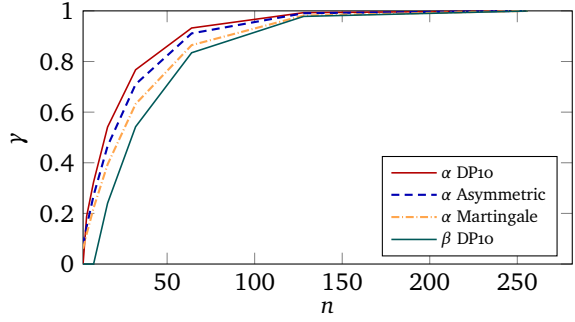
\begin{figure}[p]
\centering
\begin{tikzpicture}
  \begin{axis} [ylabel=$\gamma$, xlabel=$n$,
      xlabel near ticks,
      ylabel near ticks,
      xlabel shift={-2pt},
      ylabel shift={-2pt},
      xmin=2,
      ymin=0, ymax=1,
      legend pos=south east,
      legend style={nodes={scale=0.7, transform shape}},
      legend cell align={left},
      height=0.6*\columnwidth,
      width=0.98\columnwidth,
    ]
    \pgfplotstableread{data/threshold-0.5.tbl}{\sorted}
    \addplot[adp10] table [x=n, y=adp10] {\sorted};
    \addlegendentry{$\alpha$ DP10};
    \addplot[asymmetric] table [x=n, y=asymm.] {\sorted};
    \addlegendentry{$\alpha$ Asymmetric};
    \addplot[martingale] table [x=n, y=mart.] {\sorted};
    \addlegendentry{$\alpha$ Martingale};
    \addplot[bdp10] table [x=n, y=bdp10] {\sorted};
    \addlegendentry{$\beta$ DP10};
  \end{axis}
\end{tikzpicture}
\caption{Estimates of the worst-case confidence $\gamma$ with the detection threshold $\omega$ set to $0.5\mc$.
The dynamic programming approximations, $\alpha$ DP10 (\Cref{sec:alpha-dynamic-programming}) and $\beta$ DP10 (\Cref{sec:beta-dynamic-programming}), are with $10$ quanta ($\ct{Q}=10$).
The Asymmetric and the Martingale approximations for the $\alpha$-design are discussed in \Cref{sec:approximations}.}
\label{fig:threshold-0.5}
\end{figure}

\begin{figure}[p]
\centering
\begin{tikzpicture}
  \begin{axis} [ylabel=$\gamma$, xlabel=$n$,
      xlabel near ticks,
      ylabel near ticks,
      xlabel shift={-2pt},
      ylabel shift={-2pt},
      xmin=2,
      ymin=0, ymax=1,
      legend pos=south east,
      legend style={nodes={scale=0.7, transform shape}},
      legend cell align={left},
      height=0.6*\columnwidth,
      width=0.98\columnwidth,
    ]
    \pgfplotstableread{data/threshold-0.25.tbl}{\sorted}
    \addplot[adp10] table [x=n, y=adp10] {\sorted};
    \addlegendentry{$\alpha$ DP10};
    \addplot[asymmetric] table [x=n, y=asymm.] {\sorted};
    \addlegendentry{$\alpha$ Asymmetric};
    \addplot[martingale] table [x=n, y=mart.] {\sorted};
    \addlegendentry{$\alpha$ Martingale};
    \addplot[bdp10] table [x=n, y=bdp10] {\sorted};
    \addlegendentry{$\beta$ DP10};
  \end{axis}
\end{tikzpicture}
\caption{With the lower detection threshold of $\omega = 0.25\mc$, it becomes apparent that the $\beta$-design needs a significant number of samples to achieve non-zero confidence.}
\label{fig:threshold-0.25}
\end{figure}
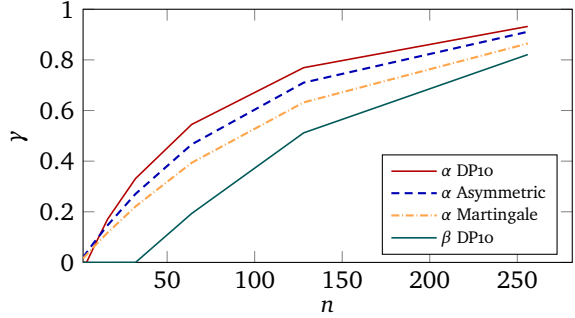

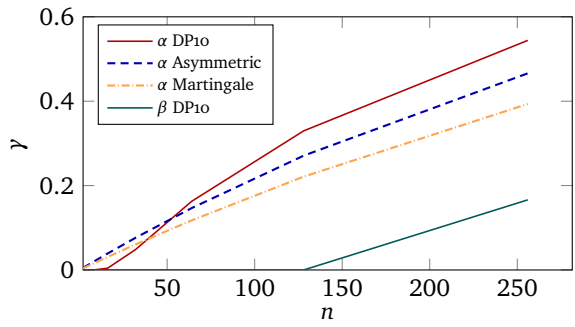
\begin{figure}[p]
\centering
\begin{tikzpicture}
  \begin{axis} [ylabel=$\gamma$, xlabel=$n$,
      xlabel near ticks,
      ylabel near ticks,
      xlabel shift={-2pt},
      ylabel shift={-2pt},
      xmin=2,
      ymin=0, ymax=0.6,
      legend pos=north west,
      legend style={nodes={scale=0.7, transform shape}},
      legend cell align={left},
      height=0.6*\columnwidth,
      width=0.98\columnwidth,
    ]
    \pgfplotstableread{data/threshold-0.125.tbl}{\sorted}
    \addplot[adp10] table [x=n, y=adp10] {\sorted};
    \addlegendentry{$\alpha$ DP10};
    \addplot[asymmetric] table [x=n, y=asymm.] {\sorted};
    \addlegendentry{$\alpha$ Asymmetric};
    \addplot[martingale] table [x=n, y=mart.] {\sorted};
    \addlegendentry{$\alpha$ Martingale};
    \addplot[bdp10] table [x=n, y=bdp10] {\sorted};
    \addlegendentry{$\beta$ DP10};
  \end{axis}
\end{tikzpicture}
\caption{We halve the detection threshold in \Cref{fig:threshold-0.25} to $\omega = 0.125\mc$, roughly stretching the graph four times along the $n$ axis.
We can see the $\alpha$ approximations become lower bounds for large enough $n$, where the simplifying assumption $\bar{\Epsilon} \indep \DIvec$ made in \Cref{sec:approximations} begins to hold to a sufficient degree.
This crossover point is at low confidence levels, which are not important for the benchmarking decision problem.}
\label{fig:threshold-0.125}
\end{figure}

\section{Experiments}
\label{sec:experimental-results}

In this section, we empirically evaluate the worst-case confidence $\gamma$ of our experimental designs under adversarial noise, assuming $\ct{F}=2$.
We utilize the dynamic programming (DP) formulation to compute numerical bounds and validate the analytical approximations derived in \Cref{sec:approximations}.

\subsection{Impact of Quantization}

Because the dynamic programming approach relies on discretizing the state space, it inherently provides an upper bound on the true worst-case confidence.
In \Crefrange{fig:quantization-at-threshold-0.5}{fig:quantization-at-threshold-0.125}, we verify the effect of the number of quanta on the tightness of this bound.
We compute the DP estimates for the $\alpha$-design (\Cref{sec:alpha-dynamic-programming}) and the $\beta$-design (\Cref{sec:beta-dynamic-programming}) using 1, 5, and 10 quanta.

We expect the sequence of quantized DP estimates to decrease towards the true worst-case confidence.
However, the marginal improvement diminishes rapidly; the difference between the bounds computed with 5 and 10 quanta is already negligible across all tested sample sizes $n$.
This justifies the use of the 10-quanta estimate (DP10) as a highly accurate proxy for the true worst-case confidence in the remainder of our experiments.
Furthermore, across these three figures, as the detection threshold $\omega$ is progressively halved from $0.5\mc$ to $0.125\mc$, the graphs stretch approximately by a factor of 4 along the horizontal axis, empirically confirming the $\bigO(n^{\negone/2})$ scaling relationship.

\subsection{Validation of Analytical Approximations}

Next, in \Crefrange{fig:threshold-0.5}{fig:threshold-0.125}, we empirically validate the Asymmetric and the Martingale approximations at various fixed detection thresholds.
We benchmark these analytical bounds against the DP10 estimates.

We consistently observe that the Asymmetric approximation provides a much tighter bound than the Martingale approximation.
At higher sample sizes, where the simplifying assumption $\bar{\Epsilon} \indep \DIvec$ begins to hold sufficiently, the analytical approximations cross over to become true lower bounds.
However, as demonstrated in \Cref{fig:threshold-0.125}, this crossover point occurs at relatively low confidence levels, which are generally irrelevant for rigorous benchmarking decisions.

These figures also highlight the finite-sample vulnerabilities of the $\beta$-design.
Particularly at lower detection thresholds (e.g., $\omega = 0.25\mc$ in \Cref{fig:threshold-0.25}), the $\beta$-design requires a substantially larger number of samples to achieve strictly positive confidence compared to the $\alpha$-design.

\subsection{Asymptotic Behavior at Constant $K$}

Finally, to isolate the asymptotic convergence from the absolute threshold scaling, we evaluate the designs while holding the scale-invariant parameter $K = \omega \sqrt{n} / \mc$ constant.
In \Crefrange{fig:k-1}{fig:k-6} (see \Cref{sec:figures-for-experimental-results}), we plot the confidence for $K \in \{1, 2, \dots, 6\}$.
By definition, the Asymmetric and Martingale approximations remain perfectly constant for a given $K$.
In contrast, the true worst-case confidences (captured by the DP10 curves) are asymptotically $\bigO_p(\sqrt{n})$ and slowly flatten out towards the analytical limits as $n \to \infty$.

This constant-$K$ perspective clearly illustrates how exploitability depends ont the sample size.
For instance, at $K=2$ (\Cref{fig:k-2}) and $n=2$, the limit can only be exceeded if the adversary successfully guesses the order of the two single runs, capping the confidence at 0.5.
As $n$ increases while $K$ remains constant, the $\alpha$-design successfully stabilizes around this level because it leaks minimal useful information to the adversary.
Conversely, the $\beta$-design's sampling without replacement provides exploitable information, and its confidence decreases.

Ultimately, these plots highlight how quickly real performance is dominated by the asymptotic rate.
Because the analytical approximations provide reliable lower bounds at all but the smallest sample sizes and lowest confidence levels, relying on the DP approximation is necessary only in its computationally feasible range.

\subsection{Number of Runs Required}

From \Cref{fig:threshold-0.5}, we can see that the $\alpha$-design achieves detection threshold $0.5\mc$ with confidence 0.9 at $n=60$, 30 runs for each program.
From this point on, we can thus expect the number of runs required to approximately quadruple for each halving the detection threshold.
The single-block $\beta$-design requires significantly more runs to achieve the same confidence.

\section{Further Considerations}
\label{sec:further-considerations}

\subsection{Meanwhile in Classical Experiment Design}
\label{sec:meanwhile-in-experimental-design}

Our approach operates within the framework of randomized controlled experiments, which traditionally rely on the Stable Unit Treatment Value Assumption (SUTVA, \citet{cox1958planning,angrist1996identification}) to identify the Average Treatment Effect (ATE).
In our context, the ATE corresponds to $\DT(\fgc)$.
It is instructive to contrast our assumptions with SUTVA based on the definition of the experimental ``unit.''
\begin{itemize}
\item If we view the \emph{individual run} as the unit, SUTVA posits \emph{non-interference}: the outcome of run $i$ depends only on its own treatment ($f_i$) and not on the history ($\vect{f}_{<i}$).
In contrast, we explicitly acknowledge interference but impose an \emph{additivity} assumption, requiring the performance differential between programs to be constant while the bias introduced by the system state varies.

\item Viewing the \emph{configuration} as the unit (analogous to a patient), our framework resembles a \emph{repeated-measures design}.
Unlike in clinical trials, our units are infinitely reusable: the same configuration can be run multiple times.
Here, our assumption corresponds to the SUTVA \emph{consistency} requirement: the expected effect is stable across repeated applications of the same treatment.

\item If we view the \emph{uncontrolled state} as the unit, our designs share the vulnerabilities of A\narrowslash{}B testing in dynamic environments.
Because the system state evolves based on previous treatments, the experiment alters the state trajectory, rendering the sample non-representative.
In observational study methodology, such \emph{covariate shift} is typically addressed via re-weighting or post-stratification.
However, because the state in modern hardware is largely unobservable, these corrections are impossible.
Instead, our designs leverage symmetry to ensure that these state-dependent biases cancel out asymptotically, obviating the need to model the underlying distribution of strata.
\end{itemize}

\subsection{Concurrent Execution}

In our measurement model, programs run sequentially in a given configuration.
We can generalize this concurrent execution on multiple identical machines if we assume that the programs run on one have no effect on the others.
In this case, the measurements gathered can be simply averaged as if they came from a single machine.
However, if execution units have more tightly coupled states (e.g. a single multi-core machine), the interactions between concurrently running programs can depend on their alignment in time.
This can easily violate the crucial assumption that the uncontrolled state affects all programs the same way.

\subsection{The Single-Machine Case}
\label{sec:the-single-machine-case}

The decision problem $\argmin_{f \in \sF} l(f)$ is invariant to positive affine transformations of the loss (scaling and shifting).
\Cref{ass:additive-global-effect} fits this strictly, as the uncontrolled state introduces only an additive bias to the run time.
\Cref{ass:multiplicative-global-effect} is analogous, but it biases the logarithm of the run time.
Since the total loss is a weighted sum of per-configuration losses, an additive bias on the run times becomes an additive constant in the objective function, leaving the optimal choice of program $f$ unchanged.

However, simple invariance is lost if the scaling factor varies by configuration.
In general, if the multiplicative noise $\Alpha_c$ depends on $c$, the weighted sum does not reduce to a single global affine transformation.
Yet, in the specific case of a single machine (or a homogeneous cluster) -- where configurations share identical hardware and system software, differing only in the benchmark task or input -- it may be reasonable to assume the multiplicative noise is independent of the configuration.

Formally, suppose that in a $\delta$-experiment $L_i = \Alpha_i \Tausub_i + \Epsilonsub_i$, where $\Alpha_i > 0$ and $\Alpha_i \indep C$.
Then,
\begin{align*}
&\EB{L_i \given F_i = f\lsh, u_1}\\
&\qquad = \E_{C_i} \EB{L_i \given F_i = f\lsh, u_1, C = C_i}\\
&\qquad = \E_{C_i} \EB{\Alpha_i \Tausub_i + \Epsilonsub_i \given F_i = f\lsh, u_1, C = C_i}\\
&\qquad = \E_{C_i} \bigl[\EB{\Alpha_i \given u_1} \EB{\Tau \given \fc = C_i}\bigr] + \EB{\Epsilonsub_i \given u_1}\\
&\qquad = \EB{\Alpha_i \given u_1} \EB{\Tau \given f} + \EB{\Epsilonsub_i \given u_1}.
\end{align*}
It follows that $\EB{\TA_{f} \given u_1}$ is a positive affine transformation of $\EB{\Tau \given f}$, where the transformation parameters depend only on the noise $u_1$ and not on the program $f$.
Consequently, the decision problems in the scaled and unscaled domains are equivalent.

\subsection{Aggregation Over Configurations}

The preceding analysis focused on single-configuration estimates, $\LE(c,f)$.
We now address the estimation of expected performance over the full configuration space, $\sC$.
For finite and tractable $\sC$, this expectation is exact: $\sum_{c \in \sC} P(c) \LE(c,f)$.

In practice, however, $\sC$ is often effectively infinite or physically inaccessible.
More critically, the set of evaluable configurations is constrained by the manual effort required to construct them (e.g., authoring benchmarks or curating input datasets).
Consequently, standard practice restricts evaluation to a \emph{fixed} subset of configurations.
We offer further statistical interpretations of this practice:
\begin{itemize}
\item \textbf{Bias--variance tradeoff:} We accept a biased estimator (the fixed set) to eliminate the variance inherent in sampling configurations.
\item \textbf{Representative performance:} We assume that absolute performance on the fixed subset approximates the global expectation.
\item \textbf{Representative contrasts:} We assume that the \emph{performance differentials} between programs are preserved on the subset, even if absolute performance is not.
\end{itemize}

\subsection{Combining Experiments}

The results presented consider estimates $\smash{\LE(f)}$ and $\smash{\LE(g)}$ obtained from the same experiment.
In practice, we may have only a subset of $\sF$ available, as is the case in a continuous integration pipeline where all versions of a program up to the present are to be compared \citep{alcocer2015tracking, grambow2019continuous}.
We could then benchmark, for example, each new version against only the previous one, and to compare versions further apart, we would sum a chain of deltas across different experiments.
As long as the set of configurations is the same, and our assumptions hold, the deltas are additive \emph{in expectation}\footnote{Indicating different experiments with a subscript, $\E[\LE_1(f) - \LE_1(g) + \LE_2(g) - \LE_2(h)] = \E[\LE(f) - \LE(g)] + \E[\LE(g) - \LE(h)] = \E[\LE(f) - \LE(h)]$.} due to asymptotic unbiasedness, and we can plot performance changes relative to an arbitrary reference point.
The resulting estimator is consistent, but in the finite case, care should be taken because this procedure may accentuate experimental noise.
The resulting decrease in confidence can be accounted for by the Bonferroni correction (based on the chain length) in the worst case.

For plotting the performance over the course of development, we can instead benchmark all versions against the same reference point.
However, if distant versions are very different -- say, disk-based caching was added to a previously completely CPU-bound algorithm -- then the ceteris paribus assumptions may not hold.


\subsection{Differential Timing}

\citet{alexandrescu2015} mentions a benchmarking technique where the performance ratio of the baseline and the contender programs are estimated from two measurements: $t_{2nb}$, the run time of $2n$ runs of baseline, and $t_{nb+nc}$, the run time of $n$ runs of the baseline followed by $n$ runs of the contender.
Denoting the true runtimes with $b$ and $c$, the proposed formula to estimate the performance ratio is
\begin{align*}
\frac{t_{2nb}}{2t_{nb+nc} - t_{2nb}} \approx \frac{2nb}{2nb + 2nc - 2nb} = \frac{b}{c},
\end{align*}
with the idea that ``some overhead noises cancel out''.
The crucial difference from our proposal is that running the baseline $n$ times and then the contender $n$ times does not address the potential bias and non-vanishing variance caused by the uncontrolled state.

\subsection{Taking the Minimum}

Another option is to take the minimum measurement instead of the mean, based on the idea that the environment can only slow programs down.
\citet{chen2016robust} justify the minimum estimator and propose an automated scheme to find the optimal number of loop iterations (blocking) to balance timer resolution against noise probability, empirically demonstrating that this minimum estimator converges robustly to what they deem the true execution time.
However, their method implicitly relies on the noise process being memoryless and sparse.

\section{Conclusions}
\label{sec:conclusions}

We argued that the prevalence of uncontrolled states renders the estimation of individual program performance intractable.
In response, we formalized benchmarking as a decision problem and developed a largely self-contained framework based on blocked and simple randomized experiment designs.
We showed that focussing on deltas enables consistent estimation at an asymptotically optimal rate, sufficient to identify the fastest program while providing probabilistic guarantees in the finite regime.

We also analyzed alternative benchmarking methodologies (\Cref{sec:paired-benchmarking}, \Cref{sec:google-benchmark}) within the class of blocked experiments and found them theorectically less robust in the absence of strong assumptions about noise dynamics.
On the surface, these methods are appealing, as they empirically reduce the sample variance of the estimates; however, they account for neither the bias nor the asymptotic variance.
This oversight may render them inconsistent under unknown environmental dynamics.
Moreover, it precludes the development of probabilistic guarantees for the decision-making process.

Future work could focus on improving sample efficiency by relaxing the assumption of bounded noise magnitude.
Incorporating ancillary statistics offers a promising path forward.
Additionally, our theoretical results rely on the assumption that uncontrolled states affect all programs additively.
Where this does not hold, additional factors must be controlled \citep{mytkowicz2009producing}; detecting the presence of these factors, identifying them, and determining when they are significant remains an open challenge.

In summary, by enabling the reliable detection of small performance differences, our method supports the optimization of critical software infrastructure, including compilers, machine learning frameworks, and database systems.

\newpage

\section*{Acknowledgements}

We thank William Cohen for the valuable review.

{
  \ifdef{\groundskip}{}{\clearpage}
  \bibliography{paper}
  \bibliographystyle{plainnat}
  \ifdef{\groundskip}
        {\vskip\lastskip
         \addvspace{2\groundskip}}
        {}
}

\appendix

\clearpage

\section{Figures for Experimental Results}
\label{sec:figures-for-experimental-results}

\vspace{-\baselineskip}

\begin{figure}[H]
\centering
\begin{tikzpicture}
  \begin{axis} [ylabel=$\gamma$, xlabel=$n$,
      xlabel near ticks,
      ylabel near ticks,
      xlabel shift={-2pt},
      ylabel shift={-2pt},
      xmin=2,
      ymin=-0.02, ymax=0.2,
      ytick distance=0.1,
      legend columns=2,
      legend style={at={(0.96, 0.15)}, anchor=south east,
                    nodes={scale=0.7, transform shape}},
      legend cell align={left},
      height=0.6*\columnwidth,
      width=0.98\columnwidth,
    ]
    \pgfplotstableread{data/k-1.tbl}{\sorted}
    \addplot[adp10] table [x=n, y=adp10] {\sorted};
    \addlegendentry{$\alpha$ DP10};
    \addplot[asymmetric] table [x=n, y=asymm.] {\sorted};
    \addlegendentry{$\alpha$ Asymmetric};
    \addplot[martingale] table [x=n, y=mart.] {\sorted};
    \addlegendentry{$\alpha$ Martingale};
    \addplot[bdp10] table [x=n, y=bdp10] {\sorted};
    \addlegendentry{$\beta$ DP10};
  \end{axis}
\end{tikzpicture}
\vspace{-0.1\baselineskip}
\caption{Estimates of the worst-case confidence $\gamma$ for $K=1$, where $K = \omega \sqrt{n} / \mc$.
For a given $K$, the Asymmetric and Martingale approximations are constant, while both the $\alpha$ and $\beta$-designs are asymptotically $\bigO_p(\sqrt{n})$, so  their curves flatten out.
Note that the confidence for the $\beta$-design is constant zero here.
Also, the crossover point (as in \Cref{fig:threshold-0.125}) is at low confidences.}
\label{fig:k-1}
\end{figure}
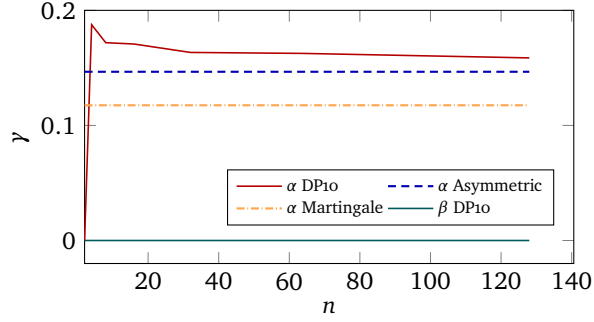

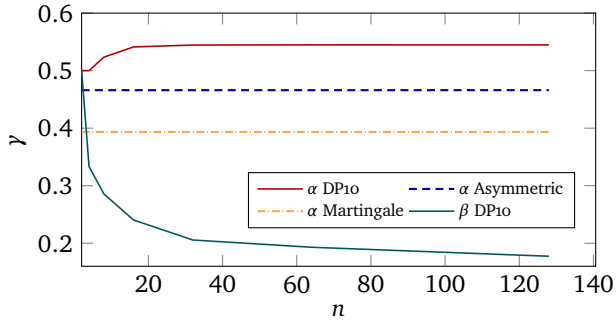
\begin{figure}[H]
\centering
\begin{tikzpicture}
  \begin{axis} [ylabel=$\gamma$, xlabel=$n$,
      xlabel near ticks,
      ylabel near ticks,
      xlabel shift={-2pt},
      ylabel shift={-2pt},
      xmin=2,
      ymin=0.16, ymax=0.6,
      legend columns=2,
      legend style={at={(0.96, 0.15)}, anchor=south east,
                    nodes={scale=0.7, transform shape}},
      legend cell align={left},
      height=0.6*\columnwidth,
      width=1.02\columnwidth,
    ]
    \pgfplotstableread{data/k-2.tbl}{\sorted}
    \addplot[adp10] table [x=n, y=adp10] {\sorted};
    \addlegendentry{$\alpha$ DP10};
    \addplot[asymmetric] table [x=n, y=asymm.] {\sorted};
    \addlegendentry{$\alpha$ Asymmetric};
    \addplot[martingale] table [x=n, y=mart.] {\sorted};
    \addlegendentry{$\alpha$ Martingale};
    \addplot[bdp10] table [x=n, y=bdp10] {\sorted};
    \addlegendentry{$\beta$ DP10};
  \end{axis}
\vspace{-0.1\baselineskip}
\end{tikzpicture}
\caption{Estimates of the worst-case confidence $\gamma$ for $K=2$.
At $n=2$, we have $n_f = n_g = 1$ for both the $\alpha$ and $\beta$-designs and that $\omega=\smash{\sqrt{2}} \mc$.
This limit can only be exceeded if the adversarial noise process guesses the order of the single runs of $f$ and $g$, which has probability $0.5$.
As $n$ increases, the $\alpha$-design stabilizes around that level as there is minimal information to exploit.
On the other hand, the $\beta$-design gets increasingly more exploitable.}
\label{fig:k-2}
\end{figure}

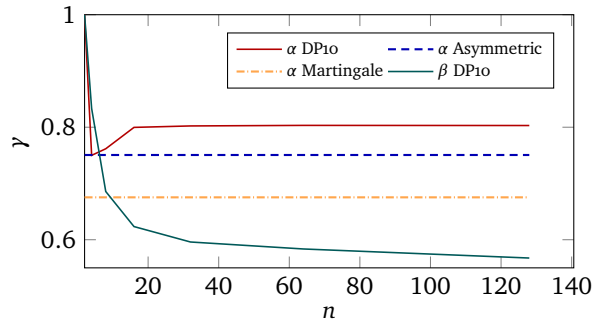
\begin{figure}[H]
\centering
\begin{tikzpicture}
  \begin{axis} [ylabel=$\gamma$, xlabel=$n$,
      xlabel near ticks,
      ylabel near ticks,
      xlabel shift={-2pt},
      ylabel shift={-2pt},
      xmin=2,
      ymin=0.55, ymax=1,
      legend columns=2,
      legend style={at={(0.96, 0.93)}, anchor=north east,
                    nodes={scale=0.7, transform shape}},
      legend cell align={left},
      height=0.6*\columnwidth,
      width=0.98\columnwidth,
    ]
    \pgfplotstableread{data/k-3.tbl}{\sorted}
    \addplot[adp10] table [x=n, y=adp10] {\sorted};
    \addlegendentry{$\alpha$ DP10};
    \addplot[asymmetric] table [x=n, y=asymm.] {\sorted};
    \addlegendentry{$\alpha$ Asymmetric};
    \addplot[martingale] table [x=n, y=mart.] {\sorted};
    \addlegendentry{$\alpha$ Martingale};
    \addplot[bdp10] table [x=n, y=bdp10] {\sorted};
    \addlegendentry{$\beta$ DP10};
  \end{axis}
\end{tikzpicture}
\vspace{-0.1\baselineskip}
\caption{Estimates of the worst-case confidence $\gamma$ for $K=3$.}
\label{fig:k-3}
\end{figure}

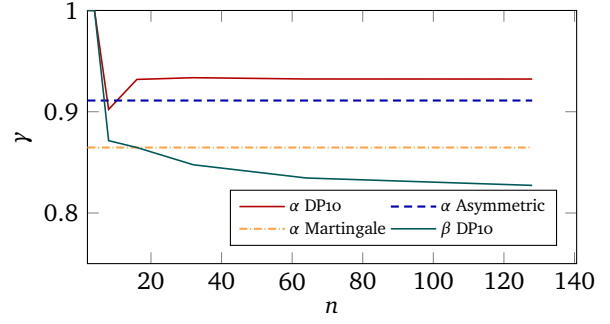
\begin{figure}[H]
\centering
\begin{tikzpicture}
  \begin{axis} [ylabel=$\gamma$, xlabel=$n$,
      xlabel near ticks,
      ylabel near ticks,
      xlabel shift={-2pt},
      ylabel shift={-2pt},
      xmin=2,
      ymin=0.75, ymax=1,
      legend columns=2,
      legend style={at={(0.96, 0.07)}, anchor=south east,
                    nodes={scale=0.7, transform shape}},
      legend cell align={left},
      height=0.6*\columnwidth,
      width=0.98\columnwidth,
    ]
    \pgfplotstableread{data/k-4.tbl}{\sorted}
    \addplot[adp10] table [x=n, y=adp10] {\sorted};
    \addlegendentry{$\alpha$ DP10};
    \addplot[asymmetric] table [x=n, y=asymm.] {\sorted};
    \addlegendentry{$\alpha$ Asymmetric};
    \addplot[martingale] table [x=n, y=mart.] {\sorted};
    \addlegendentry{$\alpha$ Martingale};
    \addplot[bdp10] table [x=n, y=bdp10] {\sorted};
    \addlegendentry{$\beta$ DP10};
  \end{axis}
\end{tikzpicture}
\caption{Estimates of the worst-case confidence $\gamma$ for $K=4$.}
\label{fig:k-4}
\end{figure}

\begin{figure}[H]
\centering
\begin{tikzpicture}
  \begin{axis} [ylabel=$\gamma$, xlabel=$n$,
      xlabel near ticks,
      ylabel near ticks,
      xlabel shift={-2pt},
      ylabel shift={-2pt},
      xmin=2,
      ymin=0.92, ymax=1,
      legend columns=2,
      legend style={at={(0.96, 0.07)}, anchor=south east,
                    nodes={scale=0.7, transform shape}},
      legend cell align={left},
      height=0.6*\columnwidth,
      width=0.98\columnwidth,
    ]
    \pgfplotstableread{data/k-5.tbl}{\sorted}
    \addplot[adp10] table [x=n, y=adp10] {\sorted};
    \addlegendentry{$\alpha$ DP10};
    \addplot[asymmetric] table [x=n, y=asymm.] {\sorted};
    \addlegendentry{$\alpha$ Asymmetric};
    \addplot[martingale] table [x=n, y=mart.] {\sorted};
    \addlegendentry{$\alpha$ Martingale};
    \addplot[bdp10] table [x=n, y=bdp10] {\sorted};
    \addlegendentry{$\beta$ DP10};
  \end{axis}
\end{tikzpicture}
\caption{Estimates of the worst-case confidence $\gamma$ for $K=5$.}
\label{fig:k-5}
\end{figure}
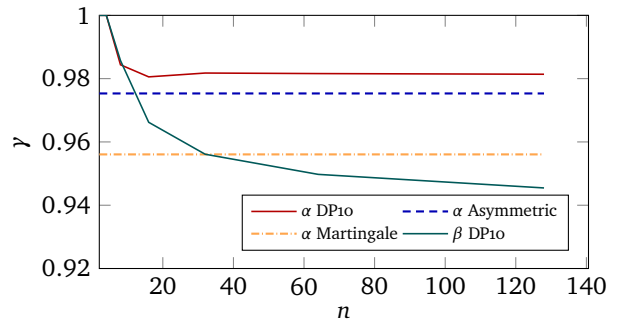

\begin{figure}[H]
\centering
\begin{tikzpicture}
  \begin{axis} [ylabel=$\gamma$, xlabel=$n$,
      xlabel near ticks,
      ylabel near ticks,
      xlabel shift={-2pt},
      ylabel shift={-2pt},
      xmin=2,
      ymin=0.98, ymax=1,
      ytick distance=0.01,
      legend columns=2,
      legend style={at={(0.96, 0.07)}, anchor=south east,
                    nodes={scale=0.7, transform shape}},
      legend cell align={left},
      height=0.6*\columnwidth,
      width=0.98\columnwidth,
    ]
    \pgfplotstableread{data/k-6.tbl}{\sorted}
    \addplot[adp10] table [x=n, y=adp10] {\sorted};
    \addlegendentry{$\alpha$ DP10};
    \addplot[asymmetric] table [x=n, y=asymm.] {\sorted};
    \addlegendentry{$\alpha$ Asymmetric};
    \addplot[martingale] table [x=n, y=mart.] {\sorted};
    \addlegendentry{$\alpha$ Martingale};
    \addplot[bdp10] table [x=n, y=bdp10] {\sorted};
    \addlegendentry{$\beta$ DP10};
  \end{axis}
\end{tikzpicture}
\caption{Estimates of the worst-case confidence $\gamma$ for $K=6$.}
\label{fig:k-6}
\end{figure}
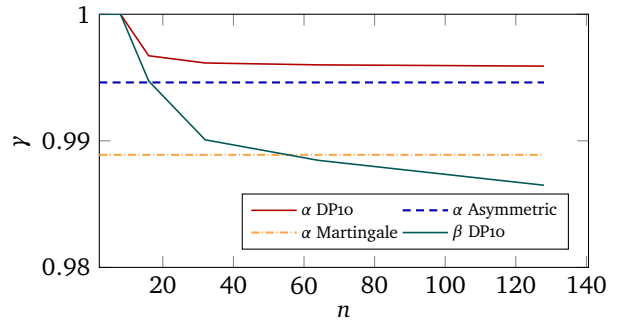

\newpage
\section{Proofs for Section Deltas to Dections}
\label{sec:proofs-for-deltas-to-decisions}

Here and in the following sections, we restate theorems and provide proof.
The proofs in this section accompany \Cref{sec:deltas-to-decision}.

\thmunambiguousregion*
\begin{proof}
The condition $\DT(\fg) < 0$ for all $g \neq f$ directly implies that $f$ is strictly faster than any other program.
Uniqueness follows from the antisymmetry of $\DT(\fg)$.

We now prove that the decision rule picks $f$ with probability at least $\gamma$.
By \Cref{rem:equivalence-thresholding}, $f$ is picked if $\de(\fg) \leq - \smash{Q_{\gammadecprime}}(g, f)$ for all $g \neq f$.
Expanding $\DE(\fg) = \DT(\fg) + Z(\fg)$, the condition for a given $g$ becomes
\begin{align*}
\DT(\fg) + Z(\fg) &\leq -Q_{\gammadecprime}(g, f).
\end{align*}
Substituting the upper bound for the true delta, $\DT(\fg) \leq S_{\gamma}(\fg) = -Q_{\gamma'}(\fg) - Q_{\gammadecprime}(g, f)$, the above condition is implied by
\begin{gather*}
-Q_{\gamma'}(\fg) - Q_{\gammadecprime}(g, f) + Z(\fg) \leq -Q_{\gammadecprime}(g, f),
\end{gather*}
which simplifies to $Z(\fg) \leq Q_{\gamma'}(\fg)$.
We need this sufficient condition to hold for all $\ct{F} - 1$ programs $g \neq f$.
By the definition of the quantile, $P(Z(\fg) \leq Q_{\gamma'}(\fg)) \geq \gamma'$ for any single $g$.
Taking a union bound, the probability that the condition fails for \textit{at least one} $g$ is at most
\begin{align*}
\sum_{g \neq f} (1 - \gamma') = (\ct{F} - 1) \frac{1 - \gamma}{\ct{F} - 1} = 1 - \gamma.
\end{align*}
Thus, the probability that the condition holds for all $g$ (and $f$ is picked) is at least $\gamma$.
\end{proof}

\thmdectodeconsistencyrate*
\begin{proof}
Let $E_c$ be the failure event $\DE^c(\fg) - \DT^c(\fg) > U_n^c$, where $P(E_c) \leq \alpha_c$.
For any part $\sS' \in \sS$, arbitrary dependencies may exist between configurations (\Cref{def:benchmark-suite-model}).
By Boole's inequality, the probability that at least one configuration in $\sS'$ fails is bounded by the sum of their individual failure probabilities: $P(\bigcup_{c \in \sS'} E_c) \leq \min(1, \sum_{c \in \sS'} \alpha_c)$.
Consequently, the probability that all configurations within $\sS'$ succeed is at least $1 - \min(1, \sum_{c \in \sS'} \alpha_c)$.
By the independence of parts, their success probabilities combine multiplicatively, yielding the stated probabilistic upper bound.

For the second claim, by the linearity of the estimators and the true deltas, we have $\DE(\fg) - \DT(\fg) = \sum_{c \in \sC} P(c) (\DE^c(\fg) - \DT^c(\fg))$.
If the success event $\DE^c(\fg) - \DT^c(\fg) \leq U_n^c$ occurs for all $c \in \sC$, then their convex combination guarantees $\DE(\fg) - \DT(\fg) \leq \sum_{c \in \sC} P(c) U_n^c = U_n$.
Since $\max(0, U_n^c) \leq K \phi(n)$ for some $K$ independent of $c$, and $\sum_{c \in \sC} P(c) = 1$, it follows that
\begin{align*}
\max(0, U_n)
&= \max\Bigl(0, \sum_{c \in \sC} P(c) U_n^c \Bigr) \leq \sum_{c \in \sC} P(c) \max(0, U_n^c) \\
&\leq K \phi(n) \sum_{c \in \sC} P(c) = K \phi(n).
\end{align*}
Thus, $\max(0, U_n) = \bigO(\phi(n))$.
\end{proof}

\thmuncertainruleconsistency*
\begin{proof}
By \Cref{thm:unambiguous-region}, the decision rule picks $f$ with probability at least $\gamma$ if $\DT(\fg) \leq S_{\gamma}(\fg)$ for all $g \neq f$, where $S_{\gamma}(\fg) = - Q_{\gamma'}(\fg) - Q_{\gammadecprime}(g, f)$.
By \Cref{thm:dec-to-de-consistency-rate}, both $Q_{\gamma'}$ and $Q_{\gammadecprime}$ vanish at rate $\bigO(\phi(n))$, thus $S_{\gamma}(\fg) \to 0$.

Since $f$ is the unique fastest program, $\DT(\fg) < 0$ for all $g \neq f$.
For any $\gamma < 1$, there exists an $N$ such that for all $n > N$, $S_{\gamma}(\fg) \geq \DT(\fg)$ holds for every $g \neq f$.
Consequently, the probability of picking $f$ is at least $\gamma$ for any $\gamma < 1$, which implies the limit is $1$.
\end{proof}

\newpage
\section{Proofs for Blocked Experiments}
\label{sec:proofs-for-blocked-experiments}

The following subsections parallel those of \Cref{sec:blocked-experiments} and provide proofs for the theorems stated there.

\subsection{Proofs for Single-Block Experiments}
\label{sec:proofs-for-single-block-experiments}

Here, we prove the theorems stated in \Cref{sec:single-block-experiments} about the contraction rates of the variance and the bias of the $\DE^c$ estimator.

To bound the bias and the variance of our estimators in single-block experiments, we explicitly construct the most adversarial noise processes $\Epsilonvec$\halfnegkern, which maximize $\abs{\E[\DIvec \cdot \Epsilonvec]}$ or $\Var[\DIvec \cdot \Epsilonvec]$ and are closely related to the bias and the variance of the estimator $\DE^c$ in blocked designs via the unscaled partial sum $\Sigma_i = \smash{\sum_{j=1}^i} \DI_j \Epsilonsub_j$ (thus $\Sigma_n = \DIvec \cdot \Epsilonvec = \DE^c(\fg) \frac{n}{\ct{F}}$).
We denote realizations of $\Sigma_i$ with $\sigma_i$.
For simplicity, we assume that $\mc = 1$, so $\abs{\Epsilonsub_i} \leq 1$ for all $i$; our results apply to other cases modulo rescaling.

\subsubsection{Worst-Case Variance in Single-Block Experiments}

Aiming to characterize the worst case variance $\Var[\DIvec \cdot \Epsilonvec] = \E[\Sigma_n^2] - \E[\Sigma_n]^2$\supkern, we focus on how to maximize $\E[\Sigma_n^2]$.
First, we show that optimal (i.e. worst-case) noise can only take extreme values, which goes by the name of bang--bang control in control theory.
Then, we identify the greedy strategy that maximizes the local expectation of $\Sigma_i^2$ given the state after run $i-1$.
Finally, we show that the greedy strategy is globally optimal and that $\E[\Sigma_n] = 0$ with it, so it also maximizes the variance.
Also, since there is a fixed number of $\DI_i = 0$ terms in a block, which affect neither the expectation nor the variance of $\Sigma_n = \DIvec \cdot \Epsilonvec$, without loss of generality, we assume $\ct{F} = 2$ in the proofs.

\begin{proposition}[Optimality of bang--bang control]
\label{thm:optimality-of-bang-bang}
In a delta design, for any choice of distribution for $\Epsilonvec$ that maximizes $\E[(\DIvec \cdot \Epsilonvec)^2]$, the noise must be extreme almost surely, that is, $P(\abs{\Epsilonsub_i} = 1) = 1$ for all $i \in [1, n]$.
\end{proposition}

\begin{proof}
Suppose that some choice of $\Epsilonvec$ maximizes $\E[\Sigma_n^2]$, but $P(\abs{\Epsilonsub_i} = 1) < 1$ for some $i$.
Then, there exists some $\divec_{<i}$ and $\sigma_{i-1}$ such that $P(\divec_{<i}, \sigma_{i-1}) > 0$ and $P(\abs{\Epsilonsub_i} < 1 \given \divec_{<i}, \sigma_{i-1}) > 0$.
Observe that due to the delta design's sequential structure, $\Epsilonvec$ also maximizes $\E[\Sigma_n^2 \given \divec_{<i}, \sigma_{i-1}]$.

We are going to show that if $\Epsilonsub_i = \epsilon \in (-1, 1)$, then its expectation can be increased under an alternative model with $\Epsilonsub_i = \epsilon'\! = \pm1$, so it cannot be optimal.
Denote the events $(\divec_{<i}, \sigma_{i-1}, \Epsilonsub_i = \epsilon)$ and $(\divec_{<i}, \sigma_{i-1}, \Epsilonsub_i = \epsilon')$ with $C$ and $C'$, respectively, and let $\delta = \epsilon' - \epsilon$.
In the alternative model, we keep $\DIvec_{i\leq} \given \divec_{<i}$ the same as in the original, and we let $\Epsilonvec_{i<} \given C'$ be the same as $\Epsilonvec_{i<} \given C$ in the original as if ignoring the clamping of $\Epsilonsub_i$ to $\epsilon'$ when generating $\Epsilonvec_{i<}$.
These two distributions define $\Sigma_n \given C'$ in the alternative model.
Denoting the expectation under the alternative model with $\E'$, we have that
\begin{align*}
\E'[\Sigma_n^2 \given C']
&= \E[(\Sigma_n + \DI_i \delta)^2 \given C]\\
&= \E[\Sigma_n^2 \given C] + \delta^2 + 2 \delta \E[\Sigma_n \DI_i \given C],
\end{align*}
where the sign of the third term is $\epsilon'\! \sign(\E[\Sigma_n \DI_i \given C])$, so at least one of $\epsilon'\! = 1$ and $\epsilon'\! = -1$ makes it non-negative.
Since $\epsilon \neq \epsilon'$\!, we have that $\delta^2 > 0$, so
\begin{align*}
\E'[\Sigma_n^2 \given C'] > \E[\Sigma_n^2 \given C],
\end{align*}
strictly increasing the expectation for any non-extreme $\epsilon$, which contradicts either that the original model was optimal or that it assigns non-zero probability to $\abs{\Epsilonsub_i} < 1$.
\end{proof}

\begin{definition}[Greedy noise]
In a delta design, we call the noise $\Epsilonvec$ greedy if $\Epsilonsub_1$ is drawn from a uniform distribution over $\{-1, 1\}$, while for $i \in [2, n]$, given any $\divec_{<i}$, $\sigma_{i-1}$ and $\epsilon_{i-1}$, with the shorthands $\psi^\sigma_i = \sign(\sigma_{i-1})$ and $\psi^\lambda_i = \sign(\E[\DI_i \given \divec_{<i}])$, the noise $\Epsilonsub_i = \psi^\sigma_i \psi^\lambda_i$ unless that is $0$, in which case $\Epsilonsub_i = \epsilon_{i-1}$ .
\end{definition}

Although $\Epsilonsub_i$ is conditioned only on $U_i$ in delta designs, the greedy strategy can be understood as encoding $\DIvec_{<i}$, $\Sigma_{i-1}$ and $\Epsilonsub_{i-1}$ in $U_i$.
As we next show, greedy noise myopically maximizes the squared distance from the origin in all steps as its name suggests.

\begin{proposition}[Greedy step optimality]
\label{thm:greedy}
In a delta design, choosing $\Epsilonsub_i$ greedily maximizes $\E[\Sigma_i^2 \given \divec_{<i}, \sigma_{i-1}]$ for all $i \in [1, n]$, $\divec_{<i}$ and $\sigma_{i-1}$.
\end{proposition}

\begin{proof}
Under the event $(\divec_{<i}, \sigma_{i-1})$,
\begin{align*}
\E[\Sigma_i^2]
&= \E[(\Sigma_{i-1} + \DI_i \Epsilonsub_i)^2]\\
&= \sigma_{i-1}^2 + \E[\DI_i^2 \Epsilonsub_i^2] + 2 \sigma_{i-1} \E[\DI_i \Epsilonsub_i]\\
&= \sigma_{i-1}^2 + \E[\DI_i^2 \Epsilonsub_i^2] + 2 \sigma_{i-1} \E[\DI_i] \E[\Epsilonsub_i],
\end{align*}
where we used that $\DI_i \indep \Epsilonsub_i \given \divec_{<i}$ in delta designs.
Here, the first term is constant, and the second term is maximized by any $\Epsilonsub_i \in \{-1, 1\}$.
As to the other non-constant term, $2 \sigma_{i-1} \E[\DI_i] \E[\Epsilonsub_i]$, it is maximized when $\E[\Epsilonsub_i \given \divec_{<i}, \sigma_{i-1}] = \psi^\sigma_i \psi^\lambda_i$, unless that is $0$, in which case it is always $0$ regardless of $\Epsilonsub_i$.
Since the greedy strategy meets these conditions, it maximizes $\E[\Sigma_i^2 \given \divec_{<i}, \sigma_{i-1}]$ for all $i \in [1, n]$, $\divec_{<i}$.
\end{proof}

\begin{assumption}[No flip]\label{ass:no-flip}
The more likely draw cannot flip sign in one step: $\Psi^\lambda_i \Psi^\lambda_{i+1} \neq -1$ under any $(\divec_{<i}, \sigma_{i-1})$.
\end{assumption}

Note that this assumption is true for all $\DIvec$ based on sampling without replacement, which includes all block designs.

\begin{proposition}[Greedy zero crossing]
\label{thm:greedy-zero-crossing}
In a delta design with greedy noise satisfying \Cref{ass:no-flip} after runs $[1, i - 1]$ ($i \in [1, n - 1]$), $\psi^\sigma_i \neq 0$ and $\psi^\lambda_i \neq 0$ implies $\Epsilonsub_{i+1} = \Epsilonsub_i$ under $(\divec_{<i}, \sigma_{i-1})$.
\end{proposition}

\begin{proof}
Let $\di_i \in \{-1, 1\}$.
The greedy strategy deterministically sets $\epsilon_i = \psi^\sigma_i \psi^\lambda_i$ since neither sign is $0$.
Greedy noise is $\pm1$, so $\abs{\sigma_i - \sigma_{i-1}} = \abs{\di_i \epsilon_i} = 1$.
Furthermore, it follows that all sums are integer, so $\psi^\sigma_{i+1} \in \{\psi^\sigma_i, 0\}$ because $\sigma_{i-1} \neq 0$.
Similarly, from $\psi^\lambda_i \psi^\lambda_{i+1} \neq -1$ and $\psi^\lambda_i \neq 0$, we have $\psi^\lambda_{i+1} \in \{\psi^\lambda_i, 0\}$.
If neither $\psi^\sigma_{i+1}$ or $\psi^\lambda_{i+1}$ is zero, then $\psi_{i+1} = \psi^\sigma_{i+1} \psi^\lambda_{i+1} = \psi^\sigma_i \psi^\lambda_i = \psi_i$.
If either one is zero, then the greedy strategy repeats the noise from the most recent step where neither was zero, which is step $i$ in this case, so we have $\Epsilonsub_{i+1} = \Epsilonsub_i$ for all $\di_i$.
\end{proof}

\begin{proposition}[Worst-case noise]
\label{thm:worst-case-noise}
In any delta design, greedy noise that satisfies \Cref{ass:no-flip} maximizes $\Var[\DIvec \cdot \Epsilonvec]$.
\end{proposition}

\begin{proof}
Let $V_i(\divec_{<i}, \sigma_{i-1}) = \max_{\Epsilonvecsub_{i\leq}} \E[\Sigma_n^2 \given \divec_{<i}, \sigma_{i-1}]$ be the maximum squared distance achievable from the given state by any noise strategy for the remaining runs.
By \Cref{thm:optimality-of-bang-bang}, we restrict our attention to bang--bang strategies and show by backward induction that the greedy strategy maximizes $V_1(\emptyset, 0) = \max_{\Epsilonvec} \E[\Sigma_n^2]$.

In the inductive base cases $n+1$ and $n$, $V_{n+1}(\divec_{\leq n}, \sigma_n) = \sigma_n^2$ leaves no noise to choose, making all strategies optimal, while for $V_n$, \Cref{thm:greedy} establishes greedy optimality for choosing $\Epsilonsub_n$.

Inductive step.
Let $1 \leq i < n$, and suppose that greedy $\Epsilonsub_j$ is optimal for all $j > i$ in all states $\divec_{<j}$, $\sigma_{j-1}$.
Let $\sigma_{i-1}$ and $\divec_{<i}$ be any state, $\psi^\sigma_i = \sign(\sigma_{i-1})$ ($i \in [1, n]$), the direction from $\sigma_{i-1}$ leading away from zero, and $\psi^\lambda_i = \sign(\E[\DI_i \given \divec_{<i}])$, the more likely flip.

\begin{itemize}
\item If $\psi^\sigma_i = 0$ (i.e. $\sigma_{i-1} = 0$), then $\abs{\sigma_i} = 1$ under any bang--bang control.
By symmetry, $V_j(\divec_{<j}, \sigma_{j-1})$ is an even function of $\sigma_j$, yielding $V_{i+1}((\divec_{<i}, \di_i), 1) = V_{i+1}((\divec_{<i}, \di_i), -1)$.
Consequently, all bang--bang strategies are optimal.

\item
If $\psi^\lambda_i = 0$, then $P(\DI_i = 1 \given \divec_{<i}) = P(\DI_i = -1 \given \divec_{<i})$, so $\DI_i \Epsilonsub_i \given \divec_{<i}$ has the same distribution for any bang--bang $\Epsilonsub_i$.
Hence, once again, all bang--bang controls have the same value and are optimal.

\item This leaves the case where $\psi^\sigma_i \neq 0$ and $\psi^\lambda_i \neq 0$.
Note that $i \neq 1$ here because $i = 1$ implies $\sigma_0 = 0$ and $\psi^\sigma_1 = 0$.
First, we address the deterministic case.
With bang--bang strategies, there are only two values for $\Epsilonsub_i$ to consider: $-1$ and $1$.
One of these corresponds to the greedy choice and the other to the opposite, non-greedy choice.
We now show that the greedy choice is always at least as good as the non-greedy choice.
For $i \in [2, n]$, let $G_i$ and $N_i$ be the maximum achievable when using the greedy and the non-greedy strategy, respectively, to determine the noise $\Epsilonsub_i$ in run $i$:
\begin{align*}
&G_i(\divec_{<i}, \sigma_{i-1})\\
&\qquad= \E[V_{i+1}((\divec_{<i},\DI_i), \sigma_{i-1} + \DI_i \Psi_i) \given \divec_{<i}, \sigma_{i-1}]\\
&N_i(\divec_{<i}, \sigma_{i-1})\\
&\qquad= \E[V_{i+1}((\divec_{<i},\DI_i), \sigma_{i-1} - \DI_i \Psi_i) \given \divec_{<i}, \sigma_{i-1}],
\end{align*}
where $\Psi_i$ denotes the greedy noise.
We also define $G_{n+1}(\divec_{\leq n}, \sigma_n) = N_n(\divec_{\leq n}, \sigma_n) = \sigma_n^2$.

From $\Psi_{i+1} = \Psi_i$ by \Cref{thm:greedy-zero-crossing} and the exchangeability of $\DIvec$, we have that taking a non-greedy then a greedy step is the same as taking a greedy then a non-greedy step, but the induction hypothesis states that in the latter step, greedy is not worse than non-greedy, so it must be not worse in the former step, either.
Formally, under the event $(\divec_{<i}, \sigma_{i-1})$,
\begin{align*}
&N_i(\divec_{<i}, \sigma_{i-1})\\
&\qquad = \E[V_{i+1}((\divec_{<i},\DI_i), \sigma_{i-1} - \DI_i \Psi_i)]\\
\explainstep{by the induction hypothesis, $G_{i+1}$ is optimal}
&\qquad = \E[G_{i+1}((\divec_{<i}, \DI_i), \sigma_{i-1} - \DI_i \Psi_i)]\\
&\qquad = \E[G_{i+2}((\divec_{<i}, \DI_i, \DI_{i+1}), \sigma_{i-1} - \DI_i \Psi_i + \DI_{i+1} \Psi_{i+1})]\\
\explainstep{since $\Psi_i = \Psi_{i+1}$ and $\DIvec_{i\leq} \given \divec_{<i}$ is exchangeable}
&\qquad = \E[G_{i+2}((\divec_{<i}, \DI_i, \DI_{i+1}), \sigma_{i-1} + \DI_i \Psi_i - \DI_{i+1} \Psi_{i+1})]\\
&\qquad = \E[N_{i+1}((\divec_{<i}, \DI_i), \sigma_{i-1} + \DI_i \Psi_i)]\\
\explainstep{by the induction hypothesis again}
&\qquad \leq \E[G_{i+1}((\divec_{<i},\DI_i), \sigma_{i-1} + \DI_i \Psi_i)]\\
&\qquad = G_i(\divec_{<i}, \sigma_{i-1}).
\end{align*}
Since the greedy value is always at least as good as the non-greedy one, the deterministic greedy strategy is not worse than any non-deterministic mixture, either.
\end{itemize}

We have thus shown that if greedy is optimal for steps $j > i$, then it is optimal for $i$, which concludes the induction.
Hence, greedy is optimal for all $i \in [1, n]$ and maximizes $V_1(\emptyset, 0) = \max_{\Epsilonvec} \E[\Sigma_n^2]$.

Since greedy maximizes $\E[\Sigma_n^2]$ and $\E[\Sigma_n] = 0$ by the symmetry provided by the first step, it follows that greedy also maximizes $\Var[\Sigma_n] = \E[\Sigma_n^2] - \E[\Sigma_n]^2$.
\end{proof}

\begin{proposition}
\label{thm:beta-noise-var}
For the $\beta(n)$ design, $\Var[\DIvec \cdot \Epsilonvec] = \bigO(n)$.
\end{proposition}

\begin{proof}
We can restrict our attention to the greedy strategy because the no-flip assumption (\Cref{ass:no-flip}) holds for blocked designs; thus, by \Cref{thm:worst-case-noise}, greedy noise maximizes variance.
With the greedy algorithm, $\Sigma_i$ is a random walk that takes steps $\DI_i \Epsilonsub_i$ and is biased away from the simple random walk $W_i$ that takes $\pm1$ steps with uniform probability.
The greedy first step is randomized, and $\E[\Sigma_i] = 0$ for all $i$ due to the resulting symmetry.
Intuitively, $\Var[\Sigma_n] \geq \Var[W_n] = \bigO(n)$.
However, what we need to show is that it is \emph{at most} of that order.
Let the rv $Z^\sigma_n = \{i \in [1, n] \colon \Psi^\sigma_i = 0 \}$ and $Z^\lambda_n = \{i \in [1, n] \colon \Psi^\lambda_i = 0 \}$.

\begin{itemize}
\item
We first show that $\abs{\Sigma_n} \leq 2 (Z^\sigma_n + Z^\lambda_n)$.
Split runs (indices) into segments where $\Psi^\lambda_i \neq 0$, and observe that each segment has the same number of $1$s and $\neg1$s among its $\DI_i$.
Note that as long as $\psi^\sigma_i \neq 0$, the number of hits and misses are almost equal.
Almost because the first $\DI_i$ in the segment is going to be the rare one (equally likely at the time though), but we may get lucky and get it correct, turning a miss into a hit.
Thus, the number of hits is either equal to the number of misses or greater by $2$.
Note that when the sum is zero, a miss has the same effect as a hit: it increases the sum by one in absolute value.
So every time the sum is zero, a miss may be turned into a hit in effect.

\item
Next, we show that $\E[(Z^\sigma_n)^2] = \bigO(n)$.
Recall that $\Psi^\sigma_i = \sign(\Sigma_{i-1})$.
For $\beta$-designs, $\Sigma_i$ is a random walk whose distance from zero is biased to increase due to the greedy construction.
Specifically, conditional on any history, the probability of $\Sigma_i$ moving away from zero is always at least $1/2$, meaning $P(\abs{\Sigma_i} = \abs{\Sigma_{i-1}} + 1 \given \abs{\Sigma_{i-1}} > 0) \geq 1/2$.
Because this outward probability is bounded from below by that of a simple symmetric random walk $W_i$, the conditional distribution of the duration of any excursion from zero for $\Sigma_i$ is stochastically lower-bounded by the excursion duration of $W_i$.
By induction over the excursions, the time of the $k$-th return to zero for $\Sigma_i$ is stochastically larger than that for $W_i$.
Consequently, the total number of returns by step $n$ is stochastically smaller, meaning $Z^\sigma_n \preceq_{\mathrm{st}} Z_n^W$.
Since $\E[(Z_n^W)^2]$ is $\bigO(n)$ by \Cref{thm:n-returns-in-simple-random-walk}, it follows that $\E[(Z^\sigma_n)^2]$ is also $\bigO(n)$.

\item
Third, we show that $\E[(Z^\lambda_n)^2] = \bigO(n)$.
Let $L_i = \sum_{j=1}^i \DI_j$ for $i \in [0, n]$.
Observe that due to the balancedness of blocks, $L_i = 0$ iff $\smash{\Psi^\lambda_{i+1}} = 0$, so the number of times $L_i$ returns to zero is equal to $Z^\lambda_n$.
For $\beta$-designs, $L_i$ is a tied-down random walk (i.e. $L_n = 0$), and from \Cref{thm:n-returns-in-tied-down-random-walk} \citep{godreche2017longest}, we have that $\E[(Z^\lambda_n)^2] = \bigO(n)$.
\end{itemize}

Using the above,
\begin{align*}
&\Var[\DIvec \cdot \Epsilonvec \given n]\\
&\qquad = \E[(\DIvec \cdot \Epsilonvec)^2 \given n]\\
&\qquad \leq \E[(2 Z^\sigma_n + 2 Z^\lambda_n)^2]\\
&\qquad = 4 \PP[\big]{\E[(Z^\sigma_n)^2] + \E[(Z^\lambda_n)^2] + 2 \E[Z^\sigma_n Z^\lambda_n]}\\
\explainstep{by the geometric--arithmetic mean inequality}
&\qquad \leq 4 \PP[\big]{\E[(Z^\sigma_n)^2] + \E[(Z^\lambda_n)^2] + \E[(Z^\sigma_n)^2 + (Z^\lambda_n)^2]}\\
&\qquad = 8 \PP[\big]{\E[(Z^\sigma_n)^2] + \E[(Z^\lambda_n)^2]}\\
\explainstep{because both expectations are $\bigO(n)$}
&\qquad = \bigO(n).\qedhere
\end{align*}
\end{proof}

\thmbetavariancerate*
\begin{proof}
Connecting greedy to the overall estimator variance for the balanced case yields
\begin{align*}
\Var[\DE^c]
&= \Var[\DInormvec \cdot \vect{T}]
= \Var[\DInormvec \cdot \Tauvec + \DInormvec \cdot \Epsilonvec]\\
\end{align*}
Applying the law of total variance to the first term yields
\begin{align*}
\Var[\DInormvec \cdot \Tauvec]
&= \Var \E[\DInormvec \cdot \Tauvec \given \vectrv{F}] +
   \E \Var[\DInormvec \cdot \Tauvec \given \vectrv{F}]
\explainstep{$\Tauvec$ are i.i.d given $\vectrv{F}$, so $\E[\DInormvec \cdot \Tauvec \given \vectrv{F}]$ is constant}
&= \E \Var[\DInormvec \cdot \Tauvec \given \vectrv{F}] \\
\explainstep{by $\Tauvec$ being i.i.d given $\vectrv{F}$ and by balancedness}
&= \bigO(n^\negone).
\end{align*}
Second, $\Var[\DIvec \cdot \Epsilonvec]$ is maximized by greedy as per \Cref{thm:worst-case-noise}, and we have shown it to be $\bigO(n)$ in \Cref{thm:beta-noise-var}.
For blocked designs, $\Var[\DInormvec \cdot \Epsilonvec] = \ct{F}^2 n^{\neg 2} \Var[\DIvec \cdot \Epsilonvec]$, and consequently $\Var[\DInormvec \cdot \Epsilonvec] = \bigO(n^\negone)$.
Finally, the covariance is upper bounded by the geometric mean of the first two terms:
\begin{align*}
\Cov[\DInormvec \cdot \Tauvec\supkern,\, \DInormvec \cdot \Epsilonvec]
\leq \sqrt{\Var[\DInormvec \cdot \Tauvec] \Var[\DInormvec \cdot \Epsilonvec]}
= \bigO(n^\negone).
\end{align*}
Hence, $\Var[\DE]$ is $\bigO(n^\negone)$.
\end{proof}

\subsubsection{Worst-Case Bias in Single-Block Experiments}

Here, we prove the rate of decay for the bias, via the construction of the worst-case noise process, which also allows to bound the bias in the finite regime.

\thmbetabiasrate*
\begin{proof}
We first show that since the noise is additive, the bias is proportional to $\sum_{i=1}^n \E[\DI_i \Epsilonsub_i]$:
\begin{align*}
\E[\DE^c - \delta^c]
&= \E[(\TauAsub_f + \EAsub_f) - (\TauAsub_g + \EAsub_g) - \delta^c]\\
&= \E[\EAsub_f - \EAsub_g + (\TauAsub_f - \TauAsub_g) - \delta^c]
\explainstep{by $\E[\TauAsub_f - \TauAsub_g] =(l(\fc) - \E_U[\Epsilon]) - (l(g, c) - \E_U[\Epsilon])= \delta^c$}
&= \E[\EAsub_f - \EAsub_g]\\
&= \frac{\ct{F}}{n} \sum_{i=1}^n \E[\DI_i \Epsilonsub_i].
\end{align*}
The terms of the sum can be maximized individually, which also maximizes the absolute value of $\E[\DE^c - \delta^c]$.
The ceteris paribus assumptions (\Cref{ass:additive-global-effect,ass:multiplicative-global-effect}) posit that $\Epsilonsub_i$ is generated from $U_i$, but since $\DI_i \indep U_i \given \DIvec_{<i}$, it suffices to choose $\Epsilonsub_i$ based on $\DIvec_{<i}$ only.
The law of total expectation gives $\E[\DI_i \Epsilonsub_i] = \E[\E[\DI_i \Epsilonsub_i \given \DIvec_{<i}]]$, where the inner expectation is maximized by $\Epsilonsub_i = \sign(\E[\DI_i \given \DIvec_{<i}]) = -\sign(L_{i-1})$, where $L_{i-1} = \sum_{j=1}^{i-1} \DI_j$ is the sign of the more likely value of $\DI_i$ given the history so far.

With these optimal $\Epsilonsub_i$, the terms in the sum can be resolved as follows.
\begin{align*}
\E[\DI_i \Epsilonsub_i]
&= \E[\E[\DI_i \Epsilonsub_i \given \DIvec_{<i}]]\\
&= \E[\Epsilonsub_i \E[\DI_i \given \DIvec_{<i}]]\\
&= \E\Pigl[-\sign(L_{i-1}) \frac{-L_{i-1}}{n-i+1}\Pigr]\\
&= \frac{\E[\abs{L_{i-1}}]}{n-i+1}.
\end{align*}
Plugging these back into the sum yields
\begin{align*}
\abs[\big]{\E[\DE^c - \delta^c]}
&\leq \frac{\ct{F}}{n} \sum_{k=0}^{n-1} \frac{\E[\abs{L_k}]}{n-k}
\explainstep{by $\E[L_k] = 0$ and Jensen's inequality, $\E[\abs{L_k}] \leq \sqrt{\Var[L_k]}$}
&\leq \frac{\ct{F}}{n} \sum_{k=0}^{n-1} \frac{\sqrt{\Var[L_k]}}{n-k}
\explainstep{$(L_k+k)/2$ follows a hypergeometric distribution}
&= \frac{\ct{F}}{n} \sum_{k=0}^{n-1} \frac{1}{n-k} \sqrt{\frac{k(n-k)}{n-1}}\\
&= \frac{\ct{F}}{n \sqrt{n-1}} \sum_{k=0}^{n-1} \sqrt{\frac{k}{n-k}}.
\end{align*}
As shown in \Cref{thm:some-sum-big-o-n}, the sum is $\bigO(n)$, from which the claim follows immediately.
\end{proof}

\subsection{Proofs for Multi-Block Experiments}
\label{sec:proof-for-multi-block-experiments}

Here, we prove the theorems stated in \Cref{sec:multi-block-experiments}.

\thmmultiblockbounds*
\begin{proof}
We express $\smash{\DE^c}$ as the sum of the per-block terms using the balancedness of blocks.
Under $\cC$
\begin{align*}
\DE^c\noscriptspace
= \frac{\ct{F}}{\sum_i b_i} \DIvec \cdot \vectrv{T}
= \frac{\ct{F}}{\sum_i b_i} \sum_{k=1}^{\len{\vect{b}}} \DIvec^k\! \cdot \vectrv{T}^k\noscriptspace
= \sum_{k=1}^{\len{\vect{b}}} \frac{b_k}{\sum_i b_i} \DE^k\supkern,
\end{align*}
where $\DE^k\noscriptspace = \smash[t]{\frac{\ct{F}}{b_k}\DIvec^k\! \cdot \vectrv{T}^k}$\supkern.
Since $\smash{\DE^k}$ is normalized by the number of runs of $f$ and $g$ (they are all $\smash{\frac{b_k}{\ct{F}}}$), it is equivalent to $\smash{\DE^c} \given \beta(b_k)$, the estimator for configuration $c$ in a single-block experiment of size $b_k$ except for its initial uncontrolled state.
Furthermore, $\smash{\DE^c}$ is a convex combination of the $\DE^k$ terms, weighted by $P(k) = \frac{b_k}{\sum_i b_i}$.
Based on these observations, we can view individual blocks as per-configuration experiments, $P(k)$ as the prior over these pseudo-configurations, and $\DE^c$ as the suitewise estimate $\DE$.
To account for the dependencies between the initial states, we choose the trivial partition $\set{[1,k]}$ for $\sS$.
With these substitutions, \Cref{thm:dec-to-bounded-de} states that the bounds for the blocks combine linearly.
\end{proof}

\subsection{Proofs Combining Per-Configuration Bounds}
\label{sec:proofs-for-combining-per-configuration-bounds}

Here, we prove the theorems stated in \Cref{sec:combining-per-configuration-bounds}.

\thmdectoboundedde*
\begin{proof}
From $\DE = \sum_{c \in \sC} P(c) \DE^c$, then from $\delta = \E[\delta^C]$ and the law of total expectation with respect to $U^c_1$, we have
\begin{align*}
\abs[\big]{\EB{\DE} - \DT}
&=         \abs[\pig]{\sum_{c \in \sC} P(c)            \EB{\DE^c} - \DT} \\
&=         \abs[\pig]{\sum_{c \in \sC} P(c) \PP[\pig]{ \EB{\DE^c} - \DT^c}}\\
&\leq             \sum_{c \in \sC} P(c) \abs[\pig]{\EB{\DE^c} - \DT^c} \\
&=                     \sum_{c \in \sC} P(c) \abs[\pig]{ \EB[\big]{ \EB{\DE^c \given U^c_1} - \DT^c}}\\
&\leq             \sum_{c \in \sC} P(c) \EB[\pig]{ \abs[\big]{ \EB{\DE^c \given U^c_1} - \DT^c}}.
\end{align*}
Although $b^c$ is a bound in the context of the measurement model, with the initial state marginalized out, the distribution of $U_1$ (as provided by the execution environment) is unknown and arbitrary.
Hence, the bound must apply to any distribution, including deterministic ones, from which it follows that $\abs{ \E[\DE^c \given U^c_1] - \DT^c} \leq b^c$\supkern.
Taking the expectation and a convex combination with respect to $P(c)$ maintains this bound, establishing the first claim.

Next, using $\Var[X + Y] = \Var[X] + \Var[Y]$ for independent, and $\Var[X + Y] \leq (\Std[X] + \Std[Y])^2$ by Cauchy--Schwarz for non-independent rvs, we get
\begin{align*}
\VarB[\big]{\DE}
&= \VarB[\pig]{\sum_{\sS' \in \sS} \sum_{c \in \sS'} P(c) \DE^c}\\
&= \sum_{\sS' \in \sS} \VarB[\pig]{ \sum_{c \in \sS'} P(c) \DE^c}\\
&\leq \sum_{\sS' \in \sS} \PP[\Big]{ \sum_{c \in \sS'} P(c) \StdB[\big]{\DE^c}}^2\!.\qedhere
\end{align*}
\end{proof}

\thmbetamomentbasedquantile*
\begin{proof}
Applying \Cref{thm:bounded-estimator-to-upper-bound} to the noise random variable $Z = \DE - \DT$, we obtain the desired bound.
From \Cref{thm:beta-variance-rate} and \Cref{thm:beta-bias-rate}, combined through \Cref{thm:dec-to-bounded-de}, we establish that both $b(n)$ and $s(n)$ are $\bigO(n^{\negone/2})$.
Hence, their positive linear combination also scales as $\bigO(n^{\negone/2})$.
\end{proof}

\newpage
\section{Proofs for simple randomized experiments}
\label{sec:proofs-for-simple-randomized-experiments}

The proofs in this section accompany \Cref{sec:alpha-experiments}.

\thmalphanoisedifference*
\begin{proof}
We express $\EAsub_f$ and $\EAsub_g$ in terms of $E_n$, $L_n$ and $\Sigma_n$ as
\begin{align*}
\EAsub_f
= \frac{\sum_{i=1}^n \indic_{F_i = f} \Epsilonsub_i}{\sum_{i=1}^n \indic_{F_i = f}}
= \frac{(E_n + \Sigma_n)/2}{(n + L_n)/2}
= \frac{E_n + \Sigma_n}{n + L_n},\\
\EAsub_g
= \frac{\sum_{i=1}^n \indic_{F_i = g} \Epsilonsub_i}{\sum_{i=1}^n \indic_{F_i = g}}
= \frac{(E_n - \Sigma_n)/2}{(n - L_n)/2}
= \frac{E_n - \Sigma_n}{n - L_n}.
\end{align*}
Combining $\EAsub_f$ and $\EAsub_g$, we get
\begin{align*}
\EAsub_f - \EAsub_g
= \frac{E_n + \Sigma_n}{n + L_n} - \frac{E_n - \Sigma_n}{n - L_n}.
\end{align*}
Finding the common denominator yields the claim.
\end{proof}

\thmalphanoiserate*
\begin{proof}
We first prove the claim for the case $\ct{F} = 2$.
In this proof, we expediently ignore the condition $\vect{N} > 0$, whose probability tends to $1$ exponentially fast.

We first establish probabilistic rates for $L_n$, $E_n$ and $\Sigma_n$.
\begin{itemize}
\item By the central limit theorem, $L_n = \bigO_p(\sqrt{n})$.
\item $E_n$, a sum of bounded terms, is trivially $\bigO_p(n)$.
\item Since $\E[\DI_i \Epsilonsub_i \given \DIvec_{<i}, \Epsilonvecsub_{<i}] = \E[\DI_i] \E[\Epsilonsub_i \given \DIvec_{<i}, \Epsilonvecsub_{<i}] = 0$, the sum $\Sigma_n$ is a martingale with bounded increments $\abs{\DI_i \Epsilonsub_i} \leq \mc$.
By the Azuma--Hoeffding inequality, $\Sigma_n = \bigO_p(\sqrt{n})$.
\end{itemize}

For the numerator in \Cref{thm:alpha-noise-difference}, a valid asymptotic rate (an upper bound) is the highest of the stochastic orders of the terms $2 n \Sigma_n$ and $2 L_n E_n$.
Using the previously established rates, both terms scale as $n^{3/2}$, so the numerator is $\bigO_p(n^{3/2})$.

Moving on to the denominator, since $L_n = \sum_{i=1}^n \DI_i$ is a sum of bounded i.i.d. variables, the strong law of large numbers guarantees $L_n / n \to 0$ almost surely.
We use concentration inequalities (Hoeffding's or Chernoff's) to show that for any $d \in (0, 1)$,
\begin{align*}
P(\abs{L_n} \geq n \sqrt{1-d}) \to 0 \quad \text{exponentially as } n \to \infty.
\end{align*}
This implies that $P(L_n^2 < (1-d)n^2) \to 1$.
Therefore, the denominator is bounded below by $d n^2$ with probability approaching $1$, establishing the necessary stochastic stability for division:
$n^2 - L_n^2 = \Omega_p(n^2)$.

Finally, combining the numerator, which is bounded in probability at rate $n^{3/2}$, with the denominator, which is bounded away from zero at rate $n^2$, we get that
\begin{align*}
\EAsub_f - \EAsub_g = \frac{\bigO_p(n^{3/2})}{\Omega_p(n^2)} = \bigO_p(n^{3/2 - 2}) = \bigO_p(n^{\negone/2}),
\end{align*}
which completes the proof for $\ct{F} = 2$.

For $\ct{F} > 2$, we consider the subsequence of runs where $F_i \in \{\fg\}$.
Let $M$ be the random count of such runs.
Conditioned on these runs, the experiment is equivalent to the $\ct{F}=2$ case with sample size $M$, so the difference is $\bigO_p(M^{\negone/2})$.
Since $F_i$ is uniform, $M \sim \operatorname{Binomial}(n, 2/\ct{F})$, implying $M = \Theta_p(n)$.
Therefore, the rate $\bigO_p(n^{\negone/2})$ holds generally.
\end{proof}

\thmalphasqrtnconsistency*
\begin{proof}
Once again we ignore the exponentially decaying effect of the $\vect{N} > 0$ condition.
Expanding the definitions of $\DE^c$ and $\delta^c$ and regrouping, we get
\begin{align*}
\DE^c - \delta^c
&= (\EAsub_f - \EAsub_g) + [(\TauAsub_f - l(\fc)) - (\TauAsub_g - l(g, c))].
\end{align*}
From \Cref{thm:alpha-noise-rate}, we already know that the first term is $\bigO_p(n^{\negone/2})$.
As to the second term, we apply \Cref{thm:alpha-noise-rate} a second time, with $\EAsub_f = (\TauAsub_f - \E[\Tau \given F = f])$ and $\EAsub_g = (\TauAsub_g - \E[\Tau \given F = g])$, which have expectation zero.
Since the theorem holds for any noise process, it holds in our case where $\Tausub_i \given F_i = f$ are independent by the properties of delta designs.
Considering that $l(\fc) = \E[\Tau \given F = f] + C_U$ and $l(g, c) = \E[\Tau \given F = g] + C_U$, where $C_U = \E_U[\Epsilon]$, we have that the second term is also $\bigO_p(n^{\neg 1/2})$.
\end{proof}

\thmalphaapproximation*
\begin{proof}
We consider the noise difference
\begin{align*}
\EAsub_f - \EAsub_g = \frac{2 n \Sigma_n - 2 L_n E_n}{n^2 - L_n^2}
\end{align*}
from \Cref{thm:alpha-noise-difference}, where $\Sigma_n$ and $L_n$ are martingales.

First, we approximate the denominator.
Since the $\DI_i$ are i.i.d. Rademacher variables, $L_n$ follows a simple symmetric random walk with $L_n = \bigO_p(n^{1/2})$.
We approximate the denominator as $1/n^2$ with error
\begin{align*}
\frac{1}{n^2 - L_n^2} - \frac{1}{n^2}
= \frac{L_n^2}{n^2(n^2-L_n^2)}
= \bigO_p(n^{\neg 3}).
\end{align*}
Considering that the numerator is $\bigO_p(n^{3/2})$, we approximate the noise difference as
\begin{align*}
\EAsub_f - \EAsub_g
&= \frac{2n \Sigma_n - 2 L_n E_n}{n^2 - L_n^2}\\
&= \frac{2n \Sigma_n - 2 L_n E_n}{n^2} + \bigO_p\bigl(n^{\neg 3/2}\bigr)\\
&= \tilde{M} - \tilde{C} + \bigO_p\bigl(n^{\neg 3/2}\bigr),
\end{align*}
which proves the first claim.
Note that this approximation slightly underestimates the variance for finite $n$; thus, the results derived here apply only for sufficiently large $n$, which we determine empirically.

Next, we calculate $\Var[\tilde{C}]$.
By the independence assumption from the theorem statement and using $\E[L_n] = 0$ as well as $\E[L_n^2] = \E[(\sum_i \DI_i)^2] = \E[\sum_i \DI_i^2] = n$ (since $\E[\DI_i \DI_j] = 0$ for $i \neq j$), we express the variance of the correction term as
\begin{align*}
\Var[\tilde{C}]
&= \E[\tilde{C}^2] - \E[\tilde{C}]^2\\
&= \frac{4}{n^2} \PP[\big]{\E[L_n^2 \bar{\Epsilon}^2] - \E[L_n \bar{\Epsilon}]^2}\\
&= \frac{4}{n^2} \PP[\big]{\E[L_n^2] \E[\bar{\Epsilon}^2] - \E[L_n]^2 \E[\bar{\Epsilon}]^2}\\
&= \frac{4}{n} \E[\bar{\Epsilon}^2].
\end{align*}

Since independence implies $\Cov[\DI_i \Epsilonsub_i, \DI_j \bar{\Epsilon}] = 0$ for $i \neq j$, we rewrite the covariance as
\begin{align*}
\Cov[\tilde{M}, \tilde{C}]
&= \frac{4}{n^2} \Cov[\Sigma_n, L_n \bar{\Epsilon}]\\
&= \frac{4}{n^2} \Cov\BB[\pig]{\sum_i \DI_i \Epsilonsub_i, \pigl(\sum_j \DI_j\pigr) \bar{\Epsilon}}\\
&= \frac{4}{n^2} \sum_i \sum_j \Cov[\DI_i \Epsilonsub_i, \DI_j \bar{\Epsilon}]\\
&= \frac{4}{n^2} \sum_i \Cov[\DI_i \Epsilonsub_i, \DI_i \bar{\Epsilon}]\\
&= \frac{4}{n^2} \sum_i \bigl(\E[\underbrace{\DI_i^2}_{(\pm 1)^2} \Epsilonsub_i \bar{\Epsilon}] - \underbrace{\E[\DI_i \Epsilonsub_i]}_{0} \E[\DI_i \bar{\Epsilon}]\bigr)\\
&= \frac{4}{n^2} \sum_i \E[\Epsilonsub_i \bar{\Epsilon}]\\
&= \frac{4}{n} \E[\bar{\Epsilon}^2].
\end{align*}
Thus, we have proved that $\Cov[\tilde{M}, \tilde{C}] = \Var[\tilde{C}]$.

Substituting these evaluations back into the variance expansion, we have
\begin{align*}
\Var[\EAsub_f - \EAsub_g]
&= \Var\BB[\big]{\tilde{M} - \tilde{C} + \bigO_p\bigl(n^{\neg 3/2}\bigr)}\\
&= \Var[\tilde{M} - \tilde{C}] + \Var\BB[\big]{\bigO_p\bigl(n^{\neg 3/2}\bigr)}\\
&\hphantom{={}} + 2\Cov\BB[\big]{\tilde{M} - \tilde{C}, \bigO_p\bigl(n^{\neg 3/2}\bigr)}\\
&= \Var[\tilde{M} - \tilde{C}] + \bigO\bigl(n^{\neg 3}\bigr)\\
&\hphantom{={}} + 2\Cov\BB[\big]{\bigO_p\bigl(n^{\neg 1/2}\bigr), \bigO_p\bigl(n^{\neg 3/2}\bigr)}\\
&\leq \Var[\tilde{M} - \tilde{C}] + \bigO(n^{\neg 3}) + 2\sqrt{\bigO\bigl(n^\negone\bigr) \bigO\bigl(n^{\neg 3}\bigr)}\\
&= \Var[\tilde{M}] + \Var[\tilde{C}] - 2\Cov[\tilde{M}, \tilde{C}] + \bigO\bigl(n^{\neg 2}\bigr)\\
&= \Var[\tilde{M}] - \Var[\tilde{C}] + \bigO\bigl(n^{\neg 2}\bigr).
\end{align*}
Since $\Var[\tilde{C}] \geq 0$, the third claim follows.
\end{proof}

\approxmartingale*
\begin{proof}[Derivation]
While decreased variance does not strictly imply a tighter exponential tail, motivated by the variance bound $\Var[\EAsub_f - \EAsub_g] \lesssim \Var[\tilde{M}]$ from \Cref{thm:alpha-approximation}, we heuristically posit that the concentration of the noise difference is no worse than that of its martingale approximation $\tilde{M}$.
We therefore bound the success probability using the Azuma--Hoeffding inequality, yielding
\begin{align*}
P(\EAsub_f - \EAsub_g \leq \omega) \gtrsim P(\tilde{M} \leq \omega) \geq 1 - \exp\left(\frac{-n \omega^2}{8 \mc^2}\right).
\end{align*}
\let\qed\relax
\end{proof}

\approxasymmetric*
\begin{proof}[Derivation]
We begin with the simple product bound
\begin{align*}
P(\EAsub_f - \EAsub_g \leq \omega) \geq P(\tilde{M} \leq \omega/2) \; P(\tilde{C} \geq -\omega/2).
\end{align*}
Because it treats the two failure modes as independent, this bound is overly conservative, as we verified against the dynamic programming solution for $n < 200$.
Given the positive covariance between $\tilde{M}$ and $\tilde{C}$ established in \Cref{thm:alpha-approximation}, we might be tempted to write
\begin{align*}
P(\EAsub_f - \EAsub_g \leq \omega) \approx P(\tilde{M} \leq \omega) \; P(\tilde{C} \geq -\omega),
\end{align*}
but this simple form overlooks two important adversarial factors.
First, if the martingale $\frac{2}{n} \Sigma_i > \omega$ at any $i$, the adversary can set the rest of the noise to $0$ to guarantee that the final value $\tilde{M} > \omega$.
Hence, we need the probability that the running maximum of the martingale never exceeds the tolerance.
By the reflection principle and using that the worst-case noise for the martingale is bang--bang, this is
\begin{align*}
&P\bigl(\forall i \colon k \Sigma_i \leq \omega\bigr)
= 2 P(k W_n \leq \omega) - 1\\
&\qquad = 2 P(W_n \leq k^\negone \omega) - 1,
\end{align*}
where $W_n$ is the simple symmetric random walk and $k = 2\mc n^\negone$.

Second, the optimal adversarial noise to push $\tilde{C} = \frac{2}{n} L_n \bar{\Epsilon}$ below $-\omega$ is $\Epsilon_i = -\sign(L_n)\mc$ for all $i$.
Since zero crossings in a random walk are rare by the arcsine law, an adversary can predict the final sign of the imbalance with high probability early in the sequence and commit to a sign for $\bar{\mathcal{E}}$ that maximizes the error.
Thus, we must bound the \emph{magnitude} of the imbalance to ensure the correction term stays above $-\omega$.
Noting that $L_n$ also behaves as a simple symmetric random walk, we derive
\begin{align*}
&P(\tilde{C} \geq -\omega)
\geq P(k \abs{L_n} \leq \omega)\\
&\qquad = P(\abs{W_n} \leq k^\negone \omega)
= 2 P(W_n \leq k^\negone \omega) - 1.
\end{align*}
Leveraging the positive correlation to treat the product as a conservative lower bound, we combine the previously derived bounds and take a normal approximation.
Thus, our asymmetric bound is
\begin{align*}
&P(\EAsub_f - \EAsub_g \leq \omega)\\
&\qquad \geq P(\forall i \colon k \Sigma_i \leq \omega) \; P(\tilde{C} \geq -\omega)\\
&\qquad = \bigl[2 P(W_n \leq k^\negone \omega) - 1\bigr]^2\\
&\qquad \approx \biggl[2 \Phi\biggl(\frac{\sqrt{n}}{2\mc} \omega\biggr) - 1\biggr]^2.
\end{align*}
\let\qed\relax
\end{proof}

\approxforseveralprograms*
\begin{proof}[Derivation]
Under an $\alpha$-design with any $\omega > 0$, provided that the martingale approximation (\Cref{thm:martingale-approximation}) holds,
\begin{align*}
P(\EAsub_f - \EAsub_g \leq \omega) \gtrsim 1 - \bigl(p e^{\neg \ct{K}} + 1 - p\bigr)^n,
\end{align*}
where $\ct{K} = \omega^2/(8\mc^2)$ and $p = 2/\ct{F}$.
Note that with two programs, $p = 1$ and the above recovers the martingale bound.

We prove the proposed bound by expressing the tail probability as the expectation over $N_f + N_g \sim \operatorname{Binomial}(n, p)$ and using the two-program martingale bound, yielding
\begin{align*}
&P(\EAsub_f - \EAsub_g \geq \omega)\\
&\qquad = \sum_{m=0}^n P(N_f + N_g = m) P(\EAsub_f - \EAsub_g \geq \omega \given N_f + N_g = m)\\
&\qquad \lesssim \sum_{m=0}^n \binom{n}{m} p^m (1-p)^{n-m} e^{-m \ct{K}}\\
&\qquad = \sum_{m=0}^n \binom{n}{m} \bigl(p e^{\neg \ct{K}}\bigr)^m (1-p)^{n-m}\\
\explainstep{by the binomial theorem}
&\qquad = \bigl(p e^{\neg \ct{K}} + 1 -p\bigr)^n\noscriptspace.
\end{align*}
\let\qed\relax
\end{proof}

\newpage

\section{Mathematical Preliminaries}

\begin{lemma}[Bounded estimator to stochastic upper bound]
\label{thm:bounded-estimator-to-upper-bound}
Let $\smash{\hat{\theta}}$ be an estimator of $\theta \in \sR$ with bounded absolute bias and standard deviation.
That is, $\abs{\EB{\hat{\theta}} - \theta} \leq b$ and $\StdB{\hat{\theta}} \leq s$.
Then, for all $0 < p < 1$, we have that
\begin{align*}
P\PP[\big]{\hat{\theta} - \theta \leq m} \geq p,
\end{align*}
where $m = s \smash[b]{\sqrt{\frac{p}{1 - p}}} + b$.
\end{lemma}

\begin{proof}
The claim follows from Cantelli's inequality and a bias correction.
For any $k > 0$, the inequality states that
\begin{align*}
P\PP[\pig]{\hat{\theta} - \EB{\hat{\theta}} \geq k} \leq \frac{\VarB{\smash[t]{\hat{\theta}}}}{\VarB{\hat{\theta}} + k^2}.
\end{align*}
Because the function $x \mapsto x / (x + k^2)$ is monotonically increasing for $x > 0$, it then follows that
\begin{align*}
P\PP[\pig]{\hat{\theta} - \EB{\hat{\theta}} \geq k} \leq \frac{s^2}{s^2 + k^2}.
\end{align*}
Taking the complement yields
\begin{align*}
P\PP[\pig]{\hat{\theta} - \EB{\hat{\theta}} < k} \geq 1 - \frac{s^2}{s^2 + k^2}.
\end{align*}
We decompose the estimation error into the centered variable and the bias as
\begin{align*}
\hat{\theta} - \theta = (\hat{\theta} - \EB{\hat{\theta}}) + (\EB{\hat{\theta}} - \theta).
\end{align*}
By assumption, $\EB{\hat{\theta}} - \theta \leq b$, so the event $\hat{\theta} - \EB{\hat{\theta}} < k$ implies the event $\hat{\theta} - \theta < k + b$, which ensures that
\begin{align*}
P\PP[\pig]{\hat{\theta} - \theta < k + b} \geq P\PP[\pig]{\hat{\theta} - \EB{\hat{\theta}} < k}.
\end{align*}
Combining the inequalities and relaxing the strict inequality to a non-strict one yields
\begin{align*}
P\PP[\pig]{\hat{\theta} - \theta \leq k + b} \geq 1 - \frac{s^2}{s^2 + k^2}.
\end{align*}
We set the right-hand side to equal $p$ and solve for $k$, getting
\begin{align*}
k = s \sqrt{\frac{p}{1 - p}}.
\end{align*}
Substituting this $k$ yields the desired form
\begin{align*}
P\PP[\bigg]{\hat{\theta} - \theta \leq s \sqrt{\frac{p}{1 - p}} + b} \geq p,
\end{align*}
which concludes the proof.
\end{proof}

The following lemma basically restates results from \citet{katzenbeisser1986note,katzenbeisser1986alternative} and (2.49) from \citep{godreche2017longest} for the first two uncentered moments of the number of returns to the origin in tied-down simple random walks.
Note that although the cited works state this result for independent rvs, the conditioning on balanced block counts (tying down) renders them only exchangeable, and this is what we require here.

\begin{lemma}[Number of returns in tied-down random walks]
\label{thm:n-returns-in-tied-down-random-walk}
Let $n$ be an even number.
Consider the tied-down random walk $(S_i)_0^n$, a sequence of partial sums of exchangeable rvs taking the values $\pm1$ subject to $S_n = 0$.
Let $Z_n = \abs{\{i \in [1, n] \colon S_i = 0\}}$, the number of returns to the origin.
Then,
\begin{align*}
\E[Z_n \given S_n = 0] = \bigO(\sqrt{n}),\\
\E[Z_n^2 \given S_n = 0] = \bigO(n).
\end{align*}
\end{lemma}

\begin{lemma}[Number of returns in simple random walks]
\label{thm:n-returns-in-simple-random-walk}
Let $n \in \sN$.
Consider the tied-down random walk $(S_i)_0^n$, a sequence of partial sums of independent rvs taking the values $\pm1$.
Let $Z_n = \abs{\{i \in [1, n] \colon S_i = 0\}}$, the number of returns to the origin.
Then,
\begin{align*}
\E[Z_n] = \bigO(\sqrt{n}),\\
\E[Z_n^2] = \bigO(n).
\end{align*}
\end{lemma}

\begin{proof}
For $\E[Z_n] = \bigO(\sqrt{n})$, see e.g. Chapter III in \citet{feller1968introduction}.
As for $\E[Z_n^2] = \bigO(n)$, we can split the random walk at the last return (or at $0$ if it does not exist) and use the result from \Cref{thm:n-returns-in-tied-down-random-walk}.
\end{proof}

\begin{lemma}
\label{thm:some-sum-big-o-n}
$\smash{\sum_{k=0}^{n-1} \sqrt{\frac{k}{n-k}} = \bigO(n)}$.
\end{lemma}

\begin{proof}
We bound the sum by an integral.
Since $f(x) = \sqrt{\frac{x}{n-x}}$ is an increasing function on $[0, n)$, the left Riemann sum is bounded by the integral:
\begin{align*}
\sum_{k=0}^{n-1} \sqrt{\frac{k}{n-k}} \leq \int_0^n \sqrt{\frac{x}{n-x}} \, dx
\end{align*}
We evaluate this integral using the substitution $x = n \sin^2 \theta$ (where $dx = 2n \sin \theta \cos \theta \, d\theta$):
\begin{align*}
\int_0^n \sqrt{\frac{x}{n-x}} \, dx
&= \int_0^{\pi/2} \!\! \sqrt{\frac{n \sin^2 \theta}{n \cos^2 \theta}} \cdot 2n \sin \theta \cos \theta \, d\theta\\
&= 2n \int_0^{\pi/2} \!\!\!\! \sin^2 \theta \, d\theta = 2n \left(\frac{\pi}{4} \right) = \frac{n\pi}{2}.
\end{align*}
\end{proof}

\end{document}